\documentclass[sn-mathphys-num]{sn-jnl}

\usepackage[T1]{fontenc}
\usepackage{graphicx}
\usepackage{amsfonts,amsmath,amssymb,amsthm}
\usepackage[dvipsnames]{xcolor}
\usepackage{verbatim}
\usepackage[ruled,vlined]{algorithm2e}
\usepackage[shortlabels]{enumitem}
\usepackage{float}
\usepackage{fullpage}
\usepackage{tikz}
\usetikzlibrary{positioning}
\usetikzlibrary{arrows}
\usepackage{caption}
\usepackage{indentfirst}

\RequirePackage{zref-clever}
\zcsetup{
        cap=true, % capitalize types
        nameinlink,
}
\zcRefTypeSetup{def}{Name-sg=Definition, Name-pl=Definitions, name-sg=definition, name-pl=definitions}
\zcRefTypeSetup{lem}{
    Name-sg = Lemma, Name-pl=Lemmas, name-sg=lemma, name-pl=lemmas}
\zcRefTypeSetup{thm}{
    Name-sg = Theorem, Name-pl=Theorems, name-sg=theorem, name-pl=theorems}
\zcRefTypeSetup{prop}{
    Name-sg = Proposition, Name-pl=Propositions, name-sg=proposition, name-pl=propositions}
\zcRefTypeSetup{cor}{
    Name-sg = Corollary, Name-pl=Corollaries, name-sg=corollary, name-pl=corollaries}
\zcRefTypeSetup{algocf}{
    Name-sg = Algorithm, Name-pl=Algorithms, name-sg=algorithm, name-pl=algorithms}
\zcRefTypeSetup{claim}{
    Name-sg=Claim, Name-pl=Claims, name-sg=claim, name-pl=claims}
\zcRefTypeSetup{fact}{
    Name-sg=Fact, Name-pl=Facts, name-sg=fact, name-pl=facts}
\zcRefTypeSetup{maintheorem}{
    Name-sg=Main Theorem, Name-pl=Main Theorems, name-sg=main theorem, name-pl=main theorems}
\zcRefTypeSetup{conj}{
    Name-sg=Conjecture, Name-pl=Conjectures, name-sg=conjecture, name-pl=conjectures}
\zcRefTypeSetup{constr}{
    Name-sg=Construction, Name-pl=Constructions, name-sg=construction, name-pl=constructions}
\zcRefTypeSetup{assum}{
    Name-sg=Assumption, Name-pl=Assumptions, name-sg=assumption, name-pl=assumptions}
\zcRefTypeSetup{remark}{
    Name-sg=Remark, Name-pl=Remarks, name-sg=remark, name-pl=remark}

\usepackage[multiuser,inline,nomargin]{fixme}

\newcommand{\cref}[1]{\zcref{#1}}
\newcommand{\Cref}[1]{\zcref[S]{#1}}

\numberwithin{equation}{section}

\RequirePackage{mathtools}

\newtheorem{thm}{Theorem}[section]
\newtheorem{lem}{Lemma}[section]
\newtheorem{definition}{Definition}[section]
\newtheorem{constr}{Construction}[section]
\newtheorem{prop}{Proposition}[section]
\newtheorem{remark}{Remark}[section]
\newtheorem{cor}{Corollary}[section]

\newcommand{\F}{\mathbb{F}}
\newcommand{\N}{\mathbb{N}}

\newcommand{\veca}{\ensuremath{\mathbf{a}}}

\newcommand{\vecc}{\ensuremath{\mathbf{c}}}

\newcommand{\vece}{\ensuremath{\mathbf{e}}}

\newcommand{\vecg}{\ensuremath{\mathbf{g}}}

\newcommand{\vecr}{\ensuremath{\mathbf{r}}}

\newcommand{\vecu}{\ensuremath{\mathbf{u}}}
\newcommand{\vecv}{\ensuremath{\mathbf{v}}}

\newcommand{\veczero}{\ensuremath{\mathbf{0}}}
\newcommand{\vecone}{\ensuremath{\mathbf{1}}}

\newcommand{\cA}{\ensuremath{\mathcal{A}}}
\newcommand{\cB}{\ensuremath{\mathcal{B}}}

\newcommand{\cE}{\ensuremath{\mathcal{E}}}
 
\newcommand{\cG}{\ensuremath{\mathcal{G}}}

\newcommand{\cM}{\ensuremath{\mathcal{M}}}

\newcommand{\cO}{\ensuremath{\mathcal{O}}}

\newcommand{\cS}{\ensuremath{\mathcal{S}}}

\newcommand{\mA}{\ensuremath{\mathbf{A}}}
\newcommand{\mB}{\ensuremath{\mathbf{B}}}

\newcommand{\mD}{\ensuremath{\mathbf{D}}}

\newcommand{\mG}{\ensuremath{\mathbf{G}}}
\newcommand{\mH}{\ensuremath{\mathbf{H}}}
\newcommand{\mI}{\ensuremath{\mathbf{I}}}

\newcommand{\mM}{\ensuremath{\mathbf{M}}}

\newcommand{\mP}{\ensuremath{\mathbf{P}}}

\newcommand{\mS}{\ensuremath{\mathbf{S}}}

\newcommand{\defn}{\vcentcolon=}

\newcommand{\diag}{\text{diag}}

\renewcommand{\cal}[1]{\mathcal{#1}}
\newcommand{\code}{\cal{C}}

\renewcommand{\Im}{\operatorname{Im}}
\DeclareMathOperator{\supp}{supp} % support
\DeclareMathOperator{\Aut}{Aut} % (permutation) automorphism group
\DeclareMathOperator{\MAut}{MAut} % monomial automorphism group
\DeclareMathOperator{\Iso}{Iso} % isomorphism group
\DeclareMathOperator{\Isom}{Isom} % linear isometry group
\DeclareMathOperator{\rank}{rank} % rank
\DeclareMathOperator{\rref}{RREF} % reduced row echelon form
\DeclareMathOperator{\id}{id} % identity map

\let\emptyset\varnothing

\DeclarePairedDelimiter\inner{\langle}{\rangle}
\DeclarePairedDelimiter\bracks{[}{]}
\DeclarePairedDelimiterX\braket[2]{\langle}{\rangle}{#1 \delimsize\vert #2}
\DeclarePairedDelimiter\set{\{}{\}}
\DeclarePairedDelimiter\abs{\lvert}{\rvert}
\DeclarePairedDelimiter\parens{(}{)}

\newcommand{\class}[1]{\ensuremath{\mathsf{#1}}\xspace}

\newcommand{\CE}{\class{CE}}
\newcommand{\GI}{\class{GI}}
\newcommand{\LIP}{\class{LIP}}
\newcommand{\PCE}{\class{PCE}}
\newcommand{\LCE}{\class{LCE}}
\newcommand{\sPCE}{\class{search\text{-}PCE}}
\newcommand{\sLCE}{\class{search\text{-}LCE}}
\newcommand{\AGEN}{\class{AGEN}}
\newcommand{\MAGEN}{\class{MAGEN}}
\newcommand{\APART}{\class{APART}}
\newcommand{\Apart}{\class{APART}}
\newcommand{\MAPART}{\class{MAPART}}
\newcommand{\MApart}{\class{MAPART}}
\newcommand{\MAcount}{\class{MACOUNT}}
\newcommand{\MACOUNT}{\class{MACOUNT}}
\newcommand{\Acount}{\class{ACOUNT}}
\newcommand{\ACOUNT}{\class{ACOUNT}}

\newcommand{\ICOUNT}{\class{ICOUNT}}
\newcommand{\MICOUNT}{\class{MICOUNT}}

\newcommand{\Gpin}{\ensuremath{\mathbf{G}_\text{fixed}}}

\newcommand{\Sym}{\ensuremath{\text{Sym}}}

\title{Code Equivalence and Automorphism Problems for Codes}

\author[1]{\fnm{Jean-Fran\c{c}ois} \sur{Biasse}}\email{biasse@usf.edu}

\author*[1]{\fnm{Alexandra V.} \sur{Hostetler}}\email{avelichehostetler@usf.edu}

\author*[1]{\fnm{Anuvrat} \sur{Jaindungarwal}}\email{anuvrat@usf.edu}

\affil*[1]{\orgdiv{Department of Mathematics and Statistics}, \orgname{University of South Florida},\\ \orgaddress{\street{ 4202 E Fowler Avenue}, \city{Tampa}, \postcode{33620}, \state{FL}, \country{U.S.A.}}}

\abstract{
We study the complexity of the Code Equivalence problem and show that it is polynomially equivalent to several computational automorphism problems for codes.
These problems ask for the cardinality (ACOUNT), an orbit partition (APART), and a generating set (AGEN) for the permutation automorphism group of a code.
We present deterministic, polynomial-time reductions between Permutation Code Equivalence (PCE) and each of these problems, including a one-shot reduction from search-PCE to AGEN that makes a single oracle call.
We present similar reductions between Linear Code Equivalence (LCE) and analogous problems for the monomial automorphism group of a code.
All of our reductions work for any linear codes.
}


\begin{document}

\maketitle

\section{Introduction}%
\label{sec:intro}

A (linear) code $\code \subseteq \F_q^n$ over a finite field $\F_q$ is a linear subspace of $\F_q^n$.
The \emph{Code Equivalence} (\CE) problem asks whether two codes $\code_1$ and $\code_2$ are structurally equivalent in some way.
Variants of this problem specify the type of equivalence:
\emph{Permutation Code Equivalence} (\PCE) asks whether the codes are equivalent up to a permutation of coordinates of their codewords, while
\emph{Linear Code Equivalence} (\LCE) asks whether the codes are equivalent up to a permutation and non-zero scaling of the codeword coordinates.
Search variants of these problems (\sPCE and \sLCE) ask one to find such an equivalence map, if it exists.

% reductions between CE variants
A few prior works have studied the relationship between variants of \CE.
Biasse and Micheli in \cite{BM23} proved a deterministic search-to-decision reduction for \PCE that runs in time polynomial in $n$, the length of the input codes, and makes $n^2$ oracle calls.
This result was extended by \cite{BMPW26}, and independently by \cite{Mazumder26}, to a similar polynomial-time reduction from \sLCE to \LCE.
We restate their results below.

\begin{lem}[adapted from \cite{BM23}]
    \label{lem:search-to-decision-PCE}
    There is a deterministic reduction from search-\PCE on $n$-length codes to \PCE which runs in $O(n^3\log(q))$ time and makes at most $n^2$ oracle calls.
\end{lem}

\begin{lem}[adapted from \cite{BMPW26}]
    \label{lem:search-to-decision-LCE}
    There is a deterministic reduction from search-\LCE on $n$-length codes to \LCE which runs in $O(n^3\log(q))$ time and makes at most $n^2$ oracle calls.
\end{lem}

In \cite{CSV25}, Cheraghchi, Shagrithaya, and Veliche (Hostetler) proved a Karp reduction from \PCE to \LCE that runs in time polynomial in $n$ and polylogarithmic in $q$, the field size.
The recent survey \cite{survey26} details the current state-of-the-art of algorithms for solving variants of \CE as well as the complexity.

%  relationship between isomorphism problems
\PCE and \LCE belong to a family of \emph{isomorphism problems}, a general class of problems that ask whether two objects are structurally the same.
These include the well-known Graph Isomorphism problem (\GI) and Lattice Isomorphism problem (\LIP).
\GI asks whether two graphs have an isomorphism that maps the vertices of one graph to the vertices of the other and preserves the edge structure.
\LIP asks whether two geometric lattices, described by two basis matrices, have an orthogonal linear transformation that maps one lattice to the other.
In \cite{PR06}, Petrank and Roth showed a polynomial-time reduction from \GI to \PCE.
Subsequent reductions from \GI to \CE variants were shown by Kaski and \"{O}sterg\aa rd in \cite{KO06} and Grochow in \cite{Grochow12}.
Bennett and Win in \cite{BW24} later studied the relationship between \GI, \LIP, and variants of \CE and showed a fine-grained dimension-preserving reduction from \GI to \LCE as well as a polynomial-time reduction from \LCE over prime fields to \LIP.

On its own, the \GI problem has received much attention in the last half-century.
Among others, Mathon \cite{Mathon79} showed that the problem of recognizing whether two graphs are isomorphic is just as hard as counting all the isomorphisms between the two graphs.
In particular, he proved that \GI is polynomially equivalent to 
\ACOUNT, which asks for the cardinality of a graph's automorphism group, 
\APART, which asks for a partition of the graph's vertex set into orbits under automorphisms, and
\AGEN, which asks for a generating set of automorphisms that generate the graph's automorphism group.

% cryptographic interest
Besides isomorphism problems being interesting in their own right, they form the foundation of several modern cryptographic schemes.
Variants of \CE have been used as the basis of zero-knowledge protocols \cite{Girault90, SS13} and the LESS identification and digital signature scheme \cite{BarenghiBNPS22,BarenghiBPS21,LESS20}.
\LIP is used as the basis of the HAWK digital signature scheme \cite{DucasW22,HAWK-paper}. 
The use of \CE in particular, as the hardness assumption of a post-quantum digital signature, was the motivation of many recent works on efficient algorithms for solving it \cite{BarenghiBPS23,DAlconzo24,BudroniCDSK24,ChouPS25,BennettBBDLLW25,Nowakowski25,BennettBBDLLW26,BardetBOSS26,BattagliolaMS26,BattagliolaHMMSSW26,BudroniE26,BudroniEFN26,AA26}.
Moreover, cryptosystems based on the hardness of isomorphism problems were presented to the NIST standardization process. 
The LESS scheme based on \CE made it to the second round~\cite{LESS-NIST}, while the HAWK proposal based on \LIP made it to the third round~\cite{HAWK-NIST}.
Given the recent and increasing interest in basing cryptographic schemes on the hardness of code- and lattice-based problems, it is all the more crucial to thoroughly understand the complexity of these computational problems relative to other more well-studied problems.
This will help inform future attempts to construct such cryptographic schemes.
\bigskip

% \subsection{Contributions}
% \label{subsec:intro-contributions}
\textbf{Contributions.}
This work gives polynomial-time reductions between variants of \CE and several automorphism problems for codes.
We define \ACOUNT, \APART, and \AGEN problems for codes analogously to those defined in \cite{Mathon79} for graphs; these ask for the cardinality, an orbit partition, and a generating set of the permutation automorphism group of a given code, respectively.
To match \cite{Mathon79}, we also define the \ICOUNT problem for codes, which asks for the cardinality of the isomorphism group of permutations that maps one code to the other;
this is closely related to \ACOUNT.
We give deterministic polynomial-time reductions that proves that \PCE is equivalent to \ACOUNT, \APART, \AGEN, and \ICOUNT.

In this work, we also prove a \emph{one-shot} deterministic polynomial-time reduction from \sPCE to \AGEN, where the algorithm solves \sPCE by making a single oracle call to \AGEN.
To our knowledge, there is no graph analog of this reduction.

We define analogous problems for the \emph{monomial automorphism} group of a code, which allows non-zero scaling of codeword coordinates in addition to permutation;
these problems are \MACOUNT, \MAPART, \MAGEN, and \MICOUNT.
We prove analogous deterministic polynomial-time reductions for all of the reductions for \PCE, showing that \LCE is equivalent to these monomial automorphism problems.

A few years ago, Grochow mentioned in \cite{Grochow20} the possibility of reducing \PCE to \AGEN for codes that cannot be decomposed into a direct sum of other codes.
This, together with the analogous result for \LCE, was stated informally in \cite{survey26}.
Our work provides formal proofs of these results with pseudocode for the algorithms and tight analysis of the runtime and number of oracle calls.
Moreover, our reductions work for \emph{any} linear codes.
\bigskip

\newlength{\mywidth}
\setlength{\mywidth}{3cm}

\begin{figure}[ht]
    \centering
    \begin{tikzpicture}%[scale=\tikzscale]
    
        \node (PCE) at (-\mywidth,0) {\PCE};
        \node (sPCE) [xshift=-\mywidth] at (canvas polar cs:radius=2cm,angle=90) {search-\PCE};
        \node (AGEN) [xshift=-\mywidth] at (canvas polar cs:radius=2.2cm,angle=18) {\AGEN};
        \node (APART) [xshift=-\mywidth] at (canvas polar cs:radius=2.4cm,angle=162) {\APART};
        \node (ICOUNT) [xshift=-\mywidth] at (canvas polar cs:radius=2cm,angle=234) {\ICOUNT};
        \node (ACOUNT) [xshift=-\mywidth] at (canvas polar cs:radius=2cm,angle=306) {\ACOUNT};

        \draw[-latex,cyan] (sPCE) to[bend left=10] node[right,rotate=0] {} (AGEN);
        \draw[-latex,cyan] (AGEN) to[bend left=10] node[right,rotate=0] {} (PCE);
        
        \draw[-latex] (sPCE) to[bend right=10] (PCE); % node[left,rotate=0]  {{\tiny \cite{BM23}}} (PCE);
        \draw[-latex] (PCE) to[bend right=10] node[right,rotate=0] {} (sPCE);

        \draw[-latex,cyan] (ACOUNT) to[bend right=10] node[left,rotate=0] {} (PCE);
        \draw[-latex,cyan] (PCE) to[bend right=10] node[right,rotate=0] {} (ACOUNT);

        \draw[-latex,cyan] (APART) to[bend right=10] node[left,rotate=0] {} (PCE);
        \draw[-latex,cyan] (PCE) to[bend right=10] node[right,rotate=0] {} (APART);

        \draw[-latex,cyan] (ICOUNT) to[bend right=10] node[left,rotate=0] {} (ACOUNT);
        \draw[-latex,cyan] (PCE) to[bend right=10] node[right,rotate=0] {} (ICOUNT);
        
    % \end{tikzpicture}
    % \hspace{1cm}
    % \begin{tikzpicture}%[scale=\tikzscale]
    
        \node (LCE) at (\mywidth,0) {\LCE};
        \node (sLCE) [xshift=\mywidth] at (canvas polar cs:radius=2cm,angle=90) {search-\LCE};
        \node (MAGEN) [xshift=\mywidth] at (canvas polar cs:radius=2.2cm,angle=18) {\MAGEN};
        \node (MAPART) [xshift=\mywidth] at (canvas polar cs:radius=2.4cm,angle=162) {\MAPART};
        \node (MICOUNT) [xshift=\mywidth] at (canvas polar cs:radius=2cm,angle=230) % 234 
        {\MICOUNT};
        \node (MACOUNT) [xshift=\mywidth] at (canvas polar cs:radius=2cm,angle=306) {\MACOUNT};

        \draw[-latex,cyan] (sLCE) to[bend left=10] node[right,rotate=0] {} (MAGEN);
        \draw[-latex,cyan] (MAGEN) to[bend left=10] node[right,rotate=0] {} (LCE);
        
        \draw[-latex] (sLCE) to[bend right=10] (LCE); % node[left,rotate=0]  {{\tiny \cite{BMPW26}}} (LCE);
        \draw[-latex] (LCE) to[bend right=10] node[right,rotate=0] {} (sLCE);

        \draw[-latex,cyan] (MACOUNT) to[bend right=10] node[left,rotate=0] {} (LCE);
        \draw[-latex,cyan] (LCE) to[bend right=10] node[right,rotate=0] {} (MACOUNT);

        \draw[-latex,cyan] (MAPART) to[bend right=10] node[left,rotate=0] {} (LCE);
        \draw[-latex,cyan] (LCE) to[bend right=10] node[right,rotate=0] {} (MAPART);

        \draw[-latex,cyan] (MICOUNT) to[bend right=10] node[left,rotate=0] {} (MACOUNT);
        \draw[-latex,cyan] (LCE) to[bend right=10] node[right,rotate=0] {} (MICOUNT);

        \draw[-latex,dashed] (PCE) to (LCE); % to[bend right=5] node[above,rotate=0]  {{\tiny \cite{CSV25}}} (LCE);
        % \draw[-latex] (LCE) to[bend right=5] (PCE);        
    \end{tikzpicture}
    \caption{
        A map of reductions between Permutation Code Equivalence (\PCE), Linear Code Equivalence (\LCE) and computational code problems relating to automorphisms.
        The reductions proved in this work are represented by the {\color{cyan} blue} arrows, and the remaining arrows represent reductions shown in prior works.
    }
    \label{fig:reduction-map}
\end{figure}
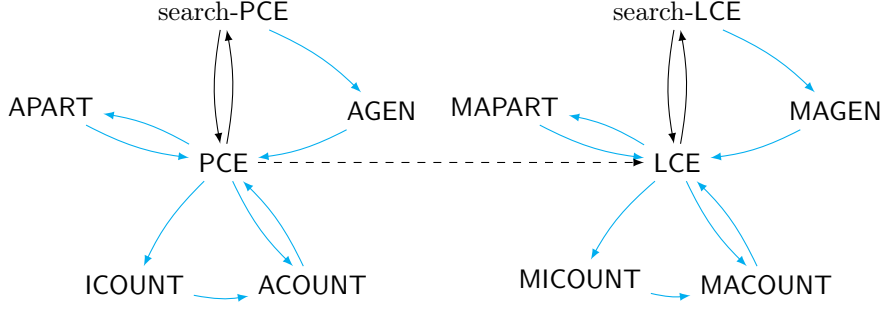

% \subsection{Technical Overview}
% \label{subsec:intro-overview}
\textbf{Technical Overview.}
The isomorphism and automorphism problems for codes are defined formally in \cref{sec:prelims}.
We present auxiliary lemmas that describe the relationship between code isomorphisms and automorphisms, and related properties.

Our reductions rely on the inherent structure of the given code(s), namely its decomposability into a direct sum of other codes and the properties of the generator matrix used to describe the code.
In \cref{sec:code-decomp}, we present a framework for code decomposition and show that any code can be uniquely decomposed into a direct sum of indecomposable codes up to a permutation of the summands and (monomial) equivalence of each code in the sum.
We prove several useful properties of indecomposable codes in \cref{subsec:properties-decomp} and provide an efficient algorithm for finding the decomposition of any code in \cref{subsec:how-to-decompose}.

In \cref{subsec:setwise-stab-chain}, we present an algebraic structure that is a key component in several of our reduction algorithms.
Given a generator matrix of a code, its columns can be partitioned into redundancy classes containing column vectors that are the same, either exactly or up to scaling.
This defines a \emph{setwise stabilizer chain} of the automorphism group of a code.
An important property of this chain is that the automorphism group can be described entirely by the coset representatives of consecutive pairs of subgroups in the chain.

\cref{sec:APART} describes our reductions from \PCE to \APART and in reverse, from \APART to \PCE.
The analogous reductions between \LCE and \MAPART are also presented herein.
Similarly, \cref{sec:ACOUNT} describes our reductions from \PCE to \ACOUNT and vice versa, together with the analogous reductions between \LCE and \MACOUNT.
\cref{sec:AGEN} describes a reduction from \sPCE to \AGEN and from \AGEN to \PCE, as well as the analogous reductions from \sLCE to \MAGEN and from \MAGEN to \LCE.
Both \cref{sec:ACOUNT,sec:AGEN} make use of the reduction algorithms for \APART.

In this work, all codes are linear and all generator matrices are implicitly full-rank.
All our reductions are deterministic and polynomial in the size of the input, unless otherwise stated.

%%% Local Variables:
%%% mode: latex
%%% TeX-master: "main"
%%% End:

\section{Preliminaries}%
\label{sec:prelims}

For any positive integer $n$, we denote $[n] \defn \set{1,\ldots, n}$ and let $\cS_n$ be the symmetric group consisting of permutations on $[n]$.
For any subset $T \subset \cS_n$, we denote its symmetric group by $\Sym(T)$; we extend this notation in the natural way to vectors indexed by $T$.

For any prime power $q$, let $\F_q$ denote the $q$-ary finite field.
All codes considered in this work are linear, so we simply call them codes throughout.
For any code $\code \subseteq \F_q^n$ of length $n$ and dimension $k \leq n$, we say that $\code$ is a \emph{$[n,k]_q$-code}.
For any matrix $\mG\in \F_q^{k\times n}$, we use $\code(\mG) \subseteq \F_q^n$ to denote the code it generates.
We assume that any generator matrix $\mG$ of a code has full row rank. 
We call the set of indices of the coordinates of $\code$ its \emph{index set} and denote this by $X \defn [n]$.

The \emph{support} of a vector $\vecc \in \F_q^n$ is the set of indices containing non-zero coordinates of $\vecc$, i.e.
$\supp(\vecc) \defn \set{ i\in [n] \mid c_i \neq 0}$.
We extend this notation to sets in the natural way: for any $\code \subseteq \F_q^n$, we write $\supp(\code) \defn \bigcup_{\vecc \in \code} \supp(\vecc) \subseteq [n]$.
Applying Gauss-Jordan elimination to a matrix $\mG$ yields a matrix in reduced row echelon form (RREF), which we denote by $\rref(\mG) = \mG_{\rref}$.
The \emph{pivot} of a row in this RREF matrix refers to the leading entry of 1 and is the unique non-zero entry in its column.

We will use the notion of a \emph{support graph of a code}, which contains an edge for every pair of non-zero entries in a row of the generator matrix in RREF.

\begin{constr}[Support Graph]
 \label{constr:supp-graph}
    Given a generator matrix $\mG$ of a $[n,k]$-code $\code$, let $\mG_{\rref}$ be its corresponding RREF matrix.
    The \emph{support graph} $\Gamma(\mG) = (V,E)$ is the graph with $n$ vertices $V \defn [n]$ which contains an edge $(i,j) \in E$ if and only if there is a row $\vecr = (r_1,\ldots,r_n)$ in $\mG_{\rref}$ that has non-zero entries $r_i \neq 0$ and $r_j \neq 0$.
\end{constr}

Note that the support graph of a code is unique.   
This follows immediately from the fact that a linear code has a uniquely determined generator matrix in reduced row echelon form (RREF).
Hence, we use $\Gamma(\code)$ to denote the unique support graph of a code $\code$.

For any $[n,k]_q$-code $\code$, constructing its support graph from its $k\times n$ generator matrix takes $O(kn^2 \log(q))$ time.
This is because adding all the edges requires iterating through all $k$ rows, each of which can contain up to $n$ non-zero entries in $\F_q$.

\begin{definition}[Direct sum of codes]
    \label{def:direct-sum}
     The \emph{direct sum} of two codes $\code_1, \code_2$ generated by matrices $\mG_1 \in \F_q^{k_1\times n_1}$ and $\mG_2 \in \F_q^{k_2\times n_2}$, respectively, is the code $\code \defn \code_1 \oplus \code_2 \subseteq \F_q^n$ generated by the matrix
    \[
        \mG \defn
        \begin{bmatrix}
            \mG_1 & \veczero \\ \veczero & \mG_2
        \end{bmatrix}
        \in \F_q^{k\times n},
    \]
    where $n \defn n_1 + n_2$ and $k \defn k_1 + k_2$.
    $\code$ contains all vectors of the form $\bracks{\vecc_1\ \vecc_2}$, where $\vecc_1 \in \code_1$ and $\vecc_2 \in \code_2$.
\end{definition}

Throughout this paper, we only take the direct sum of two codes in the context of solving \PCE or \LCE.
For this use, we can assume without loss of generality that the generator matrices $\mG_1$ and $\mG_2$ of the input codes do not contain all-zero columns; this follows from Lemma 3.4 in \cite{CSV25}.
For any such direct sum code $\code = \code_1 \oplus \code_2 \subseteq \F_q^n$, we define $X_1$ and $X_2$ to be the disjoint supports of the subcodes $\code_1$ and $\code_2$, respectively, i.e. $X_i \defn \cup_{\vecc \in \code_i} \supp(\vecc) \subseteq [n]$. 
We call $X_1$ and $X_2$ the \emph{index sets of $\code_1$ and $\code_2$ embedded in $X$}.
By definition of direct sum, these have cardinality $\abs{X_i} = n_i$.
Denoting the projection of $\code$ onto the set $X_i$ by $\code|_{X_i}$, we can write $\code = \code|_{X_1} \oplus \code|_{X_2}$.
We sometimes abuse notation and use $\code_i$ to denote the embedded code $\set{\veczero} + \dots + \code_i + \dots + \set{\veczero}$ in a direct sum $\code_1 \oplus \dots \oplus \code_i \oplus \dots \oplus \code_m$.

\begin{definition}[Bridging Element]
 \label{def:bridge}
    Let $\code$ be a code whose index set $X$ has a non-trivial partition $X = Y \cup Z$. 
    A codeword $\vecc \in \code$ is a \emph{bridging element} with respect to $(Y,Z)$ if $\supp(\vecc) \cap Y \neq \emptyset \neq \supp(\vecc) \cap Z$ and there is no pair $\vecu, \vecv \in \code$ such that $\vecc = \vecu + \vecv$ and $\supp(\vecu) \subseteq Y$ and $\supp(\vecv) \subseteq Z$. 
\end{definition}

For example, consider the binary repetition code $\code = \set{00, 11}$ and the partition $Y=\set{1}$ and $Z=\set{2}$, the codeword $11$ is a bridging element since the code does not have the elements $10$ and $01$ required to decompose it across the partition.

% ===========================================
\subsection{Computational Code Problems}
\label{subsec:problem-defs}

We now define the computational problems for codes that are the focus of study in this paper.

% \paragraph{Isomorphism Problems.}
Two codes $\code_1$ and $\code_2$ are called \emph{permutation equivalent} if the coordinates of $\code_1$ can be permuted to obtain $\code_2$. 
They are said to be \emph{isometric} if there exists an $\F_q$-linear isomorphism between them that strictly preserves the Hamming weight of each codeword.
We formalize and generalize these notions in the following two problems.

\begin{definition}[PCE]
    For $n,k \in \N $ and finite field $\F_q$, the \emph{Permutation Code Equivalence (\PCE)} problem is the following decision problem: 
    Given a pair of generator matrices $\mG_1, \mG_2 \in \F_q^{k\times n}$ for codes $\code_1 = \code(\mG_1), \code_2 = \code(\mG_2) \subseteq \F_q^n$, decide whether there exists a permutation isomorphism $\pi \in \cS_n$ such that $\pi(\code_1) = \code_2$.
\end{definition}

\begin{definition}[LCE]
    For $n,k \in \N $ and finite field $\F_q$, the \emph{Linear Code Equivalence (\LCE)} problem is the following decision problem: 
    Given a pair of generator matrices $\mG_1, \mG_2 \in \F_q^{k\times n}$ for codes $\code_1 = \code(\mG_1), \code_2 = \code(\mG_2) \subseteq \F_q^n$, decide whether there exists a linear isometry (a.k.a. monomial isomorphism) $\mu = (\vecv, \pi) \in (\F_q^*)^n \rtimes \cS_n$ such that $\mu(\code_1) = \code_2$.
\end{definition}

When two codes $\code_1, \code_2$ are permutation equivalent, their generator matrices $\mG_1,\mG_2$ form a \PCE instance; we denote this by $\code_1 \cong \code_2$ and say simply that they are \emph{equivalent}.
If they are \emph{monomially equivalent}, they form an \LCE instance, and we denote this by $\code_1 \cong_{\cM} \code_2$.
We will use the following observation about the column frequency of generator matrices of equivalent codes.

\begin{lem}[Lemma~3.3 , \cite{CSV25}]
    \label{lem:col-frequency}
    Let $k$ and $n$ be positive integers and $\F_q$ be a finite field.
    If any two matrices $\mG, \mH \in \F_q^{k\times n}$ form a \PCE instance, then no column of $\mG$ appears in $\mG$ more times than a column of $\mH$ appears in $\mH$.
\end{lem}

\begin{lem}[Remark 1 , search to decision LCE]
    \label{lem:col-frequency-LCE}
    Let $k$ and $n$ be positive integers and $\F_q$ be a finite field.
    If any two matrices $\mG, \mH \in \F_q^{k\times n}$ form an \LCE instance, then the columns of $\mG$ have the same frequency as the columns of $\mH$, up to non-zero scaling.
\end{lem}
 
The following lemma states that isometry and monomial equivalence are equivalent notions.
We thus use both terms interchangeably throughout the paper.

\begin{lem}[\cite{BGG78}, Corollary 1]
    \label{lem:ext_McWilliams}
    Any two codes $\code_1$ and $\code_2$ are monomially equivalent ($\code_1 \cong_{\cM} \code_2$) if and only if $\code_1$ is isometric to $\code_2$.
\end{lem}
Recall that \emph{isometric} here refers to the existence of an $\F_q$-linear isometry.

We will also discuss the search variants of \PCE and \LCE, which ask to find the isomorphism or linear isometry between a given pair of codes, if it exists.

\begin{definition}[Search-\PCE]
    For $n,k \in \N$ and finite field $\F_q$, \emph{search-\PCE} is the following search problem: 
    Given a pair of generator matrices $\mG_1, \mG_2 \in \F_q^{k \times n}$ for codes $\code_1 = \code(\mG_1), \code_2 = \code(\mG_2) \subseteq \F_q^n$, find a permutation $\pi \in \cS_n$ for which $\pi(\code_1) = \code_2$, if it exists.
\end{definition}

\begin{definition}[Search-\LCE]
    For $n,k \in \N$ and finite field $\F_q$, \emph{search-\LCE} is the following search problem: 
    Given a pair of generator matrices $\mG_1, \mG_2 \in \F_q^{k \times n}$ for codes $\code_1 = \code(\mG_1), \code_2 = \code(\mG_2) \subseteq \F_q^n$, find a monomial permutation $\mu \in (\F_q^*)^n \rtimes \cS_n$ for which $\mu(\code_1) = \code_2$, if it exists.
\end{definition}

Note that permutation equivalence implies monomial equivalence.
It is also natural to ask for information about the set of isomorphisms or linear isometries between a given pair of codes.
For any pair of codes $\code_1,\code_2 \subseteq \F_q^n$, we denote the set of all permutation isomorphisms between them by
\begin{equation}
    \Iso(\code_1,\code_2)
    = \set{\pi \in \cS_n \mid \pi(\code_1) = \code_2}.
\end{equation}
Similarly, we denote the set of all linear isometries between them by
\begin{equation}
    \Isom(\code_1,\code_2)
    = \set{\mu \in (\F_q^*)^n \rtimes \cS_n \mid \mu(\code_1) = \code_2}.
\end{equation}
It is immediate that $\Iso(\code_1,\code_2) \subseteq \Isom(\code_1,\code_2)$.
We are interested in computing the cardinalities of these sets, as defined in the following problems.

\begin{definition}[\ICOUNT, \MICOUNT]
    \label{def:(M)ICOUNT}
    For $n,k \in \N$ and finite field $\F_q$, the \emph{Isomorphism Counting (\ICOUNT)} (resp. \emph{Monomial Isomorphism Counting (\MICOUNT)}) problem is the following search problem: 
    Given two generator matrices $\mG_1, \mG_2 \in \F_q^{k \times n}$ for codes $\code_1, \code_2 \subseteq \F_q^n$, compute the cardinality $\abs{\Iso(\code_1, \code_2)}$ (resp. $\abs{\Isom(\code_1, \code_2)}$) of the set of (resp. monomial) isomorphisms between $\code_1$ and $\code_2$. 
    If $\code_1 \not\cong \code_2$ (resp. $\code_1 \not\cong_\cM \code_2$), output $0$. 
\end{definition}

% \paragraph{Automorphism Problems.}
Any permutation on the coordinates of a code $\code \subseteq \F_q^n$ induces a map on its codewords; 
for this reason, we use $\pi(\code)$ to denote the code obtained by applying the permutation $\pi \in \cS_n$ to $\code$.
We denote the \emph{group of permutation automorphisms} of $\code$ by
\begin{equation}
    \label{def:Aut(C)}
    \Aut(\code) \defn \set{\pi \in \cS_n \mid \pi(\code) = \code} \subseteq \cS_n.
\end{equation} 
This group corresponds to the group of permutation matrices $\mP$ that satisfy $\mS \mG \mP = \mG$ for any generator $\mG$ of the code and some invertible change-of-basis matrix $\mS$.
Similarly, we denote the \emph{group of monomial automorphisms} of $\code$ by
\begin{equation}
    \label{def:MAut(C)}
    \MAut(\code) \defn \set{\mu \in (\F_q^*)^n \rtimes \cS_n \mid \mu(\code) = \code} 
    \subseteq (\F_q^*)^n \rtimes \cS_n.
\end{equation}
Note that $\Aut(\code) \subseteq \MAut(\code)$.
Any monomial transformation $(\vecv_\mu,\mu) \in \MAut(\code)$ acts on the ambient space $\F_q^n$. 
So, for any subcode $\cA \subseteq \code$, the set $\mu(\cA)$ is the image of the subspace $\cA$ under this transformation, i.e. $\mu(\cA) \defn \set{ \mu(\veca) \mid \veca \in \cA} \subseteq \mu(\code) \subseteq \F_q^n$. 
The group of monomial automorphisms $\MAut(\code)$ corresponds to the group of monomial matrices $\mM$ that satisfy $\mS \mG \mM = \mG$ for any generator $\mG$ of $\code$ and for some invertible change-of-basis matrix $\mS$.

Isomorphisms between two codes are closely related to automorphisms of their direct sum code.
In the following lemma, we show that the set of all isomorphisms between two codes is a coset of the automorphism group of each code.

\begin{lem}[folklore]
\label{lem:iso=aut}
    For any pair of equivalent codes $\code_1, \code_2 \subseteq \F_q^n$, let $\Iso(\code_1, \code_2)$ be the set of all isomorphisms from $\code_1$ to $\code_2$. 
    For any fixed isomorphism $\pi \in \Iso(\code_1, \code_2)$,
    \begin{equation*}
        \Iso(\code_1, \code_2) 
        = \Aut(\code_2) \circ \pi 
        = \set{\alpha \circ \pi \mid \alpha \in \Aut(\code_2)}.
    \end{equation*}
    Furthermore, $\abs{\Iso(\code_1, \code_2)} = \abs{\Aut(\code_2)}$.
\end{lem}

% \begin{proof}
%     Since $\code_1 \cong \code_2$, there exists at least one isomorphism $\pi \in \Iso(\code_1, \code_2)$.
%     % We show both inclusions of the claimed equality.
%     % 
%     First, we show that $\Aut(\code_2) \circ \pi \subseteq \Iso(\code_1, \code_2)$. 
%     Consider any automorphism $\alpha \in \Aut(\code_2)$. 
%     Since $\pi(\code_1) = \code_2$ and $\alpha(\code_2) = \code_2$, the composition of $\pi$ and $\alpha$ satisfies $\alpha \circ \pi (\code_1) = \code_2$.
%     Hence, $\alpha \circ \pi \in \Iso(\code_1, \code_2)$.
%     % 
%     Now, we show that $\Iso(\code_1, \code_2) \subseteq \Aut(\code_2) \circ \pi$. 
%     Consider any isomorphism $\psi \in \Iso(\code_1, \code_2)$. 
%     % and the corresponding map $\beta \defn \psi \circ \pi^{-1}$. 
%     Since $\pi^{-1}$ is an isomorphism from $\code_2$ to $\code_1$ and $\psi(\code_1) = \code_2$, their composition satisfies $\psi \circ \pi^{-1} (\code_2) = \code_2$. 
%     Hence, $\psi \circ \pi^{-1} \in \Aut(\code_2)$ and $\psi = \psi \circ \pi^{-1} \circ \pi \in \Aut(\code_2) \circ \pi$.
%     % 
%     Therefore, $\Iso(\code_1, \code_2) = \Aut(\code_2) \circ \pi$. 
%     Moreover, by the properties of group actions, we obtain that $\abs{\Iso(\code_1, \code_2)} = \abs{\Aut(\code_2)}$.
% \end{proof}

\begin{lem}
    \label{lem:isom=maut}
    For any pair of monomially equivalent codes $\code_1, \code_2 \subseteq \F_q^n$, let $\Isom(\code_1, \code_2)$ be the set of all linear isometries from $\code_1$ to $\code_2$. 
    For any fixed linear isometry $\mu \in \Isom(\code_1, \code_2)$,
    \begin{equation*}
        \Isom(\code_1, \code_2) 
        = \MAut(\code_2) \circ \mu 
        = \set{\alpha \circ \mu \mid \alpha \in \MAut(\code_2)}.
    \end{equation*}
    Furthermore, $\abs{\Isom(\code_1, \code_2)} = \abs{\MAut(\code_2)}$.
\end{lem}

\begin{proof}
    Since $\code_1 \cong_\cM \code_2$, there exists at least one isomorphism $\pi \in \Isom(\code_1, \code_2)$.
    First, we show that $\MAut(\code_2) \circ \pi \subseteq \Isom(\code_1, \code_2)$. 
    Consider any monomial automorphism $\alpha \in \MAut(\code_2)$. 
    Since $\pi(\code_1) = \code_2$ and $\alpha(\code_2) = \code_2$, the composition of $\pi$ and $\alpha$ satisfies $\alpha \circ \pi (\code_1) = \code_2$.
    Hence, $\alpha \circ \pi \in \Isom(\code_1, \code_2)$.
    Now, we show that $\Isom(\code_1, \code_2) \subseteq \MAut(\code_2) \circ \pi$. 
    Consider any isomorphism $\psi \in \Isom(\code_1, \code_2)$. 
    Since $\pi^{-1}$ is a monomial isomorphism from $\code_2$ to $\code_1$ and $\psi(\code_1) = \code_2$, their composition satisfies $\psi \circ \pi^{-1} (\code_2) = \code_2$. 
    Hence, $\psi \circ \pi^{-1} \in \MAut(\code_2)$ and $\psi = \psi \circ \pi^{-1} \circ \pi \in \MAut(\code_2) \circ \pi$.
    Therefore, $\Isom(\code_1, \code_2) = \MAut(\code_2) \circ \pi$. 
    Moreover, by the properties of group actions, we obtain that $\abs{\Isom(\code_1, \code_2)} = \abs{\MAut(\code_2)}$.
\end{proof}

The following automorphism problems were presented in \cite{Mathon79} for graphs; we define their code counterparts here.
These problems ask for various descriptors of the automorphism group of a given code, such as its cardinality, a set of generators, and an orbit partition.

For any coordinate $i\in [n]$ of code $\code \subseteq \F_q^n$, we denote its \emph{orbit} under the automorphism group $\Aut(\code)$ by $O_i \defn \set{j\in [n] \mid \alpha(i) = j \text{ for some } \alpha \in \Aut(\code)}$.
Note that trivially $i\in O_i$.
We use a similar notation for \emph{monomial orbits} of coordinates under the monomial automorphism group $\MAut(\code)$.

\begin{definition}[\APART, \MAPART]
    For $n,k \in \N$ and finite field $\F_q$, the \emph{Automorphism Partition (\APART)} (resp. \emph{Monomial Automorphism Partition (\MAPART)}) problem is the following search problem: 
    Given a generator matrix $\mG \in \F_q^{k \times n}$ of a code $\code$, find the partition of its index set $[n]$ into (resp. monomial) orbits $O_1,\ldots,O_n \subseteq [n]$,
    where $O_i = O_j$ if and only if there exists $\alpha \in \Aut(\code)$ (resp. $\alpha = (\vecv_\alpha,\pi_\alpha)\in \MAut(\code)$) such that $\alpha(i) = j$ (resp. $\pi_\alpha(i) = j$).
\end{definition}

\begin{remark}
    The monomial orbit of an index is defined only with respect to the permutation component of any monomial automorphism in $\MAut(\code)$.
    Defining orbits with respect to index-scalar pairs is superfluous, because for any coordinate permutation, scaling every coordinate by the same scalar gives a monomial automorphism.
    To see why, consider any generator matrix $\mG$.
    Multiplication by any non-zero scalar $\lambda \in \F_q^*$ gives the identity $\mD \mG = \mG (\lambda \mI_n)$.
    Isolating $\mG$ on the right yields
    $(\lambda^{-1} \mI_k) \mG (\lambda \mI_n) = \mG$.
    The diagonal matrices $\mS = \lambda^{-1} \mI_k$ and $\mM = \lambda \mI_n$ satisfy the monomial equivalence condition $\mS \mG \mM = \mG$.
    We can decompose $\mM$ into the permutation matrix $\mP = \mI_n$ and diagonal matrix $\mD = \diag(\lambda, \ldots, \lambda)$.
    The associated automorphism maps every coordinate $i$ to itself, since $\mP_{i,i} = 1$, and multiplies each $i$-th column of $\mG$ by $\lambda$, since $\mD_{i,i} = \lambda$. 
    Therefore, any index-scalar pair $(i, 1)$ would have the same orbit as $(i, \lambda)$ for all $\lambda \in \F_q^*$.
\end{remark}

\begin{definition}[\ACOUNT, \MACOUNT]
    For $n,k \in \N$ and finite field $\F_q$, the \emph{Automorphism Counting (\ACOUNT)} (resp. \emph{Monomial Automorphism Counting (\MACOUNT)}) problem is the following computational problem: 
    Given a generator matrix $\mG \in \F_q^{k \times n}$ of a code $\code = \code(\mG)$, compute the order of its permutation automorphism group $\abs{\Aut(\code)}$ (resp. monomial automorphism group $\abs{\MAut(\code)}$).
\end{definition}

\begin{definition}[\AGEN, \MAGEN]
    For $n,k \in \N$ and finite field $\F_q$, the \emph{Automorphism Generation (\AGEN)} (resp. \emph{Monomial Automorphism Generation (\MAGEN)}) problem is the following search problem: 
    Given a generator matrix $\mG \in \F_q^{k \times n}$ of a code $\code$, find a set of permutations (resp. linear isometries) $S$ that generates the group $\inner{S} = \Aut(\code)$ (resp. $\inner{S} = \MAut(\code)$).
\end{definition}

It is natural to define similar computational problems relating to isomorphisms of a given pair of codes.
The \ICOUNT problem was defined in \cite{Mathon79} as the analog of \ACOUNT.
% For completeness, we discuss the natural analogs of \AGEN and \APART here.
The isomorphism analogs of \AGEN and \APART arise as a direct consequence of \cref{lem:iso=aut}.

Suppose we want to find the set of all isomorphisms between two equivalent codes $\code_1$ and $\code_2$.
If we can find a generating set for the automorphism group $\Aut(\code_2)$, it is enough to find one isomorphism $\pi$ from $\code_1$ to $\code_2$ to obtain the set of isomorphisms $\Iso(\code_1,\code_2)$.
In this way, this problem of finding all isomorphisms between two equivalent codes trivially reduces to solving \AGEN and \sPCE.

Similarly, if we want to find the orbits of the indices of a code $\code_1$ under any isomorphism from $\code_1$ to $\code_2$, it is enough to find the orbit partition of the index set of $\code_1$ and an isomorphism $\pi$ from $\code_1$ to $\code_2$.
In this way, this problem of finding an orbit partition of the index set of two equivalent codes under isomorphisms trivially reduces to solving \APART and \sPCE.

\subsection{Stabilizer Chains and Generating Sets}
\label{subsec:stab-chain}

Sims introduced in \cite{Sims70} the notion of a stabilizer chain determined by a base as a data structure for efficiently manipulating permutation groups.
We define similar structures in \cref{subsec:setwise-stab-chain} and use them for our reductions in \cref{subsec:ACOUNT->PCE,subsec:AGEN->PCE}.
For reference, we include the definitions and some key properties of the orginal notion here.

\begin{definition}[Base, Pointwise Stabilizer chain]
    \label{def:stabilizer-chain}
    For any positive integer $n$ and permutation group $G \leq \cS_n$, a sequence of elements $B = \set{\beta_1,\ldots,\beta_s} \subseteq [n]$ is called a \emph{base} if the only element of $G$ that fixes $B$ is the identity map.
    The \emph{(pointwise) stabilizer chain} for $G$ determined by base $B$ is the chain of nested subgroups
    \[
       G =\vcentcolon G^{(0)} \geq G^{(1)} \geq \ldots \geq G^{(s)} = \inner{\id},
    \]
    where $G^{(i)} \defn G_{\beta_1,\ldots,\beta_i}$ is the subgroup that stabilizes $\set{\beta_1,\ldots,\beta_i}$ pointwise, for each $i \in [s]$.
\end{definition}

\begin{lem}[adapted from \cite{Seress03}]
    \label{lem:stab-chain-order}
    For any $n\in \N$ and permutation group $G \leq \cS_n$ with stabilizer chain $G =\vcentcolon G^{(0)} \geq G^{(1)} \geq \ldots \geq G^{(s)} = \inner{\id}$, 
    \[
        \abs{G} = \prod_{i=0}^{s-1} \abs{G^{(i)} : G^{(i+1)}}.
    \]
\end{lem}
\vspace{-7pt}

\begin{proof}
    By Lagrange's Theorem, the order of each subgroup $G^{(i)}$ can be written as $\abs{G^{(i)}} = \abs{G^{(i)} : G^{(i+1)}} \cdot \abs{G^{(i+1)}}$. 
    Rewriting each term in the product and canceling like terms gives
    the claimed equality.
    % \begin{equation*}
    %     \prod_{i=0}^{s-1} \abs{G^{(i)}} / \abs{G^{(i+1)}}
    %     = \abs{G^{(0)}} / \abs{G^{(s)}}
    %     = \abs{G} / \abs{\set{id}}
    %     = \abs{G}.
    % \end{equation*}
\end{proof}

\begin{definition}[Strong Generating Set]
    \label{def:strong-gen-set}
    For any $n\in \N$ and group $G \leq \cS_n$, let $B$ be a base that defines the stabilizer chain $G =\vcentcolon G^{(0)} \geq G^{(1)} \geq \ldots \geq G^{(s)} = \inner{\id}$.
    A \emph{strong generating set} for $G$ with respect to base $B$ is a set $S$ that generates $\inner{S} = G$ and satisfies the property $\inner{S \cap G^{(i)}} = G^{(i)}$ for all $i\in [s]$.
\end{definition}

\begin{lem}
    \label{lem:subgroup-gener}
    For any groups $H \leq G$, let $S_H$ be a generating set for $H$ and $R_{G/H}$ be the set of coset representatives of the quotient $G/H$.
    Then, $G$ can be expressed as
    \[
        G = \inner{S_H \cup R_{G/H}}.
    \]
\end{lem}

Before showing how we adapt the structure of a stabilizer set for our purpose of counting automorphisms of a code, we need to introduce code decomposition.

%%% Local Variables:
%%% mode: latex
%%% TeX-master: "main"
%%% End:

\section{Code Decomposition}%
\label{sec:code-decomp}

For the reductions between \GI and graph automorphism problems, Mathon in \cite{Mathon79} relies on the assumption that the input graphs are connected.
Even if a graph unconnected graph, the same reductions can be performed on its complement, which is a connected graph.
The code analog of connectedness is its decomposability into a direct sum of indecomposable codes.
In this section, we define code (in)decomposability and prove some key properties that are essential for our reductions.
To ensure that our reductions work for \emph{any} codes and not just indecomposable codes, we provide an efficient method of decomposing a linear code.

\begin{definition}[Decomposable Codes]
    \label{def:decomposable}
    A code $\code$ is \emph{(direct sum) decomposable} if it is permutation (resp. monomially) equivalent to a direct sum of two or more codes, i.e. $\code \cong \code_1 \oplus \ldots \oplus \code_m$ (resp. $\code \cong_\cM \code_1 \oplus \ldots \oplus \code_m$) for some codes $\code_1,\ldots,\code_m$.
    Otherwise, we say that $\code$ is \emph{indecomposable}.
\end{definition}

\subsection{Existence and Uniqueness}
\label{subsec:Slepian}

The existence and uniqueness of the direct sum decomposition of linear codes was originally established by Slepian in \cite{Slepian60}.
We restate these facts here and extend them to monomial equivalence.
 
\begin{lem}[\cite{Slepian60}, Theorem 2]
    \label{lem:Slepian}
    Any $[n,k]$-code $\code$ is equivalent to a direct sum of indecomposable codes, i.e.
    \begin{equation*}
        \code \cong \cA_1 \oplus \dots \oplus \cA_m,
    \end{equation*}
    for some indecomposable $\cA_1, \ldots, \cA_m$. 
    Furthermore, this decomposition is unique up to permutation equivalence and reordering of the summands, i.e. if
    % \begin{equation*}
        $\code \cong \cB \oplus \dots \oplus \cB_{m'},$
    % \end{equation*}
    for indecomposable codes $\cB_1,\ldots, \cB_{m'}$, then $m=m'$ and $\cB_i \cong \cA_{\sigma(i)}$ for all $i \in [m]$ for some index permutation $\sigma \vcentcolon [m] \to [m]$.
\end{lem}

Indecomposable codes have the following useful property.

\begin{lem}
 \label{lem:exist_bridge}
    A code $\code$ is indecomposable if and only if for every non-trivial partition of its index set $X = Y \cup Z$, there exists at least one bridging element $\vecc \in \code$ with respect to $(Y, Z)$.
\end{lem}

\begin{proof}
    $(\Rightarrow)$ Let $\code$ be indecomposable.
    Assume for contradiction that there exists a non-trivial partition $X = Y \sqcup Z$ for which there is no bridging element in $\code$.
    Then, every codeword $\vecc \in \code$ whose support intersects both $Y$ and $Z$ can be expressed as a sum of two vectors with supports in $Y$ and $Z$, i.e. $\vecc = \vecu + \vecv$ for some $\vecu, \vecv \in \code$ such that $\supp(\vecu) \subseteq Y$ and $\supp(\vecv) \subseteq Z$.
    Hence, $\code$ can be written as a direct sum $\code = \code|_Y \oplus \code|_Z$, which contradicts the indecomposability of $\code$.
    
    $(\Leftarrow)$ We prove the contrapositive. 
    Suppose $\code$ is decomposable and let $X$ denote its index set.
    Then $\code \cong \code_1 \oplus \code_2$ for some codes $\code_1, \code_2$.
    By definition of direct sum, $\code_1, \code_2$ are supported index sets $Y, Z$ that form a strict, non-trivial partition $X = Y \sqcup Z$. 
    Also every codeword $\vecc \in \code$ can be uniquely written as $\vecc = [\vecu\ \vecv]$ for some $\vecu \in \code_1$ and $\vecv \in \code_2$. 
    As a result of the partition, we have $\supp(\vecu) \subseteq Y$ and $\supp(\vecv) \subseteq Z$.
    Consequently, for this partition $(Y, Z)$, every codeword in $\code$ decomposes into a sum of elements, each with non-intersecting support across the partition. 
    Thus, there is no bridging element with respect to $(Y,Z)$. 
    This proves the contrapositive claim.
\end{proof}

We prove an analogous statement to \cref{lem:Slepian} for monomial equivalence.
To do this, we use the notion of the \emph{closure of a code} introduced by Sendrier and Simos in \cite{SS13b}.
Fix an ordering of the elements in $\F_q^*$.
For any code $\code \subseteq \F_q^n$, its closure $\widetilde{\code} \subseteq \F_q^{n(q-1)}$ is the set defined by taking every codeword and multiplying each of its $n$ coordinates with all the non-zero field elements to produce a new vector of length $n(q-1)$. 
Formally, 
\begin{equation*}
    \widetilde{\code} \defn 
    \set{(\alpha_{1}\cdot c_{i},\ldots,\alpha_{q-1}\cdot c_i)_{i \in [n]} \mid \vecc = (c_1,\ldots,c_n) \in \code, \alpha_1,\ldots,\alpha_{q-1} \in \F_q^*}
    \subseteq \F_q^{n(q-1)}.
\end{equation*}

\begin{lem}[Theorem 1~\cite{SS13b}]
    \label{lem:closure-equivalence}
    For any two codes $\code_1, \code_2 \subseteq \F_q^n$, $\code_1 \cong_\cM \code_2$ if and only if $\widetilde{\code_1} \cong \widetilde{\code_2}$. 
\end{lem}

\begin{lem}
    \label{lem:closure-properties}
    The closure operation has the following properties.
    \begin{enumerate}[(i)]
        \item It distributes over direct sums, i.e. any direct sum code $\code = \bigoplus_{i=1}^m \cA_i$ has closure
        $\widetilde{\code} = \bigoplus_{i=1}^m \widetilde{\cA_i}$.
        
        \item $\code$ is indecomposable if and only if its closure $\widetilde{\code}$ is indecomposable.
    \end{enumerate}
\end{lem}

\begin{proof}
    \emph{(i)} 
    Let $\code = \bigoplus_{i=1}^m \cA_i$ be a direct sum of indecomposable codes $\cA_1,\ldots,\cA_m$ and let $\mG$ be a $k\times n$ generator matrix for $\code$.
    Then, by definition of direct sum code, $\mG$ can be written in block-diagonal form, where each $i$-th block $\mG_i$ on the diagonal is a generator matrix for $\cA_i$, for all $i\in [m]$.
    By definition of closure, a generator matrix for $\widetilde{\code}$ can be obtained by replacing each column $\vecg_j$ of $\mG$ with the matrix containing its non-zero scalar multiples, $\bracks{\alpha_1 \vecg_j\ \dots\ \alpha_{q-1} \vecg_j}$, where $\F_q^* = \set{\alpha_1,\ldots,\alpha_{q-1}}$. 
    Since scalar multiplication leaves any $0$ entry unchanged, this column expansion produces a $k\times n(q-1)$ matrix $\widetilde{\mG}$ in block-diagonal form, where the $i$-th block on the diagonal is exactly the expansion of $\mG_i$, for all $i\in [m]$. 
    Therefore, by definition of direct sum, $\widetilde{\mG}$ generates the direct sum code $\bigoplus_{i=1}^m \widetilde{\cA_i}$.
    \smallskip

    \emph{(ii)}
    Let $X = [n]$ be the index set of the indecomposable code $\code$. 
    By definition, the closure operation expands each coordinate of the code into a subspace of its scalar multiples.
    Hence, every coordinate of $\code$ with index $j\in X$ corresponds to a projective class in $\widetilde{\code}$ whose index set is $E_j = \set{j(q-1)-q+2,\ldots,j(q-1)}$. 
    % 1, 2,..., q-1
    % q, q+1,..., 2q-1
    % 2q, 2q+1,..., 3q-1,...

    % First we prove the forwar direction.
    Suppose for contradiction that $\code$ is indecomposable and its closure $\widetilde{\code}$ is decomposable.
    Then by \cref{lem:exist_bridge}, the index set of $\widetilde{\code}$ has a non-trivial partition $Y' \sqcup Z' = [n(q-1)]$ such that $\widetilde{\code} = \widetilde{\code}\mid_{Y'} \oplus\ \widetilde{\code}\mid_{Z'}$. 
    For any $j\in X$, all coordinates of $\widetilde{\code}$ with index in $E_j$ are scalar multiples of each other, and so are linearly dependent. 
    Then, by definition of direct sum, these coordinates cannot be expressed in terms of more than one indecomposable summand of $\widetilde{\code}$.
    Hence, the entire set $E_j$ must be contained either entirely within $Y'$ or entirely within $Z'$.
    This implies that there must be a non-trivial partition $Y \sqcup Z = [n]$ of the original index set, such that $j \in Y \iff E_j \subseteq Y'$ and $j \in Z \iff E_j \subseteq Z'$. 
    Then for this partition, we have $\code = \code\mid_Y \oplus \code\mid_Z$, which contradicts the assumption that $\code$ is indecomposable.
    The converse follows by a symmetric argument.
\end{proof}

\begin{lem}
    \label{lem:monomial_decomp}
    Every code $\code \subseteq \F_q^n$ is monomially equivalent to a direct sum of indecomposable codes, i.e.
    \begin{equation*}
        \code \cong_{\cM} \cA_1 \oplus \dots \oplus \cA_m, %\bigoplus_{i=1}^m A_i
    \end{equation*}
    for some indecomposable $\cA_1, \ldots, \cA_m$.
    Furthermore, this decomposition is unique up to monomial equivalence and reordering of the summands, i.e. if
    % \begin{equation*}
        $\code \cong_\cM \cB \oplus \dots \oplus \cB_{m'}$,
    % \end{equation*}
    for indecomposable codes $\cB_1,\ldots, \cB_{m'}$, then $m = m'$ and $\cB_i \cong_\cM \cA_{\sigma(i)}$ for all $i \in [m]$ for some index permutation $\sigma \vcentcolon [m] \to [m]$.
\end{lem}

\begin{proof}
    (\emph{Existence})
    By \cref{lem:Slepian}, every code $\code$ is equivalent to a direct sum of indecomposable codes, i.e.
    $\code \cong \cA_1 \oplus \dots \oplus \cA_m$
    for some indecomposable codes $A_1,\ldots,A_m$.
    Because permutation equivalence implies monomial equivalence, it trivially follows that
    $\code \cong_{\cM} \cA_1 \oplus \dots \oplus \cA_m$.
    \smallskip

    % \noindent
    (\emph{Uniqueness})
    Let $\code_A \defn \cA_1 \oplus \dots \oplus \cA_m$ and $\code_B \defn \cB_1 \oplus \dots \oplus \cB_{m'}$.
    Suppose $\code \cong_\cM \code_A$ and $\code \cong_\cM \code_B$. By transitivity, $\code_A \cong_\cM \code_B$.
    By \cref{lem:closure-equivalence}, $\code_A \cong_\cM \code_B$ implies that their closure codes satisfy $\widetilde{\code_A} \cong \widetilde{\code_B}$.
    By \cref{lem:closure-properties}, %the closure operation distributes over direct sums and strictly preserves indecomposability. Thus, 
    we have
    $\widetilde{\code_A} = \widetilde{\cA_1} \oplus \dots \oplus \widetilde{\cA_m}$
    and $\widetilde{\code_B} = \widetilde{\cB_1} \oplus \dots \oplus \widetilde{\cB_{m'}}$,
    % \begin{equation*}
    %     \widetilde{\code_A} = \bigoplus_{i=1}^m \widetilde{\cA_i}
    %     \quad \text{and}\quad
    %     \widetilde{\code_B} = \bigoplus_{j=1}^{m'} \widetilde{\cB_j},
    % \end{equation*}
    where each summand $\widetilde{\cA_i}$ and $\widetilde{\cB_j}$ is an indecomposable code.
    Substituting these direct sum identities into the equivalence yields 
    \begin{equation*}
        \widetilde{\code} = \widetilde{\cA_1} \oplus \dots \oplus \widetilde{\cA_m}
        \cong \widetilde{\cB_1} \oplus \dots \oplus \widetilde{\cB_{m'}} = \widetilde{\code_B}.
    \end{equation*}
    % $\widetilde{\code} = \bigoplus_{i=1}^m \widetilde{\cA_i}
    % \cong \bigoplus_{j=1}^{m'} \widetilde{\cB_j} = \widetilde{\code_B}$.
    By \cref{lem:Slepian}, the decomposition of a code into indecomposable summands is unique up to permutation equivalence. 
    This implies that the number of summands is identical, so $m = m'$, and there exists a bijection $\sigma \vcentcolon [m] \to [m']$ such that $\widetilde{\cA_i} \cong \widetilde{\cB_{\sigma(i)}}$ for all $i \in [m]$.
    By applying the reverse direction of \cref{lem:closure-equivalence} to the individual summands, the permutation equivalence $\widetilde{\code_A} \cong \widetilde{\code_B}$ implies the monomial equivalence $\cA_i \cong_\cM \cB_{\sigma(i)}$, for each $i\in [m]$.
    This proves the claim.
\end{proof}

\subsection{Properties of Indecomposable Codes}
\label{subsec:properties-decomp}

\begin{lem}
\label{lem:no_partial}
    Let $\code = \bigoplus_{i=1}^m \code_i \subseteq \F_q^n$ be a direct sum of indecomposable codes $\code_i$.
    % , whose supports $\set{X_i}_{i=1}^m$ partition the index set $X = \bigsqcup_{i=1}^m X_i = [n]$.
    For any monomial automorphism $\mu \in \MAut(\code)$, there exists a permutation $\sigma \in S_m$ of the summands such that $\mu(C_i) = C_{\sigma(i)}$ for all $i \in [m]$.
\end{lem}

Here we use $\mu(\code_i)$ to denote the image of $\mu$ on the embedded code $\set{\veczero} + \dots + \code_i + \dots + \set{\veczero}$, which gives some other embedded code $\set{\veczero} + \dots + \code_{\sigma(i)} + \dots + \set{\veczero}$ in the direct sum code $\code$.

\begin{proof}
    Let $\set{X_i}_{i=1}^m$ be supports of the summands $\set{C_i}_{i=1}^m$ in the index set $X = [n]$ of $\code$; these form a partition $X = \bigsqcup_{i=1}^m X_i$.
    Fix an arbitrary index $i \in [m]$ and consider any $\mu = (\vecv_\mu,\pi_\mu) \in \MAut(\code)$.
    We will show that the induced permutation on the index set, $\pi_\mu(X_i)$, is entirely contained within a single coordinate block $X_j$ for some $j \in [m]$.
    Since the index sets $\set{X_j}_{j=1}^m$ partition $X = [n]$,  the elements of $X_i$ can be partitioned according to their image under the permutation $\pi_\mu$.
    Define
    \begin{equation*}
        Y_{i,j} \defn X_i \cap \pi_\mu^{-1}(X_j) \quad \text{ for } j \in [m].
    \end{equation*}
    Then the collection of non-empty sets among the $\set{Y_{i,j}}_{j=1}^m$ forms a valid partition of $X_i$.
    
    Suppose, for the sake of contradiction, that the image $\pi_\mu(X_i)$ intersects multiple index sets non-trivially.
    Then there exist at least two distinct indices, $a$ and $b$, such that both $Y_{i,a} \neq \emptyset$ and $Y_{i,b} \neq \emptyset$. 
    This gives a partition of $X_i$ into
    $X_i = Y_{i,a} \sqcup Z$, where $Z \defn \bigsqcup_{j \neq a} Y_{i,j}$.
    % \begin{equation*}
    %     X_i = Y_{i,a} \sqcup Z, \quad \text{where } Z \defn \bigsqcup_{j \neq a} Y_{i,j}.
    % \end{equation*}
    % 
    By \cref{lem:exist_bridge}, since $\code_i$ is an indecomposable code supported on $X_i$, there exists a bridging element $\vecc \in \code_i$ with respect to the partition $\set{Y_{i,a}, Z}$. 
    By definition, $\vecc$ has support in both $Y_{i,a}$ and $Z$, and it cannot be decomposed into a sum of codewords whose supports are restricted only to $Y_{i,a}$ or $Z$.
    Now consider the image of this codeword under $\mu$ and denote this by $\vecc' \defn \mu(\vecc) \in \code$. 
    Because $\code$ is a direct sum code, every codeword in $\code$ decomposes uniquely into a sum of components, each restricted to an $X_j$ index set.
    This allows us to write $\vecc' = \sum_{j=1}^m \vecu'_j$,
    % \begin{equation*}
    %     \vecc' = \sum_{j=1}^m \vecu'_j,
    % \end{equation*}
    where $\vecu'_j \in \code_j$ and $\supp(\vecu'_j) \subseteq X_j$.
    We apply the inverse automorphism $\mu^{-1} = (\vecv_\mu^{-1},\pi_\mu^{-1})$ to this decomposition.
    The permutation component of $\mu^{-1}$ is $\pi_\mu^{-1}$ and its scaling vector is $\vecv_\mu^{-1} = (v_1^{-1},\ldots,v_n^{-1})$ whose entries are the inverses of the non-zero scalar entries of $\vecv_\mu$.
    Applying $\mu^{-1}$ to the decomposition of $\vecc'$, we obtain a decomposition of the original codeword
    \begin{equation*}
        \vecc = \mu^{-1}(\vecc') = \sum_{j=1}^m \mu^{-1}(\vecu'_j).
    \end{equation*}
    Since $\supp(\vecu'_j) \subseteq X_j$ and $\mu$ induces a bijection, it follows that $\supp(\mu^{-1}(\vecu'_j)) \subseteq \mu^{-1}(X_j)$. 
    Furthermore, since $\vecc$ is supported entirely within $X_i$, each non-zero component $\mu^{-1}(\vecu'_j)$ must be supported within $X_i \cap \pi_\mu^{-1}(X_j)$, which is precisely the set $Y_{i,j}$.
    These terms can be grouped according to the partition $\set{Y_{i,a}, Z}$, into
    \begin{equation*}
        % \vecc = \underbrace{\mu^{-1}(\vecu'_a)}_{\vecv_1} + \underbrace{\sum_{j \neq a} \mu^{-1}(\vecu'_j)}_{\vecv_2}.
        \vecc = \mu^{-1}(\vecu'_a) + \sum_{j \neq a} \mu^{-1}(\vecu'_j).
    \end{equation*}
    Denote the summands by $\vecv_1 \defn \mu^{-1}(\vecu'_a)$ and $\vecv_2 \defn \sum_{j \neq a} \mu^{-1}(\vecu'_j)$.
    By construction, $\supp(\vecv_1) \subseteq Y_{i,a}$. 
    Since all other terms map into sets distinct from $X_a$, their supports are contained in $Z$, and so $\supp(\vecv_2) \subseteq Z$.
    Both $\vecv_1$ and $\vecv_2$ are valid codewords in $\code$, which means that $\vecc$ can be written as a sum $\vecv_1 + \vecv_2$ whose summands have disjoint supports across $\set{Y_{i,a}, Z}$. 
    But this contradicts the fact that $\vecc$ is a bridging element.
    % Thus, exactly one subset $Y_{i,j}$ must be non-empty. 
    % This implies that $\pi_\mu(X_i) \subseteq X_j$ for some $j \in [m]$. 
    % Since $\pi_\mu$ is a bijection on the finite set $X$, and the collection $\set{X_i}_{i=1}^m$ partitions $X$, the containment $\pi_\mu(X_i) \subseteq X_j$ is actually an equality, $\pi_\mu(X_i) = X_j$. 
    % Therefore, $\mu$ induces a permutation $\sigma \in S_m$ on the indices such that $\pi_\mu(X_i) = X_{\sigma(i)}$.
    Thus, exactly one subset $Y_{i,j}$ must be non-empty. 
    This implies that $\pi_\mu(X_i) \subseteq X_j$ for some $j \in [m]$.
    
    Now consider the inverse automorphism $\mu^{-1} \in \MAut(C)$. 
    Since $\pi_\mu(X_i) \subseteq X_j$, we have $X_i \subseteq \pi_\mu^{-1}(X_j)$.
    Applying the same argument as above to $\mu^{-1}$, we obtain that the set $\pi_\mu^{-1}(X_j)$ is contained entirely in the index set of a single summand.
    Since $X_i \subseteq \pi_\mu^{-1}(X_j)$, this summand must be $X_i$, so $\pi_\mu^{-1}(X_j) \subseteq X_i.$
    Thus, $\pi_\mu^{-1}(X_j)=X_i$, and equivalently, $\pi_\mu(X_i)=X_j$.
    Therefore, $\mu$ induces a permutation $\sigma \in S_m$ on the summands such that $\pi_\mu(X_i)=X_{\sigma(i)}$ for all $i \in [m]$.
\end{proof}

Since $\Aut(\code) \subseteq \MAut(\code)$ for any code $\code$, \cref{lem:no_partial} immediately gives the following corollary.

\begin{cor}
\label{cor:no_partial-Aut}
    Let $\code = \bigoplus_{i=1}^m \code_i \subseteq \F_q^n$ be a direct sum of indecomposable codes $\code_i$.
    For any automorphism $\pi \in \Aut(\code)$, there exists a permutation $\sigma \in S_m$ of the summands such that $\pi(C_i) = C_{\sigma(i)}$ for all $i \in [m]$.
\end{cor}

\begin{lem}
  \label{lem:ratio}
    If $\code_1, \code_2 \subseteq \F_q^n$ are indecomposable codes and $\code \defn \code_1 \oplus \code_2$, 
    then they satisfy the ratio    
    \[
    \frac{\abs{\MAut(\code)}}{\abs{\MAut(\code_1)} \cdot \abs{\MAut(\code_2)}} = 
    \begin{cases} 
    1 & \text{if } \code_1 \not\cong_{\cM} \code_2 \\
    2 & \text{if } \code_1 \cong_{\cM} \code_2.
    \end{cases}
    \]
\end{lem}

\begin{proof}
    Let $X = X_1 \sqcup X_2$ be the index set partition induced by the direct sum $\code = \code_1 \oplus \code_2$.
    Define the direct product of the monomial automorphism groups of $\code_1$ and $\code_2$ by 
    \[
        H \defn \MAut(\code_1) \times \MAut(\code_2) = \set{(\mu_1, \mu_2) \mid \mu_i \in \MAut(\code_i), i=1,2} \subseteq \MAut(\code).
    \]
    % $H \defn \MAut(\code_1) \times \MAut(\code_2) = \set{(\sigma_1, \sigma_2) \mid \sigma_i \in \MAut(\code_i), i=1,2} \subseteq \MAut(\code)$.
    We consider the two possible cases for equivalence between $\code_1$ and $\code_2$.
    \smallskip
    
    ($\code_1 \not\cong_{\cM} \code_2$)
    We claim that $H = \MAut(\code)$. 
    Assume there exists $\mu \in \MAut(\code) \setminus H$; this has the form $\mu = (\vecv_\mu, \pi_\mu)$ for some $\vecv_\mu \in (\F_q^*)^{2n}$ and $\pi_\mu\in \cS_{2n}$.
    The associated permutation $\pi_\mu$ must map at least one index from $X_1$ to $X_2$.
    So, by \cref{lem:no_partial}, $\pi_\mu(X_1) = X_2$ which contradicts the assumption.
    Therefore, $H = \MAut(\code)$, and $\abs{\MAut(\code)} = \abs{\MAut(\code_1)} \cdot \abs{\MAut(\code_2)}$, so the value of the ratio is $1$.
    \smallskip
    
    ($\code_1 \cong_{\cM} \code_2$)
    By \cref{lem:no_partial}, there exists a monomial automorphism $\mu = (\vecv_\mu, \pi_\mu) \in \MAut(\code)$ whose associated permutation $\pi_\mu$ permutes the coordinates of $\code_1$ to $\code_2$ i.e. $ \pi_\mu(X_1) = X_2$ and $\pi_\mu(X_2) = X_1$. 
    Then $\mu \notin H$, and so the cosets $H$ and $\mu H$ are disjoint.
    This implies $\abs{\MAut(\code)} \geq 2 \cdot \abs{\MAut(\code_1)} \cdot \abs{\MAut(\code_2)}$.
    Suppose some $\mu' = (\vecv_{\mu'},\pi_{\mu'}) \in \MAut(\code)$ permutes $\pi_{\mu'}(X_1) = X_2$ and $\pi_{\mu'}(X_2) = X_1$, otherwise $\mu' \in H$.
    Consider the composition $\kappa \defn \mu^{-1} \circ \mu'$.
    The permutation associated with $\kappa$ acts on $X_1$ as $\pi_\kappa(X_1) = \pi_{\mu^{-1}}(\pi_{\mu'}(X_1)) = \pi_\mu^{-1}(X_2) = X_1$.
    Since $\kappa$ is a bijection on the finite set $X$, it must similarly map $X_2$ to $X_2$. 
    Since $\kappa$ preserves both index sets, we have $\kappa \in H$.
    Since $\mu^{-1} \circ \mu' = \kappa \in H$, multiplying by $\mu$ on the left yields $\mu' = \mu \circ \kappa \in \mu H$.

    Therefore, every element of $\MAut(\code)$ is either in $H$ or in the left coset $\mu H$, establishing $\MAut(\code) = H \sqcup \mu H$.
    Consequently, $\abs{\MAut(\code)} = 2 \cdot \abs{\MAut(\code_1)} \cdot \abs{\MAut(\code_2)}$, and the ratio has value $2$.
\end{proof}

\begin{lem}
  \label{lem:ratio-Aut}
    If $\code_1, \code_2 \subseteq \F_q^n$ are indecomposable codes and $\code \defn \code_1 \oplus \code_2$, 
    then they satisfy the ratio    
    \[
    \frac{\abs{\Aut(\code)}}{\abs{\Aut(\code_1)} \cdot \abs{\Aut(\code_2)}} = 
    \begin{cases} 
    1 & \text{if } \code_1 \not\cong \code_2 \\
    2 & \text{if } \code_1 \cong \code_2.
    \end{cases}
    \]
\end{lem}

\begin{proof}
    The proof for the permutation automorphism groups is identical to the proof of \cref{lem:ratio} for monomial automorphism groups.
    This is because the effect of scaling in $\MAut(\code)$ does not matter, but only the permutation of indices in the index sets.
\end{proof}

\begin{remark}
    A similar result to \cref{lem:ratio-Aut} is given in Equation~34 of \cite{Slepian60}.
    Slepian's formula computes the number of double cosets of any finite group with respect to two of its subgroups.
    Since no proof or reference is given, and we only require a specific restriction of this ratio that has only two possible evaluations, we present \cref{lem:ratio,lem:ratio-Aut} as independent results.
\end{remark}

\begin{lem}
    \label{lem:decomp-matrix-structure}
    Let $\code$ be a decomposable $[n,k]$-code such that $\code \cong \bigoplus_{i=1}^m \code_i$ for some $[n_i,k_i]$-codes $\code_i$.
    Denote its corresponding index set partition by $X = \bigsqcup_{i=1}^m X_i$, where $X_i \subseteq [n]$ is the index set of $\code_i$ embedded in $X$.
    Then, the unique RREF generator matrix $\mG_{\rref}$ of $\code$ has the following properties:
    \begin{enumerate}[(i)]
         \item The set of rows of $\mG_{\rref}$ can be partitioned into disjoint sets $R_1, \ldots, R_m \subseteq [k]$, such that $|R_i| = k_i$, and every row in $R_i$ has its pivot in $X_i$.
         
        \item For every row vector $\vecr$ of $\mG_{\rref}$, there exists exactly one index $i \in [m]$ such that $\supp(\vecr) \subseteq X_i$.
    \end{enumerate}
\end{lem}

\begin{proof}
    For simplicity, denote $\mG \defn \mG_{\rref}$.
    Let $P \subset [n]$ be the set of pivot indices in $\mG$; it has size $\abs{P} = k$. 
    \smallskip
    
    \emph{(i)}
    By the definition of a direct sum code, the dimension of $\code$ projected onto each embedded index set $X_i$ is $\dim(\code|_{X_i}) = \dim(\code_i) = k_i$. 
    For every $i\in [m]$, let $\mG|_{X_i}$ denote the $k \times n_i$ submatrix of $\mG$ comprised of the columns in $X_i$.
    This matrix generates the projected code $\code|_{X_i}$, so $\rank(\mG|_{X_i}) = \dim(\code|_{X_i}) = k_i$.
    
    Define $R_i \defn \set{ \vecr \in \mG \mid \vecr \text{ has its pivot in a column belonging to } X_i}$ for every $i\in [m]$.
    Since the sets $X_i$ are mutually disjoint, the row sets $R_i$ are also disjoint. 
    Because the rows of $\mG$ form a basis of $\code$, every row has exactly one pivot, and so $\bigsqcup_{i=1}^m R_i$ partitions the set of rows of $\mG$. 
    Denote $k_i' \defn \abs{R_i}$ for every $i\in [m]$.
    By the hypothesis $\code \cong \bigoplus_{i=1}^m \code_i$ and the partition of rows of $\mG$, we have $\sum_{i=1}^m k_i' = k = \sum_{i=1}^m k_i$. 
    Since there are exactly $k_i'$ rows whose pivots fall in $X_i$, the rank of the submatrix $\mG|_{X_i}$ must be at least $\rank(\mG|_{X_i}) = k_i \geq k_i'$.
    Hence, since $\sum_{i=1}^m k_i' = \sum_{i=1}^m k_i$ and $k_i \geq k_i'$ for all $i\in [m]$, we must strictly have $k_i' = k_i$ for all $i \in [m]$.
    Therefore, the rows of $\mG$ can be partitioned into these sets $R_1,\ldots,R_m$.
    \medskip
    
    \emph{(ii)}
    Let $\vecr$ be an arbitrary row vector in $\mG$.
    Since $\code \cong \bigoplus_{i=1}^m \code_i$, the vector $\vecr$ can be uniquely decomposed as a sum of its projections onto each index set $X_i$, i.e. $\vecr = \sum_{i=1}^m \vecv^{(i)}$, where $\supp(\vecv^{(i)}) \subseteq X_i$. 
    By the definition of the direct sum, every $\vecv^{(i)}$ is a valid codeword in $\code$.

    We claim that only one of the summands $\vecv^{(j)}$ is non-zero.
    By part 1, $\vecr$ belongs to exactly one set $R_j$, meaning that its pivot is in some column with index $\ell \in X_j$.
    Then, by definition of RREF, $\vecr$ has entry $1$ at index $\ell$ and $0$ at every other index in $P\setminus\set{\ell}$.

    Take any vector $\vecv^{(i)}$ where $i\neq j$ and consider the entries of $\vecv^{(i)}$ in the pivot columns $P$.
    By definition, $\supp(\vecv^{(i)}) \subseteq X_i$ and $X_i \cap X_j - \emptyset$.
    Then, $\vecv^{(i)}$ has entry $0$ at the pivot index $\ell \in P$ of $\vecr$, and $\vecv^{(i)}_{\ell'} = 0$ at any other index $\ell' \in P \setminus \set{\ell}$.
    Hence, $\vecv^{(i)}_p = 0$ for all $p \in P$.
        
    Now, since $\vecv^{(i)} \in \code$, it can be written as a linear combination of the rows of $\mG$. 
    By definition of RREF, the submatrix consisting of columns with indices in $P$ forms a $k\times k$ identity matrix. 
    Since $\vecv^{(i)}_p = 0$ for all $p \in P$ and any linear combination of the rows of $\mG$ includes a sum of scalings of the $k$ pivots in $P$, all scalers in the linear combination must be zero.
    Hence, $\vecv^{(i)} = \veczero$ for all $i\neq j$.
    Therefore, $\vecr = \vecv^{(j)}$ and so $\supp(\vecr) \subseteq X_j$.
\end{proof}

\subsection{How to Decompose a Code}
\label{subsec:how-to-decompose}

In this section we present a method for efficiently decomposing a linear code into its indecomposable summands.
In \cite{Slepian60}, Slepian presented a method to determine whether a code $\code$ is indecomposable: 
if its $k\times n$ generator matrix can be put into systematic form, and the columns in the identity block $\mI_k$ form a single connected component in the $k$-vertex graph defined by ``connected coordinates", then the code is indecomposable.
This method only indicates \emph{whether} a code is indecomposable.
The structural decomposition of linear codes shares a fundamental connection with the notion of connected components in matroid theory in the book \cite[Section 8.4]{Matroid_theory}. 
Building on this connection, Kashyap in \cite{Kashyap07} showed how algorithmic techniques for matroid decomposition can be explicitly translated to binary linear codes. 

We generalize these approaches to obtain an efficient method of decomposing any linear code over an arbitrary finite field $\F_q$.
While the foundational idea for this decomposition algorithm has been treated as folklore in several works, including \cite{Slepian60,Hoffman82}, we provide the explicit pseudocode below and formally analyze its complexity.

Our algorithm processes the RREF generator matrix of the input code to obtain a graph representation.
By finding all connected components in this graph, the indecomposable summands can then be identified.
For the code $\code$ generated by the input matrix, the algorithm outputs an explicit decomposition of $\code$ as a direct sum of some indecomposable codes $\code_1,\ldots,\code_m$ up to a permutation of coordinates, along with the permutation $\pi$ of the columns of the generator matrix that maps $\pi(\code) = \code_1 \oplus \ldots \oplus \code_m$.
This algorithm will be used as a subroutine in our reductions.

\begin{lem}
    \label{lem:alg-decompose}
    For any code $\code \subseteq \F_q^n$, \cref{alg:decompose} correctly decomposes it in $O(n^3\log(q))$ time and finds a set of generators $\set{\mG_1, \ldots, \mG_m}$ for indecomposable codes $\set{A_1,\ldots,A_m}$ and an associated permutation $\pi \in \cS_n$ satisfying $\pi(\code) = \bigoplus_{i=1}^m A_i$.
\end{lem}

\begin{algorithm}[th]
    \SetAlgoLined
    \DontPrintSemicolon
    \SetKwInOut{Input}{Input}
    \SetKwInOut{Output}{Output}
    \caption{Decomposition of a Linear Code}
    \Input{A generator matrix $\mG \in \F_q^{k \times n}$ of a $[n,k]$-code $\code$.}
    \Output{A set generators $\set{\mG_1, \ldots, \mG_m}$ for indecomposable codes $\set{A_1, \ldots, A_m}$ and a coordinate permutation $\pi \in \cS_n$ such that $\pi(\code) = \bigoplus_{i=1}^m A_i$.}

    \BlankLine
    % \tcp{Step 1: Reduced Row Echelon Form}
    1. Compute the RREF matrix $\mG_{\rref}$ from $\mG$. \;
    \BlankLine
    % \tcp{Step 2: Support Graph Construction}
    2. Construct the support graph $\Gamma(\code) = ([n],E)$ from $\mG_{\rref}$. \;
    \BlankLine
    % \tcp{Step 3: Connected Components}
    3. Run Breadth-First Search on $\Gamma(\code)$ to identify its $m$ connected components $X_1,\ldots,X_m$
    Fix any ordering of the elements in each component to get the sequences $Z_1,\ldots,Z_m$. \;
    \BlankLine
    % \tcp{Step 4: Construction of Permutation}
    4. Construct a permutation $\pi \in \cS_n$ that reorders the indices $[n]$ by mapping them to the concatenated sequence $Z_1,\ldots,Z_m$.\;
    \BlankLine
    % \tcp{Step 5: Generator Matrices of Summands}
    5. Define $\mG_1, \ldots, \mG_m$ to be the submatrices of $\mG_{\rref}$ corresponding to the non-zero rows and columns restricted to $X_1, \ldots, X_m$, respectively.\;
    \BlankLine
    % \tcp{Step 6: Extract Summands}
    % Define $A_1,\ldots, A_m$ to be the codes generated by $\mG_1, \ldots, \mG_m$, respectively.\;
    % \BlankLine
    \Return{$\set{\mG_1, \ldots, \mG_m}$ and $\pi$}
    \label{alg:decompose}
\end{algorithm}

\begin{proof}
    Let $\code \subseteq \F_q^n$ be the input code of dimension $k$.
    \smallskip

    \emph{(Correctness)}
    By \cref{constr:supp-graph}, the vertices of the support graph $\Gamma(\code)$ represent the coordinate indices of the code $\code$, and the connected components of $\Gamma(\code)$ correspond exactly to the disjoint index sets of the indecomposable summands in the decomposition of $\code$.
    Suppose there are $m$ connected components and let $X_1,\ldots,X_m \subset [n]$ denote the elements in these components.
    These subsets form a partition of the index set $X=[n]$ of $\code$.
    Running Breadth-First Search on the graph $\Gamma(\code)$ will explore all $n$ vertices and identify these $m$ connected components.
    Then, we fix an (arbitrary) ordering for the elements in each component $X_i$.
    Denote the resulting sequence by $Z_i = \set{z_{i_1}, z_{i_2}, \ldots, z_{i_{n_i}}}$.
    Then, we define the permutation $\pi \in \cS_n$ by 
    \begin{equation*}
        \pi(z_{i_j}) \defn j + \sum_{\ell=1}^{i-1} n_\ell \quad \text{ for } 1 \le j \le n_i \text{ and } i \in [m].
    \end{equation*}
    This permutation shifts the ordered indices of $X_1$ to the interval $[1, n_1]$, those of $X_2$ to $[n_1+1, n_1+n_2]$, and so on.
    By \cref{lem:decomp-matrix-structure}, every row vector $\vecr$ of $\mG_{\rref}$ has its support contained within exactly one of the original index sets $X_i$. 
    Hence, after applying $\pi$ to the index set of $\code$, the support of $\vecr$ is entirely contained within $\pi(X_i)$. 
    By construction of $\pi$, the columns of $\mG_{\rref}$ are reordered into disjoint contiguous blocks with index sets $X_1,\ldots, X_m$.
    Thus, the algorithm can identify the block submatrices $\mG_1,\ldots,\mG_m$ of $\mG_{\rref}$ that are supported on $X_1,\ldots,X_m$, respectively.
    By construction and \cref{lem:decomp-matrix-structure}, each $\mG_i$ is a $k_i \times n_i$ matrix and generates a $[n_i,k_i]$-code $A_i$.
    Therefore, for this permutation $\pi$ we have $\pi(\code) = \bigoplus_{i=1}^m A_i$ and the algorithm correctly finds $\pi$ and the corresponding the decomposition $\bigoplus_{i=1}^m A_i$.
    \medskip

    \emph{(Complexity)} 
    We analyze the runtime for each step in the algorithm.
    Each field element in $\F_q$ requires $\log(q)$ bits to represent and $O(\log(q))$ bit operations to read, copy, or compare.
    Computing the RREF matrix $\mG_{\rref}$ from the $k\times n$ matrix $\mG$ can be done using Gaussian elimination in $O(k^2 n \log(q))$ time.
    Constructing the support graph $\Gamma(\code)$ requires iterating through the $k$ rows of $\mG_{\rref}$ and the entries of each row.
    Each row has at most $n$ non-zero entries in $\F_q^*$, so at most $n^2$ edges are created per row. 
    Across all $k$ rows, this takes a total of $O(k n^2 \log(q))$ time.
    Finding the connected components in the graph $\Gamma(\code)$ with $\abs{V} = n$ vertices and $\abs{E} \leq n^2$ edges can be done using Breadth-First Search, which takes $O(n + |E|)= O(n^2)$ time.
    Constructing the permutation $\pi$ involves concatenating $m \leq n$ disjoint sets, which takes $O(n)$ time.
    Extracting the blocks $\mG_1,\ldots,\mG_m$ from $\mG_{\rref}$ corresponding to the connected components $X_1,\ldots,X_m$ requires memory access operations bounded by the matrix size, and takes $O(kn \log(q))$ time.
    % Finally, defining the codes $A_1,\ldots,A_m$ in Step 6 from the generator matrices $\mG_1,\ldots,\mG_m$ takes $O(m) = O(n)$ time.
    % 
    Since $k \le n$ and $m\leq n$, the total runtime is bounded by $O(n^3 \log(q))$.
    % $O(k^2 n) + O(k n^2) + O(n^2) + O(n) + O(kn) + O(n) = O(n^3)$.
\end{proof}

\subsection{Setwise Stabilizer Chains}
\label{subsec:setwise-stab-chain}

Now we introduce an algebraic structure called a \emph{setwise stabilizer chain} and prove some key properties.
This is adapted from the notion of a \emph{pointwise} stabilizer chain (see \cref{def:stabilizer-chain}), but is specific to the automorphism group of a code and is defined by the redundancy sets of identical columns in a generator matrix of the code.
Similar notions have been used in \cite{Mathon79,Hoffman82} for graphs.

Setwise stabilizer chains are used in our reductions in \cref{subsec:ACOUNT->PCE,subsec:AGEN->PCE} in the context of iteratively stabilizing columns of a generator matrix of a code, with the goal of finding or counting its automorphisms.
The columns of the generator matrix are partitioned into classes of identical columns to efficiently handle trivial permutations without resorting to oracle calls in these reductions.

\begin{definition}[Setwise Stabilizer Chain]
    \label{def:setwise-stab-chain}
    For any matrix $\mG \in \F_q^{k \times n}$, let $G \defn \Aut(\code(\mG))$ be the permutation automorphism group of the code generated by $\mG$.
    Denote the column vectors of $\mG$ by $\set{\vecg_1, \ldots, \vecg_n}$ and let $B' = \set{t_1, \ldots, t_s} \subseteq [n]$ be the sequence of indices of unique column vectors in $\mG$ in order of first appearance.
    This sequence $B'$ defines a partition of the index set $[n]$ into disjoint redundancy classes $E_1, \ldots, E_s$, where $E_i \defn \set{c \in [n] \mid \vecg_c = \vecg_{t_i}}$, and let $m_i \defn \abs{E_i}$.
    The \emph{(permutation) setwise stabilizer chain} is the chain of subgroups
    \begin{equation*}
        G =\vcentcolon G^{(E_0)} \geq G^{(E_1)} \geq \dots \geq G^{(E_s)},
    \end{equation*}
    where $G^{(E_i)} \defn \set{\pi \in G^{(E_{i-1})} \mid \pi(E_i) = E_i}$ is the subgroup that stabilizes the set $E_i$ within the preceding subgroup, for each $i\in [s]$.
    \medskip

    Similarly, let $M \defn \MAut(\code(\mG))$ be the monomial automorphism group of the code generated by $\mG$.
    Define the redundancy classes $E_1, \ldots, E_s$ by
    $E_i \defn \set{c \in [n] \mid \vecg_c = \alpha \vecg_{t_i} \text{ for some } \alpha \in \F_q^*}$.
    The corresponding \emph{monomial setwise stabilizer chain} is the chain of subgroups
    \begin{equation*}
        M =\vcentcolon M^{(E_0)} \geq M^{(E_1)} \geq \dots \geq M^{(E_s)},
    \end{equation*}
    where $M^{(E_i)} \defn \set{\mu \in M^{(E_{i-1})} \mid \pi_\mu(E_i) = E_i}$ is the subgroup of monomial automorphisms whose associated permutation $\pi_\mu \in \cS_n$ stabilizes the set $E_i$ within the preceding subgroup, for each $i\in [s]$.
\end{definition}

Note that the identity map is not the only permutation automorphism in $G = \Aut(\code(\mG))$ that stabilizes all indices in the set $B'$, so this $B'$ is not a base as defined in \cref{def:stabilizer-chain}.
Likewise, the smallest subgroup in the chain, $G^{(E_s)}$, does not necessarily contain only the identity automorphism, as it must for a pointwise stabilizer chain.
Instead, $G^{(E_s)}$ is the subgroup generated by column swaps, which is isomorphic to $\prod_{i=1}^s \Sym(E_i)$, the product of the symmetric groups of the equivalence classes.
\medskip

For the monomial automorphism group $M$, the smallest subgroup $M^{(E_s)}$ stabilizes $E_1,\ldots,E_s$. 
We refer to these sets as \emph{projective classes}, following the terminology in \cite{BMPW26}.
If the code $\code(\mG)$ can be decomposed into $c$ indecomposable summands, the subgroup $M^{(E_s)}$ not only contains automorphisms that internally swap columns within each class $E_i$, but also automorphisms that multiply each of the $c$ indecomposable summands by the same scalar.
We will show in \cref{lem:mon-setwise-stab-chain-order} that $M^{(E_s)}$ is isomorphic to 
$(\F_q^*)^c \times \prod_{i=1}^s \Sym(E_i)$.
\medskip

For any code $\code$ and setwise stabilizer chain of $G = \Aut(\code)$, the number of automorphisms in the group can be found be iterative applications of the Orbit-Stabilizer theorem to consecutive pairs of subgroups along the chain.
This is formalized in the following lemma.

\begin{lem}[number of permutation automorphisms]
    \label{lem:setwise-stab-chain-order}
    For any matrix $\mG \in \F_q^{k \times n}$ defining a code $\code = \code(\mG)$, its automorphism group $G \defn \Aut(\code)$, and a sequence $B' = \set{t_1, \ldots, t_s} \subseteq [n]$ of indices of unique column vectors in $\mG$ in order of first appearance, the associated setwise stabilizer chain
    $G \geq G^{(E_1)} \geq \dots \geq G^{(E_s)}$ has cardinality
    \[
        \abs{G} = \prod_{i=1}^{s} \abs[\big]{G^{(E_{i-1})} : G^{(E_i)}} \cdot \parens{m_i!} .
    \]
\end{lem}
\vspace{-7pt}

\begin{proof}
    By applying Lagrange's Theorem to consecutive subgroups along the chain, we obtain
    $\abs{G} = 
    \abs[\big]{G^{(E_s)}} \cdot \prod_{i=1}^{s} \abs[\big]{G^{(E_{i-1})} : G^{(E_i)}}.$
    To prove the claim, it suffices to show that $\abs[\big]{G^{(E_s)}} = \prod_{i=1}^{s} \parens{m_i!}$.
    By definition, the subgroup $G^{(E_s)}$ stabilizes each $E_i$ setwise, that is, its elements map indices from $E_i$ into $E_i$, for every $i\in [s]$.
    Since the sets $E_1, \ldots, E_s$ form a partition of the index set $[n]$, every element $\pi \in G^{(E_s)}$ permutes the indices $[n]$ by independently permuting the elements within each $E_i$.
    By definition, each equivalence class $E_i$ contains the indices of identical columns in $\mG$.
    Because permuting identical columns leaves $\mG$ invariant, every possible permutation of the indices in $E_i$ yields a trivial automorphism of $\code$.
    Since $E_1, \ldots, E_s$ partitions the domain $[n]$ of the automorphisms in $G^{(E_s)}$, there is a bijection between the elements of $G^{(E_s)}$ and the direct product of the symmetric groups of $E_1, \ldots, E_s$, i.e.
    $G^{(E_s)} \cong \Sym(E_1) \times \ldots \times \Sym(E_s)$.
    For each $i \in [s]$, there are $m_i!$ possible permutations for the $m_i$ elements of $E_i$.
    Therefore, we obtain $\abs{G^{(E_s)}} = \prod_{i=1}^{s} \parens{m_i!}$.
\end{proof}

\begin{lem}
    \label{lem:indecomposable-scaling}
    Let $\code \subseteq \F_q^n$ be an indecomposable code. 
    If $\mu = (\vecv_\mu, \id) \in (\F_q^*)^n \rtimes \cS_n$ is a monomial automorphism that only scales $\code$, then $\vecv_\mu = \parens{\lambda,\ldots,\lambda}$ for some $\lambda \in \F_q^*$. Consequently, there are exactly $q-1$ automorphisms in $\MAut(\code)$ that scale without permuting codeword indices.
\end{lem}

\begin{proof}
    Let $\mG \defn \mG_{\rref} \in \F_q^{k\times n}$ be the (unique) generator matrix of $\code$ in RREF and denote its entries by $g_{i,j}$ for every row $i\in [k]$ and column $j\in [n]$.
    Consider the monomial automorphism $\mu = (\vecv_\mu, \id) \in (\F_q^*)^n \rtimes \cS_n$, where the permutation is the identity map.
    By definition, applying $\mu$ to the generator matrix $\mG$ scales each $j$-th column of $\mG$ by the corresponding $j$-th entry $v_j$ in $\vecv_\mu$.
    Since $\mu \in \MAut(\code)$, applying $\mu$ to $\mG$ must yield a matrix $\mH \defn \mu(\mG)$ that generates the same code $\code = \code(\mH)$.
    We show that this can only happen if the entries of $\vecv_\mu$ are the same.

    Let $\mD \defn \diag(v_1\dots v_n)$ be the matrix whose diagonal entries are the coordinates of $\vecv_\mu$.
    By definition, $\mH = \mG \mD$.
    Hence, by definition of RREF, we have $\rref(\mH)=\rref(\mG\mD)= \mG$.
    % We perform Gaussian elimination on $\mH$ to find the scalars needed to obtain $\rref(\mH)$.
    Without loss of generality, we can partition $\mH$ into $\mH = \bracks{\mH_1 \mid \mH_2}$, where $\mH_1$ is a full-rank $k\times k$ submatrix and $\mH_2$ contains the remaining $n-k$ columns
    (otherwise the same argument works for the submatrix containing the pivot columns of $\mH$).
    We multiply $\mH$ by the inverse matrix $\mH_1^{-1}$ to obtain $\rref(\mH) = \mH_1^{-1}\mH = \bracks{\mI_k \mid \mH_1^{-1} \mH_2}$.
    Every $i$-th row of $\mH$ has pivot column at some index $p_i \in [n]$.
    By definition of $\mH$, the pivot entry $\mH_{i,p_i}$ is $v_{p_i}$.
    To ensure that the pivot entry is $1$ in the RREF matrix $\rref(\mH)$, every entry in this row is multiplied by the inverse $v_{p_i}^{-1} \in \F_q^*$.
    After this scaling, any non-pivot entry in the $i$-th row and $j$-th column of $\rref(\mH)$ becomes $g_{i,j} \cdot v_j \cdot v_{p_i}^{-1}$.
    By the equality $\rref(\mH)=\mG$, we have $g_{i,j} = g_{i,j} \cdot v_j \cdot v_{p_i}^{-1}$.
    Hence, for every non-zero entry $g_{i,j} \neq 0$, this implies that $v_j = v_{p_i}$.

    Consider the support graph $\Gamma(\code)$ of the code $\code$.
    By \cref{constr:supp-graph} and the hypothesis that $\code$ is indecomposable, $\Gamma(\code)$ is a connected graph, so there is a path between every pair of vertices.
    % By hypothesis, $\code$ is an indecomposable code, so its generator matrix $\mG$ cannot be expressed as a block-diagonal matrix of other matrices, as in \cref{def:direct-sum}.
    By construction, this can only be the case if there are two non-zero entries in every row of $\mG$.
    A path in $\Gamma(\code)$ between any pair of vertices $j,j'\in [n]$ corresponds to a sequence of rows in $\mG$ with indices $\set{i_1,i_2,\ldots,i_\ell}$ in which the $i_1$-th row has non-zero entries in column $j$ and some other column $j_1$, the $i_2$-th row has non-zero entries in column $j_1$ and some other column $j_2$, and so on until the $i_\ell$-th row which has non-zero entries in column $j_{\ell-1}$ and column $j'$.

    By the condition above, in every $i$-th row, both the pivot entry $g_{i,p_i} \neq 0$ and all other entries $g_{i,j} \neq 0$ are multiplied by the same scalar $v_j^{-1}$.
    By transitivity and the existence of a path between every pair of vertices, we obtain that the scalars are all the same, so that $v_1 = v_2 = \ldots = v_n$.
    Therefore, $\vecv_\mu = \parens{\lambda,\ldots,\lambda}$ for some $\lambda \in \F_q^*$.
    Since there are exactly $q-1$ choices for $\lambda$, there are $q-1$ scaling-only automorphisms in $\MAut(\code)$.
\end{proof}

\begin{lem}[number of monomial automorphisms]
    \label{lem:mon-setwise-stab-chain-order}
    For any matrix $\mG \in \F_q^{k \times n}$ defining a code $\code = \code(\mG)$ that is decomposable into $c$ indecomposable codes, its monomial automorphism group $M \defn \MAut(\code)$, and a sequence $B' = \set{t_1, \ldots, t_s} \subseteq [n]$ of indices of unique projective column representatives in $\mG$ in order of first appearance, the associated monomial setwise stabilizer chain
    $M \geq M^{(E_1)} \geq \dots \geq M^{(E_s)}$ has cardinality
    \[
        \abs{M} = (q-1)^c \cdot \prod_{i=1}^{s} \abs[\big]{M^{(E_{i-1})} : M^{(E_i)}} \cdot \parens{m_i!},
    \]
    where $m_i \defn \abs{E_i}$ for all $i\in [s]$.
\end{lem}

\begin{proof}
    By applying Lagrange's Theorem to consecutive subgroups along the chain, we obtain
    $\abs{M} = \abs[\big]{M^{(E_s)}} \cdot \prod_{i=1}^{s} \abs[\big]{M^{(E_{i-1})} : M^{(E_i)}}$.
    To prove the claim, it suffices to show that $\abs[\big]{M^{(E_s)}} = (q-1)^c \cdot \prod_{i=1}^{s} \parens{m_i!}$.

    By definition, the subgroup $M^{(E_s)}$ stabilizes each projective column class $E_1,\ldots,E_s$ setwise.
    Any monomial automorphism $\mu \in M^{(E_s)} \subset (F_q^*)^n \rtimes \cS_n$ consists of a permutation automorphism $\pi_\mu \in \cS_n$ and a scaling vector $\vecv_\mu \in (\F_q^*)^n$.
    We count the number of possible $\pi_\mu$ and $\vecv_\mu$.

    Consider the homomorphism $\rho \vcentcolon M^{(E_s)} \to \prod_{i=1}^s \Sym(E_i)$ defined by $\rho(\vecv_\mu, \pi_\mu) = \pi_\mu$; this ignores any scaling.
    This map $\rho$ is well-defined because every element in $M^{(E_s)}$ stabilizes $E_1,\ldots,E_s$ setwise.
    $\rho$ is also surjective.
    
    The kernel of $\rho$ consists of all scaling-only automorphisms of the form $(\vecv_\mu, \id) \in M^{(E_s)}$. 
    By hypothesis, $\code \cong_\cM \code_1 \oplus \ldots \oplus \code_c$ for some indecomposable codes $\code_1,\ldots,\code_c$.
    By \cref{lem:indecomposable-scaling}, any scaling-only automorphism in $\MAut(\code)$ must scale each indecomposable summand $\code_i$ by a scalar $\lambda_i \in \F_q^*$.
    Since there are $c$ summands, and the scalar for each summand can be chosen independently, there are $(q-1)^c$ such scaling-only automorphisms.
    Hence, $\abs{\ker (\rho)} = (q-1)^c$.

    Any permutation $\pi_\mu \in \prod_{i=1}^s \Sym(E_i)$ can only permute columns within the same projective class. 
    By definition, the columns in each $E_i$ are all scalar multiples of a single representative vector.
    We can construct a corresponding canonical scaling vector $\vecv_\mu$ that accounts for the scaling of each column in $E_i$ relative to its representative vector, for every class $E_i$.
    For every such canonical vector, there are $(m_i)!$ possible permutations for each class $E_i$.
    Hence, $\abs{\Im(\rho)} = \prod_{i=1}^s (m_i!)$.

    Then, by the First Isomorphism Theorem, we obtain
    $\abs[\big]{M^{(E_s)}} = \abs{\ker \rho} \cdot \abs{\Im(\rho)} = (q-1)^c \cdot \prod_{i=1}^s (m_i!)$.
\end{proof}

\begin{definition}[Strong Generating Set for Setwise Stabilizer Chain]
    \label{def:strong-gen-set-setwise}
    For any $k, n\in \N$ with $k< n$, and any matrix $\mG \in \F_q^{k \times n}$ defining a code $\code = \code(\mG)$, its automorphism group $G' \defn \Aut(\code)$, and a sequence $B' = \set{t_1, \ldots, t_s} \subseteq [n]$ of indices of unique column vectors in $\mG$ in order of first appearance, let $G' \geq G^{(E_1)} \geq \dots \geq G^{(E_s)}$ be the associated setwise stabilizer chain.
   A \emph{setwise strong generating set} for $G'$ with respect to $B'$ is a set $S'$ that generates $\inner{S'} = G'$ and satisfies the property $\inner{S' \cap G^{(E_i)}} = G^{(E_i)}$ for all $i\in [s]$.
\end{definition}

\begin{remark}
    \label{rem:setwise-SGS}
     For any matrix $\mG \in \F_q^{k \times n}$, let $G \defn \Aut(\code(\mG))$ and $B' = \set{t_1, \ldots, t_s} \subseteq [n]$ be the sequence of indices of unique column vectors in $\mG$ in order of first appearance.
     A strong generating set $S$ for $G$ with respect to base $B=[n]$ can be obtained from any strong generating set $S'$ for the setwise stabilizer chain $G = G^{(E_0)} \geq G^{(E_1)} \geq \dots \geq G^{(E_s)}$ defined by $B'$ (see \cite[Section 4.2]{Seress03}).
\end{remark}

\begin{lem}
    \label{lem:setwise-gen}
    For any matrix $\mG \in \F_q^{k \times n}$, let $G \defn \Aut(\code(\mG))$ and $B' = \set{t_1, \ldots, t_s} \subseteq [n]$ be the sequence of indices of unique column vectors in $\mG$ in order of first appearance.
    Let $G = G^{(E_0)} \geq G^{(E_1)} \geq \dots \geq G^{(E_s)}$ be the setwise stabilizer chain defined by $B'$. 
    Suppose that $S_{E_s}$ is a generating set for $G^{(E_s)}$, and $R_i$ is a set of left coset representatives for the quotient space $G^{(E_{i-1})}/G^{(E_i)}$, for each $i \in [s]$. 
    Then, the union of these sets generates the full group group
    \[
        G = \inner[\Big]{ S_{E_s} \cup \bigcup_{i=1}^s R_i }.
    \]
\end{lem}

\begin{proof}
    We use induction on the subgroups of the chain, from smallest to largest.
    For the base case $i=s$, we have $G^{(E_s)} = \inner{S_{E_s}}$ by definition. 
    Suppose that the subgroup $G^{(E_i)}$ is generated by $S_{E_s} \cup \bigcup_{j=i+1}^s R_j$ for some $i\in [s]$. 
    By applying \cref{lem:subgroup-gener} to the adjacent subgroups $G^{(E_i)} \leq G^{(E_{i-1})}$, we obtain that $G^{(E_{i-1})}$ is generated by the union of the generating set for $G^{(E_i)}$ and the coset representatives $R_i$. 
    Iterating this application of \cref{lem:subgroup-gener} for every step $i = s, s-1, \ldots, 1$, we obtain that the full group $G^{(E_0)} = G$ is generated by the union of $S_{E_s}$ and all transversals $R_i$.
\end{proof}

%%% Local Variables:
%%% mode: latex
%%% TeX-master: "main"
%%% End:

\section{\APART and \MAPART}%
\label{sec:APART}

In this section we prove that determining whether two codes are equivalent is a polynomially equivalent problem to computing the partition of a code's coordinates into orbits under its automorphism group.
We first state the main results.

\begin{thm}[\PCE reduces to \APART, \LCE reduces to \MAPART]
    \label{thm:PCE/LCE->APART/MAPART}
    There is a reduction from \PCE (resp. \LCE) on any pair of $n$-length codes to \APART (resp. \MAPART) that runs in $O(n^3 \log(q))$ time and makes at most $n^2$ oracle calls.
\end{thm}

\begin{thm}[\APART reduces to \PCE, \MAPART reduces to \LCE] 
    \label{thm:APART/MAPART->PCE/LCE}
    There is a reduction from \APART (resp. \MAPART) on any $n$-length code to \PCE (resp. \LCE) which runs in $O(n^3 \log(q))$ time and makes $O(n^2)$ oracle calls on codes of length at most $2n$.
\end{thm}

We prove the reduction from \PCE in \cref{subsec:PCE->APART} and its monomial analog in \cref{subsec:LCE->MAPART}.
We prove the reverse reductions in \cref{subsec:APART->PCE,subsec:MAPART->LCE}.

% ===========================================
\subsection{\PCE Reduces to \Apart}
\label{subsec:PCE->APART}

\begin{prop}[\PCE Reduces to \APART]
    \label{prop:alg-PCE->APART}
    \Cref{alg:PCE->APART} solves \PCE for any two $n$-length $q$-ary codes in $O(n^3 \log(q))$ time using at most $n^2$ calls to an \Apart oracle.
\end{prop}

\begin{algorithm}[th]
    \DontPrintSemicolon
    \SetKwInOut{Input}{Input}
    \SetKwInOut{Output}{Output}
    % \SetKwFunction{APART}{APART}

    \caption{Reduction from PCE to APART}
    \label{alg:PCE->APART}
    
    \Input{Generator matrices $\mG_1, \mG_2 \in \F_q^{k \times n}$ for codes $\code_1, \code_2 \subseteq \F_q^n$.}
    \Output{\textsc{Yes} if $\code_1 \cong \code_2$, \textsc{No} otherwise.}
    \hrulefill
    
    \tcp{Step 1: Code decomposition.}
    Run \cref{alg:decompose} on $\mG_1, \mG_2$ to obtain decompositions $\pi_1(\code_1)= \bigoplus_{i=1}^{m_1} A_i$ and $\pi_2(\code_2)= \bigoplus_{j=1}^{m_2} B_j$,
    for some indecomposable codes $\set{A_i}_{i=1}^{m_1}$ and $\set{B_j}_{j=1}^{m_2}$ and permutations $\pi_1,\pi_2 \in \cS_n$.\;
    
    \If{$m_1 \neq m_2$}{
        \Return \textsc{No}.\;
    }
    Let $m \defn m_1 = m_2$.\;

    \BlankLine
    \tcp{Step 2: Construction of bipartite equivalence graph.}
    Define vertex sets $U \defn \set{a_1,\ldots,a_m}$ and $V \defn \set{b_1,\ldots,b_m}$.\;
    Set $E \defn \emptyset$.\;
    \For{$i \gets 1$ \KwTo $m$}{
        $n_i \defn \text{length of code } A_i$.\;
        \For{$j \gets 1$ \KwTo $m$}{
            \BlankLine
            Run $\APART(A_i \oplus B_j)$ to obtain a set of coordinate orbits $\mathcal{O}$.

            \BlankLine
            \tcp{Find an orbit containing indices from both coord sets of $A_i$ and $B_j$.}
            \If{$\exists\ O \in \mathcal{O} \text{ with } O \cap [n_i] \neq \emptyset \neq O \setminus [n_i]$}{
                $E \gets E \cup \set{(a_i,b_j)}$.\;
            }
        }
    }
    Construct graph $G = (U \sqcup V, E)$.\;
    
    \BlankLine
    \tcp{Step 3: Find a perfect matching between summands.}
    Run the \textsc{Hopcroft-Karp Algorithm} on $G$ to obtain a maximum bipartite matching $M \subseteq E$\;
    \If{$\abs{M} = m$}{
        \Return \textsc{Yes}.\;
    }
    \Return \textsc{No}.\;
\end{algorithm}

\begin{proof}
    Let $\code_1, \code_2\subseteq \F_q^n$ be the two input codes.
    \smallskip
     
    \emph{(Correctness)} 
    In Step 1, the algorithm calls \cref{alg:decompose} to obtain a decomposition of the input codes $\code_1$ and $\code_2$ into indecomposable codes $\set{A_i}_{i=1}^{m_1}$ and $\set{B_j}_{j=1}^{m_2}$, respectively.
    By \cref{lem:Slepian}, any code can be decomposed uniquely into a direct sum of indecomposable codes, up to equivalence and permutation of the summands. 
    Hence, $\code_1 \cong \code_2$ only if $\code_1$ and $\code_2$ have the same number of summands, i.e. $m_1 = m = m_2$.
    In this case, there exists a bijection that pairs each summand $A_i$ of $\code_1$ with a summand $B_j$ of $\code_2$.
    
    In Step 2, the algorithm constructs a bipartite graph with vertices $U = \set{a_1,\ldots,a_m}$ and $V = \set{b_1,\ldots,b_m}$ representing the summands of $\code_1$ and $\code_2$, respectively.
    We claim that the edges between the vertex sets represent equivalence between the corresponding summands.
    Consider an arbitrary pair of vertices $a_i$ and $b_j$ for some $i,j \in [m]$.
    There are two cases to consider for their respective summands $A_i$ and $B_j$.
    
    Suppose that $A_i \cong B_j$.
    In this case, $A_i$ and $B_j$ must have the same length $n_i$ and there exists an isomorphism $\pi \in S_{n_i}$ that maps $\pi(A_i) = B_j$.
    Any such permutation $\pi$ yields an automorphism $\alpha_\pi \in S_{2n_i}$ of the direct sum code $A_i \oplus B_j \subseteq \F_q^{2n}$.
    This $\alpha_\pi$ acts on the index set by swapping the summands; this has the form $\alpha_\pi(A_i \oplus B_j) = \pi^{-1}(B_j) \oplus \pi(A_i)$.
    This means that there exists an orbit that crosses the boundary between the embedded index sets of $A_i$ and $B_j$ and contains an index from $\set{1,\ldots,n_i}$ and another from $\set{n_i+1,\ldots,2n_i}$
    The \Apart oracle will find this orbit and add an edge $(a_i, b_j)$ to represent the equivalence between $A_i$ and $B_j$.
    
    Conversely, suppose that $A_i \not\cong B_j$. 
    Then the direct sum $A_i \oplus B_j$ forms a decomposition consisting of two non-equivalent indecomposable codes. 
    By \cref{cor:no_partial-Aut}, any automorphism of a direct sum of indecomposable codes must map each summand to an equivalent summand. 
    Hence, no automorphism can map any coordinate of $A_i$ into the index set of $B_j$. 
    Thus, no orbit returned by \Apart will contain indices from the coordinates sets of both $A_i$ and $B_j$, so no edge $(a_i, b_j)$ will be added.
    Therefore, there is exactly one edge in $E$ between every pair of  equivalent summands.

    In Step 3, the Hopcroft-Karp algorithm finds a maximum bipartite matching \cite{HK73}.
    By \cref{lem:Slepian}, $\code_1 \cong \code_2$, if and only if there is a bijection between the $m$ summands in the codes' decomposition.
    If $\code_1 \cong \code_2$, this corresponds to a perfect matching of size $m$ in the constructed graph, and so the algorithm will output \textsc{Yes}.
    Otherwise, if $\code_1 \not\cong \code_2$, the size of the maximum matching will be less than $m$, and so the algorithm will output \textsc{No}.
    \medskip

    \emph{(Complexity)} 
    The number of indecomposable components for each input code is bounded by $m_i \leq n$. 
    If the number of summands in the decompositions match, then Step 2 will test every pair of summands from different code decompositions.
    This requires $m^2 \leq n^2$ calls to the \Apart oracle.

    Decomposing the codes in Step 1 uses \cref{alg:decompose}, which requires $O(n^3)$ time.
    For each test pair $(A_i,B_j)$ where $i,j \in [m]$, constructing the direct sum $A_i \oplus B_j$ takes $O(n_i \log(q))$ time, requiring a total of $O(n \log(q))$ time across all summands. 
    Finally, the Hopcroft-Karp algorithm is used to find a maximum bipartite matching in the graph with $\abs{V} = 2m \leq 2n$ vertices and $\abs{E} \leq m^2 \leq n^2$ edges; this requires
    $O(|E|\sqrt{|V|}) = O(n^{2.5})$ time.
    Thus, the total runtime is $O(n^3) + O(n\log(q)) + O(n^{2.5}) = O(n^3 \log(q))$ time.
\end{proof}

% ===========================================
\subsection{\LCE Reduces to \MAPART}
\label{subsec:LCE->MAPART}

Now we prove the monomial analog of \cref{prop:alg-PCE->APART}.
The reduction algorithm is nearly identical to \cref{alg:PCE->APART} but uses the properties of monomial equivalence and a \MAPART oracle instead of \APART. 
For completeness, we include the full algorithm and proof below.

\begin{prop}[\LCE Reduces to \MAPART]
    \label{prop:alg-LCE->MAPART}
    \Cref{alg:LCE->MAPART} solves \LCE for any two $n$-length $q$-ary codes in $O(n^3 \log(q))$ time using at most $n^2$ calls to a \MAPART oracle.
\end{prop}

\begin{proof}
    The proof is identical to that of \cref{alg:PCE->APART}, but replacing \cref{lem:Slepian,cor:no_partial-Aut} with the analogous results in \cref{lem:monomial_decomp,lem:no_partial}.
\end{proof}

\begin{algorithm}[th]
    \DontPrintSemicolon
    \SetKwInOut{Input}{Input}
    \SetKwInOut{Output}{Output}
    %\SetKwFunction{MAPART}{MAPART}

    \caption{Reduction from LCE to MAPART}
    \label{alg:LCE->MAPART}
    
    \Input{Generator matrices $\mG_1, \mG_2 \in \F_q^{k \times n}$ for codes $\code_1, \code_2 \subseteq \F_q^n$.}
    \Output{\textsc{Yes} if $\code_1 \cong_\cM \code_2$, \textsc{No} otherwise.}
    \hrulefill
    
    \BlankLine
    \tcp{Step 1: Code decomposition.}
    Run \cref{alg:decompose} on $\mG_1, \mG_2$ to obtain decompositions $\pi_1(\code_1)= \bigoplus_{i=1}^{m_1} A_i$ and $\pi_2(\code_2)= \bigoplus_{j=1}^{m_2} B_j$,
    for some indecomposable codes $\set{A_i}_{i=1}^{m_1}$ and $\set{B_j}_{j=1}^{m_2}$ and permutations $\pi_1,\pi_2 \in \cS_n$.\;
    
    \If{$m_1 \neq m_2$}{
        \Return \textsc{No}.\;
    }
    Let $m \defn m_1 = m_2$.\;

    \BlankLine
    \tcp{Step 2: Construction of bipartite equivalence graph.}
    Define vertex sets $U \defn \set{a_1,\ldots,a_m}$ and $V \defn \set{b_1,\ldots,b_m}$.\;
    Set $E \defn \emptyset$.\;
    \For{$i \gets 1$ \KwTo $m$}{
        $n_i \defn \text{length of code } A_i$.\;
        \For{$j \gets 1$ \KwTo $m$}{
            \BlankLine
            Run $\MAPART(A_i \oplus B_j)$ to obtain a set of monomial coordinate orbits $\cO$.

            \BlankLine
            \tcp{Find an orbit containing indices from coord sets of both $A_i$ and $B_j$.}
            \If{$\exists\ O \in \mathcal{O} \text{ with } O \cap [n_i] \neq \emptyset \neq O \setminus [n_i]$}{
                $E \gets E \cup \set{(a_i,b_j)}$.\;
            }
        }
    }
    Construct graph $G = (U \sqcup V, E)$.\;
    
    \BlankLine
    \tcp{Step 3: Find a perfect matching between summands.}
    Run the \textsc{Hopcroft-Karp Algorithm} on $G$ to obtain a maximum bipartite matching $M \subseteq E$\;
    \If{$\abs{M} = m$}{
        \Return \textsc{Yes}.\;
    }
    \Return \textsc{No}.\;
\end{algorithm}

% ===========================================
\subsection{\Apart Reduces to \PCE}
\label{subsec:APART->PCE}

The reduction from \Apart to \PCE uses a similar approach to that used in the search-to-decision \PCE reduction from \cite{BM23}.
In this algorithm, we iterate through the index set $[n]$ of the input code and use the \PCE oracle to check whether each index $i\in [n]$ can be mapped to any other index.
If so, we assign both $i$ and its image to the same orbit, and in this way partition the entire index set $[n]$ into orbits.

Let $m$ denote the maximum number of times a column appears in the generator matrix $\mG$ of the input code $\code$.
Denote the columns of $\mG$ by $\vecg_1,\ldots,\vecg_n$ and let $\mG_i$ be the matrix consisting of exactly $m$ copies of column $\vecg_i$. 

\begin{prop}
    \label{prop:alg-APART->PCE}
    \cref{alg:APART->PCE} correctly solves \Apart for any $n$-length $q$-ary code in $O(n^3 \log(q))$ time using $O(n^2)$ calls to a \PCE oracle on codes of length at most $2n$.
\end{prop}

\begin{algorithm}[th]
    \DontPrintSemicolon
    \SetKwInOut{Input}{Input}
    \SetKwInOut{Output}{Output}
    \SetKwFunction{Oracle}{PCE}

    \caption{Reduction from \Apart to \PCE}
    \label{alg:APART->PCE}
    
    \Input{A generator matrix $\mG \in \F_q^{k \times n}$ of a code $\code \subseteq \F_q^n$.}
    \Output{The set of coordinate orbits $\set{X_1,\ldots,X_\ell}$ that partition $[n]$.}
    \hrulefill
    
    \BlankLine
     $U \defn \set{1, 2, \dots, n}$. \tcp{set of unassigned indices}
    $X_1,\ldots,X_n \gets \emptyset$. \tcp{orbits of indices}
    $m \defn \text{maximum column multiplicity in } \mG$.\;
    \BlankLine
    \While{$U \neq \emptyset$}{
    % \For{$i\gets 1$ \KwTo $n$}{
        $i \gets \min\set{U}$. \tcp{smallest unassigned index}
        $X_i \gets X_i \cup \set{i}$. \tcp{maps to itself under identity automorphism}
        $U \gets U \setminus \set{i}$. \tcp{mark $i$ as assigned}     
        \BlankLine
        \For{$j \in U$}{
            $\mG_A \gets \bracks{\mG \mid \mG_i}$. \tcp{append $m$ copies of column $i$}
            $\mG_B \gets \bracks{\mG \mid \mG_j}$. \tcp{append $m$ copies of column $j$}

            \BlankLine
            \tcp{Test if $j$ is in the orbit of $i$ using the \PCE oracle.}
            \If{\Oracle{$\mG_A$, $\mG_B$} = \emph{\textsc{Yes}}}{
                $X_i \gets X_i \cup \set{j}$. \tcp{add $j$ to orbit of $i$}
                $U \gets U \setminus \set{j}$. \tcp{mark $j$ as assigned}
            }
        }
    }
    $\ell \defn$ number of non-empty sets $X_i$.\;
    \BlankLine
    \Return $X_1,\ldots,X_\ell$.
\end{algorithm}

\begin{proof}
    Let $\code \subseteq \F_q^n$ be the input code and $\mG \in \F_q^{k\times n}$ be its generator matrix.
    \smallskip

    \emph{(Correctness)}
    We show that \cref{alg:APART->PCE} partitions the index set $[n]$ into disjoint orbits $X_1, \dots, X_\ell$ such that $i,j\in X_t$ (for $t\in [\ell]$) if and only if there exists an automorphism $\alpha \in \Aut(\code)$ for which $\alpha(i) = j$.
    
    The algorithm maintains a set of unassigned indices $U$ to ensure that each index is assigned to exactly one orbit.
    This guarantees that the output collection of subsets forms a partition of $[n]$.
    In each iteration of the while-loop, the smallest unassigned index $i$ is used as a base for a new orbit set.
    Then the algorithm iterates through all remaining unassigned indices $j\in U$ to determine if $j$ belongs in the orbit of $i$.
    
    To do this, the algorithm constructs the augmented matrices $\mG_A = \bracks{\mG \mid \mG_i}$ and $\mG_B = \bracks{\mG \mid \mG_j}$ by appending exactly $m$ identical copies of column $\vecg_i$ and $\vecg_j$, respectively, where $m$ is the maximum number of times a column appears in the given generator matrix $\mG$.
    By definition, appending $m$ copies guarantees that the frequencies of $\vecg_i$ in $\mG_A$ and $\vecg_j$ in $\mG_B$ become at least $m+1$. 
    Since no other column in the base matrix appears more than $m$ times in $\mG_A$ and $\mG_B$, the columns in the appended blocks $\mG_i$ and $\mG_j$ have a unique frequency.
    By \cref{lem:col-frequency}, any permutation $\pi$ for which $\pi(\code(\mG_A)) = \code(\mG_B)$ will necessarily map all copies of $\vecg_i$ to all copies of $\vecg_j$. 
    Hence, the \PCE oracle outputs \textsc{Yes} on input $(\mG_A, \mG_B)$ if and only if there exists an automorphism of the original code $\code$ that maps index $i$ to $j$. 
    In this way, the algorithm correctly identifies all $j$ in the orbit of $i$ to define their orbit set.
    In the case where every index maps to itself, there will be $n$ singleton orbits, so up to $n$ of the sets $X_1,\ldots, X_n$ will be populated.
    The algorithm outputs the non-empty sets among these as the orbit partition.
    \medskip

    \emph{(Complexity)}
    First we bound the number of oracle calls.
    The while-loop iteratively partitions the $n$ indicess into orbits by testing a base index $i$ against all remaining unassigned indices in $U$.
    In the worst case, where the automorphism group is trivial and all orbits have size $1$, exhaustively testing each index $i$ against all remaining $n-i$ targets requires up to
    $\sum_{i=1}^{n-1} (n - i) = n(n-1)/2 = O(n^2)$
    \PCE oracle calls.

    Now we bound the runtime.
    For each test, constructing $\mG_A$ and $\mG_B$ requires appending $m \leq n$ columns to $\mG$, yielding test matrices of width at most $2n$ over the $q$-ary field $\F_q$. 
    This augmentation, along with the basic set operations for $U$ and $X_i$, takes $O(n \log(q))$ time per query. 
    Therefore, \cref{alg:APART->PCE} solves \Apart for any $n$-length $q$-ary code in $O(n^3 \log(q))$ time.
\end{proof}

% ===========================================
\subsection{\MAPART Reduces to \LCE}
\label{subsec:MAPART->LCE}

Now we prove the monomial analog of \cref{prop:alg-APART->PCE}.
The reduction algorithm is nearly identical to \cref{alg:APART->PCE} but uses the properties of monomial equivalence and an \LCE oracle instead of \PCE. 
For completeness, we include the full algorithm and proof below.

Let $m$ denote the maximum number of times any column, or any of its non-zero scalar multiples, appears in the generator matrix $\mG$ of the input code $\code$; this is called the \emph{maximum projective multiplicity} of $\mG$.
Denote the columns of $\mG$ by $\vecg_1,\ldots,\vecg_n$ and define $\mG_i$ to be the matrix consisting of exactly $m$ copies of $\vecg_i$. 

\begin{prop}
    \label{prop:alg-MAPART->LCE}
    \cref{alg:MAPART->LCE} correctly solves \MAPART for any $n$-length $q$-ary code in $O(n^3 \log(q))$ time using $O(n^2)$ calls to an \LCE oracle on codes of length at most $2n$.
\end{prop}

\begin{proof}
    The proof is identical to that of \cref{alg:APART->PCE}, but replacing \cref{lem:col-frequency} with the analogous result in \cref{lem:col-frequency-LCE}.
\end{proof}

\begin{algorithm}[th]
    \DontPrintSemicolon
    \SetKwInOut{Input}{Input}
    \SetKwInOut{Output}{Output}
    \SetKwFunction{Oracle}{LCE}

    \caption{Reduction from \MAPART to \LCE}
    \label{alg:MAPART->LCE}
    
    \Input{A generator matrix $\mG \in \F_q^{k \times n}$ of a code $\code \subseteq \F_q^n$.}
    \Output{The set of monomial index orbits $\set{X_1,\ldots,X_\ell}$ that partition $[n]$.}
    \hrulefill
    
    \BlankLine
     $U \defn \set{1, 2, \dots, n}$. \tcp{set of unassigned indices}
    $X_1,\ldots,X_n \gets \emptyset$. \tcp{monomial index orbits}
    $m \defn \text{maximum projective column multiplicity in } \mG$.
    
    \BlankLine
    \While{$U \neq \emptyset$}{
        $i \gets \min\set{U}$. \tcp{smallest unassigned coord}
        $X_i \gets X_i \cup \set{i}$. \tcp{maps to itself under identity}
        $U \gets U \setminus \set{i}$. \tcp{mark $i$ as assigned}     
        \BlankLine
        \For{$j \in U$}{
            $\mG_A \gets \bracks{\mG \mid \mG_i}$. \tcp{append $m$ copies of column $i$}
            $\mG_B \gets \bracks{\mG \mid \mG_j}$. \tcp{append $m$ copies of column $j$}

            \BlankLine
            \tcp{Test if $j$ is in the monomial orbit of $i$ using the \LCE oracle.}
            \If{\Oracle{$\mG_A$, $\mG_B$} = \emph{\textsc{Yes}}}{
                $X_i \gets X_i \cup \set{j}$. \tcp{add $j$ to spatial orbit of $i$}
                $U \gets U \setminus \set{j}$. \tcp{mark $j$ as assigned}
            }
        }
    }
    $\ell \defn$ number of non-empty sets $X_i$.\;
    \BlankLine
    \Return $X_1,\ldots,X_\ell$.
\end{algorithm}

\section{\Acount and \MAcount}%
\label{sec:ACOUNT}

In this section, we prove that determining whether two codes are equivalent is a polynomially equivalent problem to computing the cardinalities of their automorphism groups.

\begin{thm}[\PCE reduces to \ACOUNT, \LCE reduces to \MACOUNT]
    \label{thm:PCE/LCE->ACOUNT/MACOUNT}
    There is a reduction from \PCE (resp. \LCE) on any pair of $n$-length codes to \ACOUNT (resp. \MACOUNT) which runs in $O(n^3 \log(q))$ time and makes at most $3n^2$ oracle calls.
\end{thm}

\begin{thm}[\ACOUNT reduces to \PCE, \MACOUNT reduces to \LCE]
    \label{thm:ACOUNT/MACOUNT->PCE/LCE}
    There is a reduction from \ACOUNT (resp. \MACOUNT) on any $n$-length $q$-ary code to \PCE (resp. \LCE) which runs in $O(n^6 \log q)$ time and makes at most $O(n^2)$ oracle calls.
\end{thm}

We prove the reduction from \PCE in \cref{subsec:PCE->ACOUNT} and its monomial analog in \cref{subsec:LCE->MACOUNT}.
We prove the reverse reductions in \cref{subsec:ACOUNT->PCE,subsec:MACOUNT->LCE}.

% ===========================================
\subsection{\PCE Reduces to \ACOUNT}
\label{subsec:PCE->ACOUNT}

The reduction from \PCE takes the input codes and decomposes them into direct sums of indecomposable codes.
These can only be permutation equivalent if the decompositions have the same number of summands and the summands can be paired into isomorphic pairs of codes.
The algorithm identifies pairs of isomorphic summands using a ratio test, as given by \cref{lem:ratio-Aut}.

\begin{prop}
    \label{prop:PCE->ACOUNT}
    \cref{alg:PCE->ACOUNT} correctly solves \PCE for any pair of input $n$-length codes in $O(n^3 \log(q))$ time using at most $3n^2$ calls to an \ACOUNT oracle.
\end{prop}

\begin{algorithm}[th]
    \SetAlgoLined
    \DontPrintSemicolon
    \SetKwInOut{Input}{Input}
    \SetKwInOut{Output}{Output}
    \SetKwFunction{ACOUNT}{ACOUNT}
    
    \Input{Codes $\code_1, \code_2 \subseteq \F_q^n$.}
    \Output{\textsc{Yes} if $\code_1 \cong \code_2$, \textsc{No} otherwise.}
    \hrulefill
    
    \BlankLine
    \tcp{Step 1: Code decomposition.}
    Run \cref{alg:decompose} on each input code to obtain the decompositions $\code_1 \cong \bigoplus_{i=1}^m A_i$ and $\code_2 \cong \bigoplus_{j=1}^{m'} B_j$,\\
    for some $[n_{A_i},k_{A_i}]$-codes $\set{A_i}_{i=1}^m$ and $[n_{B_j},k_{B_j}]$-codes $\set{B_j}_{j=1}^{m'}$.\;
    
    \BlankLine
    \tcp{Step 2: Match summands between $\code_1$ and $\code_2$.}
    \If{$m \neq m'$}{
        \Return{\textsc{No}.} 
    }
    $I \gets \{1, 2, \dots, m\}$ 
    
   \For{$i \gets 1$ \KwTo $m$}{
    match-found $\gets \textsc{No}$\;
    \For{\textbf{each} $j \in I$}{
        \tcp{Test for isomorphism using the ACOUNT ratio test.}
        Run \ACOUNT oracle to obtain the cardinalities $|\Aut(A_i \oplus B_j)|, |\Aut(A_i)|, |\Aut(B_j)|$.

        \BlankLine
        Compute $R \defn \abs{\Aut(A_i \oplus B_j)}\ /\ \parens{\abs{\Aut(A_i)} \cdot \abs{\Aut(B_j)}}$.\;
        
        \If{$R = 2$}{
            $I \gets I \setminus \{j\}$\;
            match-found $\gets \textsc{Yes}$\;
            \textbf{break}.
        }
    }
    \If{match-found $= \textsc{No}$}{
        \Return{\textsc{No}}
    }
}
\BlankLine
\Return{\textsc{Yes}}
    
    \caption{Reduction from PCE to ACOUNT}
    \label{alg:PCE->ACOUNT}
\end{algorithm}

\begin{proof}
    Let $\code_1, \code_2\subseteq \F_q^n$ be the two input codes
    \smallskip
     
    \emph{(Correctness)} To prove correctness, we must show that the algorithm outputs \textsc{Yes} if given a pair of equivalent codes, and outputs \textsc{No} otherwise.
    Consider the two possible cases.
    \smallskip

    $(\code_1 \cong \code_2)$: 
    By \cref{lem:Slepian}, the decompositions of $\code_1$ and $\code_2$ found in Step 1 are unique up to permutation equivalence.
    Since the codes are permutation equivalent, we must have $m=m'$ and $A_i \cong B_{\sigma(i)}$ for all $i\in [m]$ for some permutation $\sigma$ on $[m]$.
    Step 2 of the algorithm will implicitly match the summands according to this unknown $\sigma$.
    Since $m=m'$, the algorithm bypasses the first check and implements the nested for-loop to test isomorphism for each (non-paired) pair of summands, $A_i$ and $B_j$.
    In this loop, the \ACOUNT oracle is called three times to obtain the relevant cardinalities and their ratio is computed.
    By \cref{lem:ratio-Aut}, the value of this ratio correctly indicates whether the pair of summands is equivalent: the value is $2$ if and only if the two summands form an equivalent pair.
    If multisets $\set{A_1,\ldots,A_m}$ and $\set{B_1,\ldots,B_m}$ are paired perfectly in this way, then \emph{match-found} = \textsc{Yes} when the algorithm terminates.
    Since the nested for-loop exhaustively searches for a valid matching, and we know such a matching $\sigma$ exists, the algorithm will output \textsc{Yes}.    
        
    $(\code_1 \not\cong \code_2)$:
    By \cref{lem:Slepian}, the decompositions of $\code_1$ and $\code_2$ found in Step 1 are unique up to permutation equivalence.
    Suppose for contradiction that the algorithm outputs \textsc{Yes} in this case.
    Then, it must have found a matching $\mu$ between the summands in the multisets $\set{A_1,\ldots,A_m}$ and $\set{B_1,\ldots,B_{m'}}$ such that $A_i \cong B_{\mu(i)}$ for all $i\in [m]$.
    But this would imply that $m=m'$ and 
    \begin{equation*}
        \code_1 \cong A_1 \oplus \ldots \oplus A_m
        \cong B_1 \oplus \ldots \oplus B_m \cong \code_2,
    \end{equation*}
    which contradicts the assumption that the codes are not permutationally equivalent.
    Therefore, the algorithm must return \textsc{No} in this case.
    \medskip

    \emph{(Complexity)}
    By \cref{lem:alg-decompose}, decomposing both $n$-length $q$-ary input codes takes $2\cdot O(n^3\log(q)) = O(n^3 \log(q))$ time.
    The matching process in Step 2 involves a nested for-loop, where the outer loop iterates over the $m$ summands, and the inner loop iterates over the remaining $\leq m$ unmatched summands.
    In the worst case, this requires $m + (m-1) + \ldots + 1 = m(m+1)/2$ iterations. 
    The \ACOUNT oracle is called at most three times per iteration, so at most $3\cdot m(m+1)/2 \leq 3m^2 \leq 3n^2$ oracle calls are made in total.
    Constructing the direct sum code in each iteration requires $O((n_{A_i}+n_{B_j}) \log(q)) = O(n \log(q))$ time, amounting to a total of $O(m^2 n \log(q)) = O(n^3\log(q))$ time across the entire algorithm.
    Therefore, the total runtime is bounded by $O(n^3 \log(q))$.
\end{proof}

% ===========================================
\subsection{\LCE Reduces to \MACOUNT}
\label{subsec:LCE->MACOUNT}

Now we prove the monomial analog of \cref{prop:PCE->ACOUNT}.
The reduction algorithm is nearly identical to \cref{alg:PCE->ACOUNT} but uses the properties of monomial equivalence and a \MACOUNT oracle instead of \ACOUNT.
For completeness, we include the algorithm and its proof here.

\begin{prop}
    \label{prop:LCE->MACOUNT}
    \cref{alg:LCE->MACOUNT} correctly solves \LCE for any pair of input $n$-length $q$-ary codes in $O(n^3 \log(q))$ time, using at most $3n^2$ calls to an \MACOUNT oracle.
\end{prop}

\begin{algorithm}[H]
    \SetAlgoLined
    \DontPrintSemicolon
    \SetKwInOut{Input}{Input}
    \SetKwInOut{Output}{Output}
    \SetKwFunction{ACOUNT}{ACOUNT}
    
    \Input{Codes $\code_1, \code_2 \subseteq \F_q^n$.}
    \Output{\textsc{Yes} if $\code_1 \cong_\cM \code_2$, \textsc{No} otherwise.}
    \hrulefill
    
    % \BlankLine
    \tcp{Step 1: Code decomposition.}
    Run \cref{alg:decompose} on each input code to obtain the decompositions $\code_1 \cong \bigoplus_{i=1}^m A_i$ and $\code_2 \cong \bigoplus_{j=1}^{m'} B_j$,\\
    for some $[n_{A_i},k_{A_i}]$-codes $\set{A_i}_{i=1}^m$ and $[n_{B_j},k_{B_j}]$-codes $\set{B_j}_{j=1}^{m'}$.\;
    
    \BlankLine
    \tcp{Step 2: Match summands between $\code_1$ and $\code_2$.}
    \If{$m \neq m'$}{
        \Return{\textsc{No}.} 
    }
    $I \gets \{1, 2, \dots, m\}$ 
    
    \For{$i \gets 1$ \KwTo $m$}{
        match-found $\gets \textsc{No}$.\;
        \For{\textbf{each} $j \in I$}{
            \tcp{Test for isomorphism using the \MACOUNT ratio test.}
            Run \MACOUNT oracle to obtain the cardinalities $|\MAut(A_i \oplus B_j)|, |\MAut(A_i)|, |\MAut(B_j)|$.

            \BlankLine
            Compute $R \defn \abs{\MAut(A_i \oplus B_j)}\ /\ \parens{\abs{\MAut(A_i)} \cdot \abs{\MAut(B_j)}}$.\;
            
            \If{$R = 2$}{
                $I \gets I \setminus \{j\}$ \;
                match-found $\gets \textsc{Yes}$\;
                \textbf{break} 
            }
        }
        \If{match-found $= \textsc{No}$}{
        \Return{\textsc{No}}
    }
    }
    \BlankLine
    \Return{\textsc{Yes}}
    
    \caption{Reduction from LCE to MACOUNT}
    \label{alg:LCE->MACOUNT}
\end{algorithm}

\begin{proof}
    The proof is identical to that of \cref{prop:PCE->ACOUNT}, but replacing \cref{lem:Slepian,lem:ratio-Aut} with the analogous results in \cref{lem:monomial_decomp,lem:ratio}.
\end{proof}

% ===========================================
\subsection{\Acount to \PCE}
\label{subsec:ACOUNT->PCE}

For the reduction from \Acount to \PCE, we use a setwise stabilizer chain for the permutation automorphism group of the input code, as defined in \cref{def:setwise-stab-chain}.
For any code $\code$ with generator matrix $\mG$, let $G \defn \Aut(\code)$ denote its permutation automorphism group and let $G = G^{(E_0)} \geq G^{(E_1)} \geq \dots \geq G^{(E_s)}$ be the setwise stabilizer chain given by $\mG$.
To compute the cardinality $\abs{G}$, we use the expression from \cref{lem:setwise-stab-chain-order},
which states
 \[
    \abs{G} = \prod_{i=1}^{s} \abs[\big]{G^{(E_{i-1})} : G^{(E_i)}} \cdot \parens{m_i!}.
\]
To compute $\abs[\big]{G^{(E_{i-1})} : G^{(E_i)}}$ for every $1\leq i \leq s$, we count the number automorphisms that map indices in $E_i$ to some other $E_j$ under the condition that all indices in $E_0,\ldots,E_{i-1}$ are mapped back into their respective sets.
To count all such automorphisms, we iterate through every index $j$ in the orbit of $i$ and test whether there is an automorphism that maps $i$ to $j$ using the \PCE oracle.
This technique loosely resembles that used in the search-to-decision \PCE reduction from \cite{BM23}, which uses a frequency argument to determine the existence of a permutation map between pairs of indices.
The key difference in our algorithm is that that all previous indices are fixed and mapped to themselves in each iteration.
The algorithm relies on the crucial fact that two generator matrices can generate equivalent codes only if they have the same frequency distribution of columns.

\begin{prop}
    \label{prop:alg-ACOUNT->PCE}
    \cref{alg:ACOUNT->PCE} correctly solves \Acount for any $q$-ary $n$-length code in $O(n^6\log(q))$ time using at most $O(n^2)$ calls to a \PCE oracle.
\end{prop}

\begin{algorithm}[th]
    \DontPrintSemicolon
    \SetKwInOut{Input}{Input}
    \SetKwInOut{Output}{Output}
    \SetKwFunction{Oracle}{PCE}

    \caption{Reduction from ACOUNT to PCE}
    \label{alg:ACOUNT->PCE}
    
    \Input{A generator matrix $\mG \in \F_q^{k \times n}$ of a linear code $\code \subseteq \F_q^n$.}
    \Output{The order of the automorphism group $\abs{\Aut(\code)} \in \N$.}
    \hrulefill
    
    \BlankLine
    \tcp{Partition index set into classes of identical columns and construct set of unique column representatives.}
    Let $\set{\vecg_1,\ldots,\vecg_n}$ be the set of column vectors of $\mG$.\;
    $\cE \defn \emptyset$. \tcp{collection of redundancy sets}
    $\cG \defn \emptyset$. \tcp{keep track of unique columns}
    $B' \defn \emptyset$. \tcp{indices of first appearances}

    \BlankLine
    \For{$i \gets 1$ \KwTo $n$}{
    \tcp{If a column matches a previously seen column, sort it into the same class.}
        \If{$\vecg_i \in \cG$}{
            \For{$E\in \cE$}{
                $\vecv \gets E[1]$.\;
                \If{$\vecg_i = \vecv$}{
                    $E \gets E \cup \set{\vecg_i}$.\;
                }
            }
        }
        \tcp{Otherwise, if a column is new, create a new class.}
        \Else{
            $E \defn \set{\vecg_i}$. \tcp{if }
            $\cG \gets \cG \cup \set{\vecg_i}$\;
            $\cE \gets \cE \cup \set{E}$\;
            $B' \gets B' \cup \set{i}$. \tcp{record index of first appearance}
        }
    }
    $m \gets 0$. \tcp{maximum column multiplicity in $\mG$}
    $N \gets 1$. \tcp{initialize the automorphism count} % (start with the trivial one)}
    \tcp{compute the maximum column multiplicity}
    \For{$E \in \cE$}{
        $N \gets N \cdot |E|!$\;
        \If{$|E| > m$}{
            $m \gets |E|$\;
        }
    }
    \BlankLine
    \tcp{Find orbit partition of index set via \Apart oracle.}
    Run $\APART(\mG)$ to obtain coordinate orbits $\set{X_1,\ldots,X_\ell}$.\;

    \BlankLine
    $\cG \gets \emptyset$. \tcp{reset to track fixed columns in stabilizer chain}
    $\Gpin \gets \mG$. \tcp{initialize global base matrix}
    \emph{(continued on next page)}   
\end{algorithm}
% manual page break
\begin{algorithm}[th]
    \DontPrintSemicolon
    \emph{(continued from previous page)} 
    
    \tcp{Compute $\abs[\big]{G^{(E_{i-1})} : G^{(E_i)}}$ for each consecutive pair of subgroups in the chain.}
    \For{$i \in B'$}{
        $X \defn$ unique orbit in $\set{X_1,\ldots,X_\ell}$ containing $i$.\;
        $d_i \gets 1$. \tcp{initialize local orbit count}
        
        \BlankLine
        \tcp{Check if an automorphism maps column $i$ to non-identical column $j$.}
        \For{$j \in X \cap B'$ \emph{\textbf{where}} $j > i$}{
            $l \gets |\cG|$. \tcp{number of unique classes fixed so far}
            $\mA \defn (l+1)m$ copies of $\vecg_i$.\;
            $\mB \defn (l+1)m$ copies of $\vecg_j$.\;
            $\mG_A \gets \bracks{\Gpin \mid \mA}$.\;
            $\mG_B \gets \bracks{\Gpin \mid \mB}$.\;
            
            \BlankLine           
            \If{\Oracle{$\mG_A$, $\mG_B$} = \emph{\textsc{Yes}}}{
                $d_i \gets d_i + 1$.
            }
        }
        
        \BlankLine
        $N \gets N \cdot d_i$. \tcp{update automorphism count}
        
        \BlankLine
        \tcp{Stabilize set $E_i$ for all future iterations.}
        $\mA \defn (l+1)m$ copies of $\vecg_i$.\;
        $\Gpin \gets \bracks{\Gpin \mid \mA}$.\;
        $\cG \gets \cG \cup \set{\vecg_i}$.\;
    }
    
    \BlankLine
    \Return{$N$}.
\end{algorithm}

\begin{proof}
    Let $\mG \in \F_q^{k \times n}$ be the input generator matrix of the code $\code$, and $G = \Aut(\code)$ denote its automorphism group. 
    Let $\set{\vecg_1,\ldots,\vecg_n}$ denote the columns of $\mG$ in order of appearance.
    We prove that \cref{alg:ACOUNT->PCE} correctly outputs $\abs{G}$ by traversing a setwise stabilizer chain for $G$ and applying \cref{lem:setwise-stab-chain-order}.
    \smallskip

    \emph{(Correctness)}   
    The algorithm first partitions the columns of $\mG$ into disjoint redundancy classes $\cE = \set{E_1, \dots, E_s}$, where each $E_i$ contains identical columns of a particular type, and no two classes contain the same vector.
    For this sorting process, the algorithm identifies the indices of the unique column vectors in $\mG$ in order of first appearance; let $B' = \set{t_1,\ldots,t_s} \subseteq [n]$ denote the elements of this subset.
    Note that each index $t_i \in B'$ defines the redundancy class $E_i \in \cE$.
    By \cref{def:setwise-stab-chain}, this $B'$ defines a setwise stabilizer chain $G \geq G^{(E_1)} \geq \dots \geq G^{(E_s)}$.
    
    In the next for-loop, the algorithm initializes the group order counter $N = \prod_{E \in \cE} \abs{E}!$. 
    By \cref{lem:setwise-stab-chain-order}, this value corresponds to the order of the smallest subgroup $G^{(E_s)}$ in the chain. 
    It represents the core of automorphisms that permute identical columns within their respective classes in $\cE$.

    Next, the algorithm makes a single oracle call to an \Apart oracle to obtain a partition of the index set $[n]$ into orbits under automorphisms in $G$.
    While not strictly necessary, this step provides an optimization that restricts the search space for automorphisms in the last step.
    
    To compute the remaining number of automorphisms in $G$, the algorithm traverses the chain $G \geq G^{(E_1)} \geq \dots \geq G^{(E_s)}$ in descending order by iterating through the elements of $B'$.
    In the $t_i$-th iteration, we compute the quotient index $\abs{G^{(E_{i-1})} : G^{(E_i)}}$ by setwise stabilizing the columns in $E_{i-1}$ and searching for automorphisms that map $E_i$ to some $E_j$ for all larger indices $t_j \in B'$ in the orbit $X$ of $t_i$.
    By testing only the unique representatives of each redundancy class, the algorithm evaluates exactly one target per class. 
    This prevents any redundant counting for identical columns within a target class, as the permutations among columns of a class are already counted by the initial $N = \prod_{i=1}^s \abs{E_i}!$. 
    
    Any automorphism that maps $E_i$ to $E_j$ maps all copies of column $\vecg_{t_i}$ to all copies of column $\vecg_{t_j}$, where $t_i,t_j\in B'$.
    To test for the existence of such automorphisms, we use a similar technique to that in \cite{BM23,CSV25} and augment the matrix $\mG$ with copies of columns to ensure unique frequency.
    In particular, we construct test matrices $\mG_A$ and $\mG_B$ by appending blocks of $(l+1)m$ copies of $\vecg_{t_i}$ and $(l+1)m$ copies of $\vecg_{t_j}$, respectively, to the current stabilized matrix in which column classes $E_1,\ldots,E_{i-1}$ are fixed setwise.
    By definition of $m$ and $l$, no other column in $\mG_A$ and $\mG_B$ appears the same number of times, so $\vecg_{t_i}$ and $\vecg_{t_j}$ have the same unique frequency.
    By \cref{lem:col-frequency}, the \PCE oracle can only respond \textsc{Yes} if there is an automorphism that maps $E_i$ to $E_j$,
    Thus, the oracle is forced to decide equivalence exclusively within the stabilizer subgroup $G^{(E_{i-1})}$.
    The counter $d_i$ increments for each target class found, beginning from $1$ which counts the trivial map of $E_i$ to itself.
    Hence, once the inner for-loop is completed, $d_i$ represents the setwise orbit size, i.e. $d_i = \abs{G^{(E_{i-1})} : G^{(E_i)}}$.

    At the end of each iteration $t_i \in B'$, the running total count $N$ is multiplied by $d_i$. 
    Because the algorithm iterates through every unique class representative in $B'$, it calculates the product of all quotient indices.
    Then, by \cref{lem:setwise-stab-chain-order}, multiplying the initial kernel order $\prod_{i=1}^s \abs{E_i}!$ by these indices yields the order of the full automorphism group
    \begin{equation*}
        N = \prod_{i=1}^{s} \abs{E_i}! \cdot d_{i}
        % \parens{\prod_{i=1}^{s} \abs{E_i}!} \cdot \prod_{i \in B'} d_i 
        = \abs{G^{(E_s)}} \cdot \prod_{i=1}^{s} \abs[\big]{G^{(E_{i-1})} : G^{(E_i)}} 
        = \abs{G}.
    \end{equation*}
    % Thus, the algorithm correctly returns the exact order of the full automorphism group. 
    
    \emph{(Complexity)}
    We first bound the number of oracle calls. 
    The algorithm makes a single query to the \Apart oracle to obtain the orbit partition of the index set $[n]$.
    As shown in \cref{thm:APART/MAPART->PCE/LCE}, \Apart can be solved using $O(n^2)$ queries to a \PCE oracle. 
    In the main for-loop, a target $t_j \in B'$ is tested against a base index $t_i \in B'$ only if both belong to the same orbit $X$. 
    For any orbit $X$, the number of distinct column classes is bounded by $|X|$. 
    Thus, the outer for-loop evaluates the test branch at most $|X|$ times for this orbit, and each inner for-loop tests at most $|X|-1$ targets. 
    This limits the number of \PCE oracle queries to $O(|X|^2)$ for each orbit $X$.
    Because the coordinate orbits form a partition of $[n]$, the sum of $O(|X|^2)$ across all orbits $X$ is strictly bounded above by $O(n^2)$. 
    Therefore, the algorithm makes a total of % 2\cdot O(n^2)
    $O(n^2)$ calls to the \PCE oracle.

    Now we bound the runtime.
    Partitioning the index set $[n]$ into redundancy classes and constructing the set $B'$ requires comparing column vectors at most $n^2$ times, so this step requires $O(n^2 \log(n))$ time.
    By \cref{prop:alg-APART->PCE}, running the \Apart oracle requires $O(n^3 \log(q))$ time.
    Now we evaluate the cost of constructing the test matrices, which dominates the time complexity of the algorithm.
    The maximum column multiplicity in $\mG$ is $m \leq n$, and in each iteration $t_i \in B'$, the number of unique column classes fixed so far is $l \leq n$. 
    The base matrix $\Gpin$ of setwise stabilized columns accumulates $\sum_{a=1}^l a\cdot m$ appended columns in each iteration.
    To construct $\mG_A$ and $\mG_B$, an additional block of $(l+1)m$ columns is appended to $\Gpin$.
    Thus, the number of columns appended to $\mG$ is bounded above by $n + m l(l+1)/2 + (l+1)m = O(n^3)$. 

    Each field element in $\F_q$ requires $\log(q)$ bits to represent and $O(\log(q))$ bit operations to read, copy, or compare. 
    The initial partitioning phase and computing the kernel size requires $O(n^2\log(q))$ and $O(n)$ operations respectively.
    Then, for each of the $O(n^2)$ oracle queries, copying at most $O(n^3)$ column vectors for the augmented matrices requires $O(k n^3\log(q))$ operations.
    Therefore, since $k \leq n$, the total runtime is bounded by $O(n^2) \cdot O(kn^3\log(q)) = O(n^6\log(q))$ time.
\end{proof}

% ===========================================
\subsection{\MACOUNT to \LCE}
\label{subsec:MACOUNT->LCE}

Now we prove the monomial analog of \cref{prop:alg-ACOUNT->PCE}.
The reduction algorithm is nearly identical to \cref{alg:ACOUNT->PCE} but uses the properties of monomial equivalence and a \LCE oracle instead of \PCE. 
For completeness, we include the full algorithm and proof below.

\begin{prop}
    \label{prop:alg-MACOUNT->LCE}
    \cref{alg:MACOUNT->LCE} correctly solves \MACOUNT for any $q$-ary $n$-length code in $O(n^6\log(q))$ time using at most $O(n^2)$ calls to an \LCE oracle.
\end{prop}

\begin{proof}
    \cref{alg:MACOUNT->LCE} is nearly identical to \cref{alg:ACOUNT->PCE} with the exception of an additional step to find the number $c$ of indecomposable summands via \cref{alg:decompose}.
    The proof of correctness is identical to that of \cref{prop:alg-ACOUNT->PCE}, but replacing \cref{lem:col-frequency,lem:setwise-stab-chain-order} with their monomial analogs \cref{lem:col-frequency-LCE,lem:mon-setwise-stab-chain-order}.

    The time complexity remains $O(n^6\log(q))$ because the $O(n^3\log(q))$ runtime of Step 1 is dominated by the cost of constructing the augmented test matrices in Step 2.
\end{proof}

\begin{algorithm}[th]
    \DontPrintSemicolon
    \SetKwInOut{Input}{Input}
    \SetKwInOut{Output}{Output}
    \SetKwFunction{Oracle}{LCE}
    \SetKwFunction{MAPART}{MAPART}

    \caption{Reduction from MACOUNT to LCE}
    \label{alg:MACOUNT->LCE}
    
    \Input{A generator matrix $\mG \in \F_q^{k \times n}$ of a linear code $\code \subseteq \F_q^n$.}
    \Output{The order of the monomial automorphism group $\abs{\MAut(\code)} \in \N$.}
    \hrulefill
    
    \BlankLine
    \tcp{Step 1: Code decomposition to find the number of independent summands.}
    Run \cref{alg:decompose} on $\mG$ to obtain the decomposition $\code \cong_\cM \bigoplus_{r=1}^c A_r$.\;
    $c \defn$ number of indecomposable summands.\;

    \BlankLine
    \tcp{Step 2: Partition index set into projective redundancy classes.}
    Let $\set{\vecg_1,\ldots,\vecg_n}$ be the set of column vectors of $\mG$.\;
    $\cE \defn \emptyset$. \tcp{collection of redundancy sets}
    $\cG \defn \emptyset$. \tcp{keep track of unique projective column representatives}
    $B' \defn \emptyset$. \tcp{indices of first appearances}

    \BlankLine
    \For{$i \gets 1$ \KwTo $n$}{
        \tcp{If a column is a scalar multiple of a previously seen column, sort it into the same class.}
        \If{$\exists\ \vecv \in \cG$ \emph{\textbf{such that}} $\vecg_i = \lambda \vecv$ \emph{for some} $\lambda \in \F_q^*$}{
            \For{$E\in \cE$}{
                $\vecv \gets E[1]$.\;
                \If{$\vecg_i = \lambda \vecv$ \emph{for some} $\lambda \in \F_q^*$}{
                    $E \gets E \cup \set{\vecg_i}$.\;
                }
            }
        }
        \tcp{Otherwise, if a column is unique up to scaling, create a new class.}
        \Else{
            $E \defn \set{\vecg_i}$.\;
            $\cG \gets \cG \cup \set{\vecg_i}$\;
            $\cE \gets \cE \cup \set{E}$\;
            $B' \gets B' \cup \set{i}$. \tcp{record index of first appearance}
        }
    }
    
    \BlankLine
    $m \gets 0$. \tcp{maximum projective column multiplicity in $\mG$}
    $N \gets (q-1)^c$. \tcp{initialize count with the global scaling automorphisms}
    \tcp{count the permutations stabilizing the projective classes}
    \For{$E \in \cE$}{
        $N \gets N \cdot |E|!$\;
        \If{$|E| > m$}{
            $m \gets |E|$\;
        }
    }
    
    \BlankLine
    \tcp{Step 3: Find the monomial orbit partition of the index set via \MAPART oracle.}
    Run $\MAPART(\mG)$ to obtain monomial coordinate orbits $\set{X_1,\ldots,X_\ell}$.\;

    \BlankLine
    $\cG \gets \emptyset$. \tcp{reset to track fixed columns in stabilizer chain}
    $\Gpin \gets \mG$. \tcp{initialize global base matrix}

    \emph{(continued on next page)}   
\end{algorithm}

% manual page break
\begin{algorithm}[H]
    \ContinuedFloat
    \DontPrintSemicolon
    \emph{(continued from previous page)}   
    
    \BlankLine
    \tcp{Step 4: Compute index $\abs[\big]{G^{(E_{i-1})} : G^{(E_i)}}$ for each consecutive pair of subgroups in the chain.}
    \For{$i \in B'$}{
        $X \defn$ unique orbit in $\set{X_1,\ldots,X_\ell}$ containing $i$.\;
        $d_i \gets 1$. \tcp{initialize local orbit count}
        
        \BlankLine
        \tcp{Check if an automorphism maps column $i$ to a valid spatial target $j$.}
        \For{$j \in X \cap B'$ \emph{\textbf{where}} $j > i$}{
            $l \gets |\cG|$. \tcp{number of unique projective classes fixed so far}
            $\mA \defn (l+1)m$ copies of $\vecg_i$.\;
            $\mB \defn (l+1)m$ copies of $\vecg_j$.\;
            $\mG_A \gets \bracks{\Gpin \mid \mA}$.\;
            $\mG_B \gets \bracks{\Gpin \mid \mB}$.\;
            
            \BlankLine           
            % \tcp{Use LCE oracle to respect scalar transformations during the check.}
            \tcp{Use the \LCE oracle to check for scalar transformations that fix projective class of $\vecg_i$.}
            \If{\Oracle{$\mG_A$, $\mG_B$} = \emph{\textsc{Yes}}}{
                $d_i \gets d_i + 1$.\;
            }
        }
        
        \BlankLine
        $N \gets N \cdot d_i$. \tcp{update total monomial automorphism count}
        
        \BlankLine
        \tcp{Stabilize projective class $E_i$ for all future iterations.}
        $\mA \defn (l+1)m$ copies of $\vecg_i$.\;
        $\Gpin \gets \bracks{\Gpin \mid \mA}$.\;
        $\cG \gets \cG \cup \set{\vecg_i}$.\;
    }
    
    \BlankLine
    \Return{$N$}.
\end{algorithm}

% ===========================================
\subsection{\ICOUNT and \MICOUNT}%
\label{subsec:(M)ICOUNT}

Here we present the isomorphism analog of the reductions between code equivalence and the automorphism counting problem.
\ICOUNT (resp. \MICOUNT) is the problem of counting the number of (resp. monomial) isomorphisms that exist between two input codes (see \cref{def:(M)ICOUNT}).

\begin{thm}[\PCE Reduces to \ICOUNT, \LCE Reduces to \MICOUNT]
    \label{thm:PCE/LCE->ICOUNT/MICOUNT}
    There is a reduction from \PCE (resp. \LCE) on any pair of $n$-length $q$-ary codes to \ICOUNT (resp. \MICOUNT) that runs in $O(n^2\log(q))$ time and makes a single oracle call.
\end{thm}

\begin{proof}
    Given generator matrices $\mG_1,\mG_2 \in \F_q^{k\times n}$ for $\code_1,\code_2 \subseteq \F_q^n$, make one oracle call $\ICOUNT(\code_1,\code_2)$.
    If the output is $0$, then the given codes are not equivalent, and if it outputs a positive value, then the given codes are equivalent.
    An identical reduction can be made in which a \MICOUNT oracle is called once to solve \LCE.
    Copying the two $k\times n$ generator matrices over $\F_q$ requires $2nk \log(q) = O(n^2 \log(q))$ time.
\end{proof}

For the reverse reductions from \ICOUNT to \PCE (and \MICOUNT to \LCE), we use a similar approach as that in \cite{Mathon79} and rely on the relationship between (monomial) isomorphisms and automorphisms described in \cref{lem:iso=aut,lem:isom=maut}.

\begin{cor}[\ICOUNT\ Reduces to \PCE, \MICOUNT\ Reduces to \LCE]
    \label{cor:ICOUNT/MICOUNT->PCE/LCE}
    There is a reduction from \ICOUNT (resp. \MICOUNT) on any pair of $n$-length $q$-ary codes to \PCE (resp. \LCE) that runs in $O(n^6 \log(q))$ time and makes $O(n^2)$ oracle calls.
\end{cor}

\begin{proof}
    Let $\mG_1$ and $\mG_2$ be the input generator matrices for codes $\code_1$ and $\code_2$.
    The reduction algorithm first makes a single call to the \PCE oracle to check whether $\code_1 \cong \code_2$. 
    If $\PCE(\mG_1,\mG_2)$ outputs \textsc{No}, then no isomorphism exists between the codes and $\Iso(\code_1,\code_2) = 0$.
    In this case, the reduction algorithm simply outputs $0$. 
    If $\PCE(\mG_1,\mG_2)$ outputs \textsc{Yes}, then there exists at least one isomorphism from $\code_1$ to $\code_2$.
    By \cref{lem:iso=aut}, we have $\abs{\Iso(\code_1, \code_2)} = \abs{\Aut(\code_2)}$. 
    To find this cardinality, we run \cref{alg:ACOUNT->PCE} on $\mG_2$ to obtain $\abs{\Aut(\code_2)}$.

    Reading and copying the two $k\times n$ matrices over $\F_q$ requires $2nk\log(q) = O(n^2 \log(q))$ time.
    Checking equivalence requires $1$ oracle call to \PCE.
    In the event that the \PCE oracle outputs \textsc{Yes}, the algorithm runs \cref{alg:ACOUNT->PCE}, which requires $O(n^2)$ additional oracle calls to \PCE and runs in $O(n^3 \log(q))$ time.
    Thus, in the worst case, the reduction algorithm requires $1 + O(n^2) = O(n^2)$ oracle calls to \PCE and $O(n^3 \log(q))$ time.
    \medskip

    A nearly identical reduction can be made for \MICOUNT using an \LCE oracle.
    The \LCE oracle is called once to check whether $\code_1 \cong_\cM \code_2$.
    If $\LCE(\mG_1,\mG_2)$ outputs \textsc{No}, then we output $0$. 
    Otherwise, if $\LCE(\mG_1,\mG_2)$ outputs \textsc{Yes}, \cref{lem:isom=maut} guarantees that $\abs{\Isom(\code_1, \code_2)} = \abs{\MAut(\code_2)}$. 
    We run \cref{alg:MACOUNT->LCE} on $\mG_2$ to compute this cardinality. 
    As with the permutation case, this requires $1 + O(n^2) = O(n^2)$ oracle calls to \LCE and $O(n^6 \log(q))$ time, the bottleneck being the runtime of \cref{alg:MACOUNT->LCE}.
\end{proof}

%%% Local Variables:
%%% mode: latex
%%% TeX-master: "main"
%%% End:

\section{\AGEN and \MAGEN}%
\label{sec:AGEN}

\begin{thm}[\sPCE Reduces to \AGEN, \sLCE Reduces to \MAGEN]
    \label{thm:search-PCE/LCE->AGEN/MAGEN}
    There is a reduction from \sPCE (resp. \sLCE) on any pair of $n$-length $q$-ary codes $\code_1, \code_2$ to \AGEN (resp. \MAGEN) on their direct sum $\code \defn \code_1 \oplus \code_2$ that makes a single oracle call and runs in $O(s n^2 \log(q) + n^4)$ time, where $s$ is the size of the generating set found by the oracle.
\end{thm}

\begin{thm}[\AGEN Reduces to \PCE, \MAGEN Reduces to \LCE]
    \label{thm:AGEN/MAGEN->PCE/LCE}
    There is a reduction from \AGEN (resp. \MAGEN) on any $n$-length $q$-ary code that runs in $O(n^6\log(q))$ time and makes $O(n^4)$ oracle calls to \PCE (resp. \LCE).
\end{thm}

% ===========================================
\subsection{Search-\PCE Reduces to \AGEN}
\label{subsec:sPCE->AGEN}

In this section, we prove that finding an isomorphism between two codes, if it exists, is a polynomially equivalent problem to that of finding a generating set for the automorphism group of their direct sum code.

At a high-level, the idea for the reduction is similar to that in \cref{alg:PCE->ACOUNT} for reducing \PCE to \ACOUNT.
This reduction algorithm also decomposes the two input codes $\code_1, \code_2$ into direct sums of indecomposable codes and searches for pairs of isomorphic summands.
In this case, the algorithm finds a generating set for the direct sum code $\code_1 \oplus \code_2$ using the \AGEN oracle.
Every automorphism for $\code_1 \oplus \code_2$ corresponds to an isomorphism between a summand of $\code_1$ and a summand of $\code_2$.
If there is a bijection between the summands of $\code_1$ and the summands of $\code_2$, then the algorithm constructs an isomorphism between $\code_1$ and $\code_2$.

\begin{prop}
    \label{prop:alg-sPCE->AGEN}
    \Cref{alg:search-PCE->AGEN} correctly solves \sPCE for any pair of $n$-length codes in $O(s n^2 \log(q) + n^4)$ time using a single call to an \AGEN oracle,
    where $s$ is the size of the generating set found by the oracle.
\end{prop}

% \anote{Thought for post-submission: Can we modify this algorithm to ensure that the output generating set is an SGS?}

\begin{algorithm}[hb]
    \DontPrintSemicolon
    \SetKwInOut{Input}{Input}
    \SetKwInOut{Output}{Output}

    \caption{Reduction from \sPCE to \AGEN}
    \label{alg:search-PCE->AGEN}
    
    \Input{Generator matrices $\mG_1, \mG_2 \in \F_q^{k \times n}$ for codes $\code_1, \code_2 \subseteq \F_q^n$.}
    \Output{An isomorphism $\sigma \in \cS_n$ such that $\sigma(\code_1) = \code_2$.}
    \hrulefill
    
    \BlankLine
    \tcp{Step 1: Code Decomposition}
    Run \cref{alg:decompose} on $\mG_1, \mG_2$ to obtain decompositions $\pi_1(\code_1)= \bigoplus_{i=1}^{m_1} A_i$ and $\pi_2(\code_2)= \bigoplus_{j=1}^{m_2} B_j$,
    for some indecomposable codes $\set{A_i}_{i=1}^{m_1}$ and $\set{B_j}_{j=1}^{m_2}$ and permutations $\pi_1,\pi_2 \in \cS_n$.\;
    
    \If{$m_1 \neq m_2$}{
        \Return $\bot$.
    }
    Let $m \defn m_1 = m_2$.\;
    Denote the index sets of $\code_1$ and $\code_2$ by $\set{X_i}_{i=1}^m$ and $\set{Y_j}_{j=1}^m$.

    \BlankLine
    \tcp{Step 2: \AGEN Oracle Call}
    Construct the direct sum code $\code \defn \pi_1(\code_1) \oplus \pi_2(\code_2$).\;
    Run \AGEN on $\code$ to obtain a generating set $S = \set{\phi_1,\ldots, \phi_s}$ of $\Aut(\code)$.\;

    \BlankLine
    \tcp{Step 3: Construction of the Equivalence Graph}
    Define the vertex sets $U \defn \set{a_1,\ldots,a_m}$ and $V \defn \set{b_1,\ldots,b_m}$.\;
    Set $E' \defn \emptyset$.\;
    Set $maps \defn \set{}$. \tcp{dictionary of equivalence maps}

    \BlankLine
    \For{$\phi \in S$}{
        \For{$i \gets 1$ \KwTo $m$}{
            $Im(A_i) \defn \phi\mid_{X_i}(A_i)$.\;
            $Im(B_i) \defn \phi\mid_{Y_i}(B_i)$.\;

            \BlankLine
            \tcp{Find where $a_i$ maps to.}
            \For{$j \gets 1$ \KwTo $m$}{
                \If{$Im(A_i) = B_j$}{
                    $E' \gets E' \cup \set{(a_i,b_j)}$. \tcp{add directed edge}
                    $maps[(a_i,b_j)] \defn \phi$. \tcp{update edge label}
                }
                \ElseIf{$Im(A_i) = A_j$ and $j\neq i$}{
                    $E' \gets E' \cup \set{(a_i,a_j)}$.\tcp{add directed edge}
                    $maps[(a_i,a_j)] \defn \phi$. \tcp{update edge label}
                }
            }
            \tcp{Find where $b_i$ maps to.}
            \For{$j \gets 1$ \KwTo $m$}{
                \If{$Im(B_i) = A_j$}{
                    $E' \gets E' \cup \set{(b_i,a_j)}$.\tcp{add directed edge}
                    $maps[(b_i,a_j)] \defn \phi$. \tcp{update edge label}
                }
                \ElseIf{$Im(B_i) = B_j$ and $j\neq i$}{
                    $E' \gets E' \cup \set{(b_i,b_j)}$.\tcp{add directed edge}
                    $maps[(b_i,b_j)] \defn \phi$. \tcp{update edge label}
                }
            }
        }
    }
    Define graph $G' \defn (U \cup V, E')$.\;

    \emph{(continued on next page)}   
\end{algorithm}
% manual page break
\begin{algorithm}
    \DontPrintSemicolon
    \emph{(continued from previous page)}    

    \BlankLine
    \tcp{Step 4: Construction of the Bipartite Graph}
     Set $E \defn \emptyset$.\;
     Set $equivmaps \defn \set{}$. \tcp{edge labels}
        Let $G_{a_i}$ denote the connected component in $G'$ that contains $a_i$.\;
        Run \textsc{Dijkstra's Algorithm} on $G'$ from root $a_i$ to obtain paths $\set{P_v \subseteq E' \vcentcolon v\in G_{a_i}}$ to all vertices in its connected component $G_{a_i}$.

        \BlankLine
        \For{$v \in G_{a_i} \cap V$}{
        $j \gets V.\mathrm{index}(v)$.\tcp{$v=b_j$ for some $j\in[m]$}
        
        $\psi_{i,j} \gets \id$.\tcp{global automorphism}

        \For{$e \in P_v$}{
            $\phi_e \gets \mathrm{maps}[e]$.\tcp{get directed edge label}
            $\psi_{i,j} \gets \phi_e \circ \psi_{i,j}$.\;
        }

        $E \gets E \cup \{(a_i,b_j)\}$.\;
        $\psi_{i,j} \gets \psi_{i,j}|_{X_i}$.\tcp{restrict to support of $A_i$}
        $\mathrm{equivmaps}[(a_i,b_j)] \gets \psi_{i,j}$.\tcp{update edge label}
    }

    Define graph $G \defn (U \sqcup V, E)$.

    \BlankLine
    \tcp{Step 5: Perfect Matching between Summands}
    Run the \textsc{Hopcroft-Karp Algorithm} on $G$ to obtain a maximum bipartite matching $M \subseteq E$.\;
    \If{$\abs{M} \neq m$}{
        \Return $\bot$.\;
    }
    \tcp{Find the isomorphism corresponding to the perfect matching.}
    Let $\sigma \vcentcolon [n] \to [n]$.\;
    \tcp{Define isomorphism from $\code_1$ to $\code_2$ piecewise}
    \For{$(i,j) \in M$}{
        $\psi_{i,j} \gets equivmaps[(a_i,b_j)]$. \tcp{get edge label}
        $\sigma\mid_{X_i} \defn \psi_{i,j}$. \tcp{image of $i$-th set of the partition of $[n]$}
    }
    \Return{$\pi=\pi_2^{-1}\sigma \pi_1$}. \tcp{compose isomorphism with equivalence maps}
\end{algorithm}

\begin{proof}
    Let $\code_1, \code_2 \subseteq \F_q^n$ be the two input codes.
    \smallskip

    \emph{(Correctness)}
    To prove that \cref{alg:search-PCE->AGEN} is correct, we must show that if $\code_1 \cong \code_2$, the algorithm outputs an isomorphism $\pi$ between the two codes, and outputs $\bot$ otherwise.
    
    In Step 1, \cref{alg:decompose} is applied to the input generator matrices to decompose the codes into direct sums of indecomposable codes.
    This yields 
    % \begin{equation*}
    %     \pi_1(\code_1) = A_1 \oplus \dots \oplus A_{m_1} \quad
    %     \text{ and } \quad \pi_2(\code_2)= B_1 \oplus \dots \oplus B_{m_2},
    % \end{equation*}
    $\pi_1(\code_1) = A_1 \oplus \dots \oplus A_{m_1}$ and $\pi_2(\code_2)= B_1 \oplus \dots \oplus B_{m_2}$
    for some  $m_1,m_2 \leq n$ and indecomposable codes $\set{A_i}_{i=1}^{m_1}$ and $\set{B_j}_{j=1}^{m_2}$.
    By \cref{lem:Slepian}, the codes can only be equivalent if they have the same number of summands and $m_1 = m_2$.
    If this is the case, we denote the corresponding index sets by $\set{X_i}_{i=1}^m$ and $\set{Y_j}_{j=1}^m$.
    
    In Step 2, the algorithm constructs the direct sum code $\code = \pi_1(\code_1) \oplus \pi_2(\code_2)$.
    Let $Z_1 \defn [n]$ and $Z_2 \defn \set{n+1,\ldots,2n}$ denote the index sets of $\pi_1(\code_1)$ and $\pi_2(\code_2)$ in $\code$, respectively.
    These form a partition $[2n] = Z_1 \sqcup Z_2$.
    The \AGEN oracle is then called once on $\code$ to obtain a set $S$ of automorphisms that generate the group $\Aut(\code)$.
    By definition, any automorphism of $\code$ can be written as a composition of generators in $S$.

    In Step 3, we construct a directed graph $G'$ that represents all equivalences between the indecomposable summands indicated by the generating set $S$.
    The summands $\set{A_i}_{i=1}^m$ and $\set{B_j}_{j=1}^{m}$ are represented by the vertices $\set{a_1,\ldots,a_m}$ and $\set{b_1,\ldots,b_m}$, respectively.
    The algorithm then iterates through all $s$ generators in $S$. 
    For each generator $\phi \in S$, it finds the image of each summand $A_i$ under $\phi$, given by the restriction $\phi\mid_{X_i}(A_i)$.
    By \cref{cor:no_partial-Aut}, if $\set{A_i}_{i=1}^m$ and $\set{B_j}_{j=1}^{m}$ are decompositions of equivalent codes, then every $A_i$ must map either to some $A_j$ or $B_j$ under $\phi$.
    A directed edge is added to $G'$ between the vertex for $A_i$ and its image:
    If $\phi\mid_{X_i}(A_i) = B_j$ for some $j\in [m]$, then the edge $(a_i,b_j)$ is added. 
    Otherwise if $\phi\mid_{X_i}(A_i) = A_j$, edge $(a_i,a_j)$ is added; in the event that $A_i$ maps to itself, no self-loop edge is added so that the final graph $G'$ is a simple graph.
    The algorithm keeps track of the automorphisms in $S$ that maps each summand to another, using a dictionary object \emph{maps} which functions as a set of labels for each edge.
    While there may be multiple automorphisms that map $A_i$ to $B_j$ (or $A_j$ for $j\neq i$), the algorithm updates the corresponding dictionary entry to keep only the most recently encountered one.
    By construction, the final graph $G'$ contains edges only between equivalent summands and does not include self-loops or multiple edges between vertices.

    In Step 4, the algorithm constructs a bipartite graph $G$ using the equivalence graph $G'$.
    For each summand $A_i$, it runs Dijkstra's algorithm to find a shortest path between the vertex $a_i$ and every other node in its connected component, $G_{a_i}$.
    For every vertex $b_j$ in $V$ that is reachable from $a_i$, the algorithm adds an edge $(a_i,b_j)$ to the new graph $G$.
    Then, it traverses the path $P_{b_j}$ from $a_i$ to $b_j$, retrieves the automorphism label for each edge in this path, and composes the sequence of automorphisms to obtain an automorphism that maps $A_i$ to $B_j$.
    For example, if $\set{e_1,\ldots,e_r}$ is a path from $a_i$ to $b_j$, the corresponding sequence of generators $\phi_{e_1},\ldots, \phi_{e_r}$ is composed into $\psi_{i,j} \defn \parens{\phi_{e_r} \circ \ldots \circ \phi_{e_2} \circ \phi_{e_1}}\mid_{X_i}$.
    By \cref{cor:no_partial-Aut}, the images of the summands are disjoint, so this composition is a well-defined bijection that satisfies $\psi_{i,j}(A_i) = B_j$.
    By construction, the resulting graph $G$ is bipartite on vertex sets $U$ and $V$, and contains one edge for every equivalent pair of summands belonging to different codes.

    Suppose that $\code_1 \cong \code_2$.
    In this case, we have $\code_1 \cong A_1 \oplus \dots \oplus A_m \cong B_1 \oplus \dots \oplus B_m \cong \code_2$. 
    By \cref{lem:Slepian}, there exists an isomorphism map $\pi \vcentcolon \pi_1(\code_1) \to \pi_2(\code_2)$ that reorders the summands $\set{A_i}_{i=1}^m$ and maps each one to an equivalent summand in $\set{B_j}_{j=1}^{m}$, i.e. $\pi\mid_{X_i} (A_i) = B_{\sigma_\pi(i)}$ for all $i\in [m]$ for some permutation $\sigma_\pi \vcentcolon [n] \to [n]$ induced by $\pi$.
    By construction of $\code$, the group $\Aut(\code)$ must contain 
    % an automorphism that maps $\pi_1(\code_1)$ to $\pi_2(\code_2)$ in this way.
    % More specifically, there exists 
    some $\pi' \in \Aut(\code)$ such that $\pi'\mid_{Z_1}(\pi_1(\code_1)) = \pi_2(\code_2) = \pi(\pi_1(\code_1))$ 
    (and symmetrically, $\pi'\mid_{Z_2}(\pi_2(\code_2)) = \pi_1(\code_1) = \pi^{-1}(\pi_2(\code_2))$).
    We show that the algorithm outputs $\pi$.

    In Step 5, the algorithm explores the bipartite graph $G$ and runs the Hopcroft-Karp algorithm from \cite{HK73} to find a maximum bipartite matching $M$ between the vertices $a_1,\ldots,a_m$ and $b_1,\ldots,b_m$.
    If the matching is not perfect and has fewer than $m$ edges, then no isomorphism exists from $\set{a_i}_{i=1}^m$ to $\set{b_j}_{j=1}^{m}$.
    If $\abs{M} = m$, then we have a bijective map between $\set{A_i}_{i=1}^m$ and $\set{B_j}_{j=1}^{m}$.
    The automorphism labels for the edges in the matching $M$ are then pieced together to obtain a isomorphism $\pi$ from from $\pi_1(\code_1)$ to $\pi_2(\code_2)$.
    By construction, the automorphisms for the edges in $G$, and hence $M$, have disjoint supports, so their composition is a well-defined bijection that maps $\pi(\pi_1(\code_1))=\pi_2(\code_2)$.
    Finally, composing this isomorphism with the equivalence maps yields an isomorphism $\sigma= \pi_2^{-1}\pi\pi_1$ that satisfies
    \begin{equation*}
        \sigma(\code_1)
        = \pi_2^{-1}\pi\pi_1(\code_1)
        = \pi_2^{-1}(\pi_2(\code_2))
        = \code_2.
    \end{equation*}

    \emph{(Complexity)}
    We analyze the runtime for each step.
    Step 1 runs \cref{alg:decompose} on both input codes, which requires $2\cdot O(n^3 \log(q)) = O(n^3 \log(q))$ time.
    In Step 2, the direct sum code $\code = \code_1 \oplus \code_2$ can be constructed from the given generator matrices in $O(n \log(q))$ time.
    A single oracle call to \AGEN on $\code$ yields a generating set $S$ for its automorphism group of size $s \defn \abs{S}$.
    The equivalence graph $G'$ constructed in Step 3 requires iterating over all $s$ generators in $S$ and every pair of summands in the sets $\set{A_i}_{i=1}^m$ and $\set{B_j}_{j=1}^{m}$.
    This requires a total of $O(s\cdot m^2 \log(q))$ time
    In Step 4, the algorithm runs Dijkstra's algorithm starting from $a_i$ for each of the $m$ vertices in $U$.
    In the worst case, when all the summands are equivalent to each other, $G'$ will be a connected graph; this requires every vertex to be explored in each iteration.
    Hence, finding shortest paths from $a_i$ to every other vertex in its connected component $G_{a_i}$ takes at most $\Theta(\abs{E'} + \abs{U\cup V}^2) = \Theta(m(2m-1) + (2m)^2) = O(m^2)$ time.
    After exploring the connected component $G_{a_i}$, the algorithm iterates through each vertex $b_j$ in $G_{a_i} \cap V$, of which there are at most $m$.
    Then it traverses the path between $a_i$ and all such $b_j$, which contains at most $m(2m-1)$ edges.
    So, for all $m$ vertices in $U$, this step requires at most $m\cdot \parens{O(m^2) + O(m\cdot m(2m-1))} = O(m^4)$ time.
    Finally, in Step 5, running the Hopcroft-Karp algorithm on the graph $G$ with $\abs{V} = 2m \leq 2n$ vertices and $\abs{E} \leq 2m^2 \leq 2n^2$ edges requires at most $O(\abs{E}\cdot \sqrt{\abs{V}}) = O(m^{2.5}) = O(n^{2.5})$ time.
    In the case where this outputs a perfect matching $M$ of size $m$, iterating through each of the $m$ edges and composing the automorphisms takes $O(m)$ time.
    
    Therefore, the total asymptotic runtime of \cref{alg:search-PCE->AGEN} is
    $O(n^3 \log(q)) + O(n \log(q)) + O(s\cdot m^2 \log(q)) + O(m^4) + O(m^{2.5}) 
    = O(s \cdot n^2 \log(q) + n^4)$.
\end{proof}

% ===========================================
\subsection{Search-\LCE Reduces to \MAGEN}
\label{subsec:sLCE->MAGEN}

Now we prove the monomial analog of \cref{prop:alg-sPCE->AGEN}.
The reduction algorithm is nearly identical to \cref{alg:search-PCE->AGEN} but uses the properties of monomial equivalence and a \MAGEN oracle instead of \AGEN. 
We include the algorithm and its proof here for completeness.

\begin{prop}
    \label{prop:alg-sLCE->MAGEN}
    \Cref{alg:search-LCE->MAGEN} correctly solves \sLCE for any pair of $n$-length codes in $O(s\cdot n^2 \log(q) + n^4)$ time using a single call to an \MAGEN oracle,
    where $s$ is the size of the generating set found by the oracle.
\end{prop}

\begin{proof}
    The proof is structurally identical to that of \cref{prop:alg-sPCE->AGEN}. 
    By running the \MAGEN oracle on the direct sum code $\code \defn \pi_1(\code_1) \oplus \pi_2(\code_2)$, we obtain a generating set for $\MAut(\code)$ that contains the linear isometries mapping the summands of $\code_1$ to $\code_2$.
    Because the monomial automorphism group action preserves the direct sum decomposition, as in \cref{alg:search-PCE->AGEN}, the equivalence graph construction and shortest-path search produce a mapping between the indecomposable summands (if it exists). 
    Piecewise composition of the maps corresponding to the edges of the shortest path in the graph yields the final linear isometry between $\code_1$ to $\code_2$. 

    A single oracle call is made to the \MAGEN oracle.
    For the time complexity, finding the shortest paths and extracting the maximum bipartite matching requires the same time as for \cref{alg:search-PCE->AGEN}; 
    this is dominated by the $O(s\cdot n^2 \log(q) + n^4)$ time required to process the $s$ monomial generators.
\end{proof}

\begin{algorithm}[H]
    \DontPrintSemicolon
    \SetKwInOut{Input}{Input}
    \SetKwInOut{Output}{Output}

    \caption{Reduction from \sLCE to \MAGEN}
    \label{alg:search-LCE->MAGEN}
    
    \Input{Generator matrices $\mG_1, \mG_2 \in \F_q^{k \times n}$ for codes $\code_1, \code_2 \subseteq \F_q^n$.}
    \Output{A linear isometry $\pi \in (\F_q^*)^n \rtimes \cS_n$ such that $\pi(\code_1) = \code_2$.}
    \hrulefill
    
    \tcp{Step 1: Code Decomposition}
    Run \cref{alg:decompose} on $\mG_1, \mG_2$ to obtain decompositions $\pi_1(\code_1)= \bigoplus_{i=1}^{m_1} A_i$ and $\pi_2(\code_2)= \bigoplus_{j=1}^{m_2} B_j$,
    for some indecomposable codes $\set{A_i}_{i=1}^{m_1}$ and $\set{B_j}_{j=1}^{m_2}$ and permutations $\pi_1,\pi_2 \in \cS_n$.\;
    
    \If{$m_1 \neq m_2$}{
        \Return $\bot$.
    }
    Let $m \defn m_1 = m_2$.\;
    Denote the index sets of $\code_1$ and $\code_2$ by $\set{X_i}_{i=1}^m$ and $\set{Y_j}_{j=1}^m$.

    \BlankLine
    \tcp{Step 2: \MAGEN Oracle Call}
    Construct the direct sum code $\code \defn \pi_1(\code_1) \oplus \pi_2(\code_2$).\;
    % \BlankLine
    Run \MAGEN on $\code$ to obtain a generating set $S = \set{\phi_1,\ldots, \phi_s}$ of $\MAut(\code)$.\;

    \BlankLine
    \tcp{Step 3: Construction of the Equivalence Graph}
    Define the vertex sets $U \defn \set{a_1,\ldots,a_m}$ and $V \defn \set{b_1,\ldots,b_m}$.\;
    Set $E' \defn \emptyset$.\;
    Set $maps \defn \set{}$. \tcp{dictionary of equivalence maps}
    
    \emph{(continued on next page)}   
\end{algorithm}
% manual page break
\begin{algorithm}[H]
    \DontPrintSemicolon
    \emph{(continued from previous page)}  
    
    \For{$\phi \in S$}{
        \For{$i \gets 1$ \KwTo $m$}{
            $Im(A_i) \defn \phi\mid_{X_i}(A_i)$.\;
            $Im(B_i) \defn \phi\mid_{Y_i}(B_i)$.\;

            % \BlankLine
            \tcp{Find where $a_i$ maps to.}
            \For{$j \gets 1$ \KwTo $m$}{
                \If{$Im(A_i) = B_j$}{
                    $E' \gets E' \cup \set{(a_i,b_j)}$.\;
                    $maps[(a_i,b_j)] \defn \phi$. \tcp{update edge label}
                }
                \ElseIf{$Im(A_i) = A_j$ and $j\neq i$}{
                    $E' \gets E' \cup \set{(a_i,a_j)}$.\;
                    $maps[(a_i,a_j)] \defn \phi$. \tcp{update edge label}
                }
            }
            \tcp{Find where $b_i$ maps to.}
            \For{$j \gets 1$ \KwTo $m$}{
                \If{$Im(B_i) = A_j$}{
                    $E' \gets E' \cup \set{(b_i,a_j)}$.\;
                    $maps[(b_i,a_j)] \defn \phi$. \tcp{update edge label}
                }
                \ElseIf{$Im(B_i) = B_j$ and $j\neq i$}{
                    $E' \gets E' \cup \set{(b_i,b_j)}$.\;
                    $maps[(b_i,b_j)] \defn \phi$. \tcp{update edge label}
                }
            }
        }
    }

    Define graph $G' \defn (U \cup V, E')$.\; 
    \BlankLine
    \tcp{Step 4: Construction of the Bipartite Graph}
     Set $E \defn \emptyset$.\;
     Set $equivmaps \defn \set{}$. \tcp{edge labels}
     \For{$i \gets 1$ to $m$}{
        \BlankLine
        Let $G_{a_i}$ denote the connected component in $G'$ that contains $a_i$.
        
        \BlankLine
        Run \textsc{Dijkstra's Algorithm} on $G'$ from root $a_i$ to obtain paths $\set{P_v \subseteq E' \vcentcolon v\in G_{a_i}}$ to all vertices in its connected component $G_{a_i}$.

        \BlankLine
        \For{$v \in G_{a_i} \cap V$}{
            $j \defn V.\text{index}(v)$. \tcp{$v = b_j$ for some index $j \in [m]$.}
            $\psi_{i,j} \defn (\vecone, \id)$. \tcp{global monomial automorphism}
            \For{$e \in P_v$}{
                $\phi_e \defn maps[e]$. \tcp{get edge label}
                $\psi_{i,j} \gets \phi_e \circ \psi_{i,j}$. \tcp{compose monomial automorphisms}
            }
            $E \gets E \cup \set{(a_i, b_j)}$.\;
            $\psi_{i,j} \gets \psi_{i,j}\mid_{X_i}$. \tcp{restrict to support of $A_i$}
            $equivmaps[(a_i, b_j)] \defn \psi_{i,j}$. \tcp{update edge label}
        }
    }
    Define graph $G \defn (U \sqcup V, E)$.

    % \BlankLine
    \tcp{Step 5: Perfect Matching between Summands}
    Run the \textsc{Hopcroft-Karp Algorithm} on $G$ to obtain a maximum bipartite matching $M \subseteq E$.\;
    \If{$\abs{M} \neq m$}{
        \Return $\bot$.\;
    }
    \tcp{Find the linear isometry corresponding to the perfect matching.}
    Let $\sigma \defn (\vecone, \id)$.\tcp{initialize map to be defined}

    % \BlankLine 
    \tcp{Define the linear isometry from $\code_1$ to $\code_2$ piecewise}
    \For{$(i,j) \in M$}{
        $\psi_{i,j} \gets equivmaps[(a_i,b_j)]$. \tcp{get edge label}
        $\sigma\mid_{X_i} \defn \psi_{i,j}$. \tcp{image of $i$-th set of the partition of $[n]$}
    }
    $\pi \defn (\vecone,\pi_2^{-1}) \circ \sigma \circ (\vecone,\pi_1)$.\tcp{compose linear isometry with permutation isomorphisms}
    \Return{$\pi$}. 
\end{algorithm}

% ===========================================
\subsection{\AGEN Reduces to \PCE}
\label{subsec:AGEN->PCE}

We now prove the reverse reduction, from \AGEN to \PCE.
Note that \AGEN reduces to \sPCE as a consequence, since there is a trivial reduction from \PCE to \sPCE.
The reduction algorithm relies on two oracles, for solving \Apart and \sPCE; these oracles can be replaced with \cref{alg:APART->PCE} and the algorithm from \cref{lem:search-to-decision-PCE}, which rely on a \PCE oracle.

For any input code $\code$ with generator matrix $\mG$, let $G \defn \Aut(\code)$ denote its permutation automorphism group and let $G = G^{(E_0)} \geq G^{(E_1)} \geq \dots \geq G^{(E_s)}$ be the setwise stabilizer chain given by $\mG$.
The algorithm finds a generating set for $G$ by finding a generating set for the smallest subgroup $G^{(E_s)}$ in the chain, together with a set of left coset representatives for the quotient space given by each pair of consecutive subgroups along the chain; this is given by \cref{lem:setwise-gen}.
The subgroup $G^{(E_s)}$ stabilizes every projective class $E_1,\ldots,E_s$ setwise, and so contains all the automorphisms that swap indices belonging to the same class.
This group is generated by all possible transpositions for every projective class.
To find left coset representatives for $G^{(E_{i-1})}/G^{(E_i)}$ for every $1\leq i \leq s$, the algorithm searches for automorphisms that maps indices from $E_i$ to some $E_j$ for every $j$ in the orbit of $i$, under the condition that indices in $E_1,\ldots,E_{i-1}$ are stabilized within their respective classes.
To test whether such an automorphism exists, the algorithm invokes the \PCE on two carefully-constructed test matrices which guarantee that any automorphism between them belongs to $G^{(E_{i-1})}/G^{(E_i)}$ and maps $i$ to $j$.

\begin{prop}
    \label{prop:alg-AGEN->PCE}
    \Cref{alg:AGEN->PCE} correctly solves \AGEN for any $n$-length $q$-ary code in $O(n^6 \log(q))$ time using $1$ call to an \Apart oracle, $O(n^2)$ queries to a \sPCE oracle, and $O(n^2)$ additional calls to a \PCE oracle.
\end{prop}

As an immediate consequence of \cref{thm:APART/MAPART->PCE/LCE,lem:search-to-decision-PCE}, we obtain a reduction from \AGEN to \PCE that makes $O(n^4)$ oracle calls, as claimed in \cref{thm:AGEN/MAGEN->PCE/LCE}.
\vspace{5pt}

\begin{algorithm}[H]
    \DontPrintSemicolon
    \SetKwInOut{Input}{Input}
    \SetKwInOut{Output}{Output}
    \SetKwFunction{SearchPCE}{Search-PCE}
    % \SetKwFunction{APART}{APART}

    \caption{Reduction from AGEN to PCE}
    \label{alg:AGEN->PCE}
    
    \Input{A generator matrix $\mG \in \F_q^{k \times n}$ of a linear code $\code \subseteq \F_q^n$.}
    \Output{A generating set $S \subseteq \cS_n$ for the automorphism group $\Aut(\code)$.}
    \hrulefill
    
    \BlankLine
    \tcp{Step 1: Partition index set into column classes and construct coord set $B'$.}
    Let $\set{\vecg_1,\ldots,\vecg_n}$ be the set of column vectors of $\mG$.\;
    $\cE \defn \emptyset$. \tcp{collection of redundancy sets}
    $seen \defn \emptyset$. \tcp{keep track of unique columns}
    $B' \defn \emptyset$. \tcp{indices of first appearances}
    
    \For{$i \gets 1$ \KwTo $n$}{
        \If{$\vecg_i \in seen$}{
            \For{$E\in \cE$}{
                $\vecv \gets E[1]$.\;
                \If{$\vecg_i = \vecv$}{
                    $E \gets E \cup \set{\vecg_i}$.\;
                }
            }
        }
        \Else{
            $E \defn \set{\vecg_i}$.\;
            $seen \gets seen \cup \set{\vecg_i}$.\;
            $\cE \gets \cE \cup \set{E}$.\;           
            $B' \gets B' \cup \set{i}$. \tcp{record coord of new unique column}
        }
    }
    $m \gets 0$.\;
    $S \gets \emptyset$. \tcp{initialize generating set}
    
    \emph{(continued on next page)}   
\end{algorithm}
% manual page break
\begin{algorithm}[H]
    \DontPrintSemicolon
    \emph{(continued from previous page)}  
     
     % \BlankLine
    \tcp{Step 2: Generate the smallest subgroup $G^{(E_s)}$ using transpositions.}
    \For{$E \in \cE$}{
        $\vece_1 \gets E[1]$.\;
        \For{$\vece_k \in E$ \emph{\textbf{where}} $\vece_k \neq \vece_1$}{
            $S \gets S \cup \set{(\vece_1, \vece_k)}$.\;
        }
        \tcp{Find maximum column multiplicity.}
        \If{$\abs{E} > m$}{
            $m \gets \abs{E}$.\;
        }
    }    
    \BlankLine
    \tcp{Step 3: Find orbit partition using the \Apart oracle.}
    Run $\APART(\mG)$ to obtain coordinate orbits $X_1,\ldots,X_\ell$.\;

    \BlankLine
    $\cG \gets \emptyset$. \tcp{track permanently fixed columns in subchain}
    $\Gpin \gets \mG$. \tcp{initialize global base matrix}
    
    \tcp{Step 4: Compute all coset transversals $\sigma \in R_i$ along setwise stabilizer chain.}
    \For{$i \in B'$}{
        $X \defn$ unique orbit in $\set{X_1,\ldots,X_\ell}$ containing $i$.\;
        
        \BlankLine
        \tcp{Find a mapping from $E_i$ to each valid target class.}
        \For{$j \in X \cap B'$ \emph{\textbf{where}} $j > i$}{
            $l \gets \abs{\cG}$.\;
            $\mA \defn (l+1)m$ copies of $\vecg_i$.\;
            $\mB \defn (l+1)m$ copies of $\vecg_j$.\;
            $\mG_A \gets \bracks{\Gpin \mid \mA}$.\;
            $\mG_B \gets \bracks{\Gpin \mid \mB}$.\;
            
           \BlankLine
                \tcp{Check if there exists an automorphism that maps $i$ to $j$.}
                
                \If{$\PCE(\mG_A, \mG_B) \to $ \emph{\textsc{Yes}}}{
                    \BlankLine
                    \tcp{If it exists, find it via the \SearchPCE oracle.}
                    Run $\SearchPCE(\mG_A, \mG_B)$ to obtain a map $\pi$.
                    
                    $U \gets \set{1, \dots, n}$. \tcp{keep track of mapped indices}
                    $\sigma \defn \id$. \tcp{initialize generator}

                    \For{$c \gets 1$ \KwTo $n$}{
                        \eIf{$\pi(c) \le n$}{
                            $t \defn \pi(c)$\;
                        }{
                            $t \defn$ coord of the first column in $\mG_B$ identical to column $\pi(c)$.
                           
                        }
                        $t' \defn \min\set{ u \in U \mid \vecg_u = \vecg_t}$\;
                        $\sigma(c) \gets t'$\;
                        $U \gets U \setminus \set{t'}$\;
                    }
                $S \gets S \cup \set{\sigma}$. \tcp{add transversal representative}
            }
        }
        
        \BlankLine
        $\mA \defn (l+1)m$ copies of $\vecg_i$.\;
        $\Gpin \gets \bracks{\Gpin \mid \mA}$. \tcp{fix base point for next iteration}
        $\cG \gets \cG \cup \set{\vecg_i}$.\;
    }
    
    \BlankLine
    \Return{$S$}.
\end{algorithm}

\begin{proof}
    Let $\mG \in \F_q^{k \times n}$ be the input generator matrix of the code $\code$, and $G = \Aut(\code)$ denote its automorphism group. 
    Let $\set{\vecg_1,\ldots,\vecg_n}$ denote the columns of $\mG$ in order of appearance.
    We prove that \cref{alg:ACOUNT->PCE} correctly outputs a generating set $S$ for $G$ by traversing a setwise stabilizer chain for $G$ and applying \cref{lem:setwise-gen}.
    \smallskip
    
    \emph{(Correctness)}
    The algorithm first partitions the columns of $\mG$ into disjoint redundancy classes $\cE = \set{E_1, \dots, E_s}$, where each $E_i$ contains identical columns of a particular type, and no two classes contain the same vector.
    For this sorting process, the algorithm identifies the coordinates of the unique column vectors in $\mG$ in order of first appearance; let $B' = \set{t_1,\ldots,t_s} \subseteq [n]$ denote the elements of this subset.
    Note that each coordinate $t_i \in B'$ defines the redundancy class $E_i \in \cE$.
    By \cref{def:setwise-stab-chain}, this $B'$ defines a setwise stabilizer chain $G \geq G^{(E_1)} \geq \dots \geq G^{(E_s)}$.
    By definition, the smallest subgroup $G^{(E_s)}$ contains all permutations that trivially swaps columns in each redundancy class.
    % ; this is isomorphic to the direct product $\Sym(E_1) \times \dots \times \Sym(E_s)$.
    Then, every automorphism in $G^{(E_s)}$ can be expressed as a composition of transpositions, so a generating set for $G^{(E_s)}$ consists of all 2-cycle permutations for each redundancy class.
    The algorithm finds and adds all these transpositions to the final generating set $S$, and so it includes a generating subset $S_{E_s}$ for $G^{(E_s)}$.
    The maximum column multiplicity $m$ in $\mG$ is then computed by finding the cardinality of the largest redundancy class.

    Next, the algorithm makes a single oracle call to an \Apart oracle to obtain a partition of the index set $[n]$ into orbits under automorphisms in $G$.
    While not strictly necessary, this step provides an optimization that restricts the search space for automorphisms in the following step.

    To find a generating set for the remaining automorphisms in $G$, the algorithm traverses the chain $G \geq G^{(E_1)} \geq \dots \geq G^{(E_s)}$ in descending order by iterating through the elements of $B'$.
    In the $t_i$-th iteration, we find the coset representatives $R_i$ for the quotient space $G^{(E_{i-1})}/G^{(E_i)}$ by setwise stabilizing the columns in $E_{i-1}$ and searching for automorphisms that map $E_i$ to some $E_j$ for all larger indices $t_j \in B'$ in the orbit $X$ of $t_i$.
    By testing only the unique representatives of each redundancy class, the algorithm evaluates exactly one target per class, avoiding redundant oracle queries.

    Any automorphism that maps $E_i$ to $E_j$ maps all copies of column $\vecg_{t_i}$ to all copies of column $\vecg_{t_j}$, where $t_i,t_j\in B'$.
    To test for the existence of such automorphisms, we use a similar technique to that in \cite{BM23,CSV25} and augment the matrix $\mG$ with copies of columns to ensure unique frequency.
    In particular, we construct test matrices $\mG_A$ and $\mG_B$ by appending blocks of $(l+1)m$ copies of $\vecg_{t_i}$ and $(l+1)m$ copies of $\vecg_{t_j}$, respectively, to the current stabilized matrix in which column classes $E_1,\ldots,E_{i-1}$ are fixed setwise.
    By definition of $m$ and $l$, no other column in $\mG_A$ and $\mG_B$ appears the same number of times, so $\vecg_{t_i}$ and $\vecg_{t_j}$ have the same unique frequency.
    By \cref{lem:col-frequency}, the \PCE oracle can only respond \textsc{Yes} if there is an automorphism that maps $E_i$ to $E_j$,
    Thus, the oracle is forced to decide equivalence exclusively within the stabilizer subgroup $G^{(E_{i-1})}$.

    If an automorphism that maps $E_i$ to $E_j$ exists, then the algorithm calls the \sPCE oracle to find it.
    It outputs a map $\pi$ that permutes the columns of the current fixed matrix $\mG_A$, which has at least $n + (l+1)m \geq n+m$ columns in any given iteration.
    We derive a permutation $\sigma \in \cS_n$ on the original index set by extracting the images $\pi(i)$ for all $i \in [n]$.
    Any coordinate $c$ that maps into the desired range $[n]$ is preserved; if $c$ is otherwise mapped to a larger index $\pi(c) > n$, the algorithm identifies the redundancy class of the target $\pi(c)$ and finds the coordinate in $[n]$ of the corresponding identical column.
    To ensure that no collisions are created among the image elements of $\sigma$, the algorithm tracks all unassigned image values with the set $U \subseteq [n]$ and assigns $\sigma(c)$ to be the smallest available coordinate $t' \in U$.
    % \anote{What does this mean?}
    % Since structural equivalence strictly preserves the multiplicities of all column classes, the pool $U$ is guaranteed to contain an available target of the correct class for every $c \in [n]$, ensuring $\sigma$ is a perfect bijection on the original coordinates. 
    This ensures that $\sigma$ is a bijection on the index set $[n]$.
    Since the algorithm finds an automorphism for every valid pair $E_i$ and $E_j$, it finds all coset representatives $R_i$ for the  quotient space $G^{(E_{i-1})}/G^{(E_i)}$, for all $i\in [s]$.
    Thus, the final generating set $S$ contains $R_1,\ldots,R_s$ and a generating set $S_{E_s}$ for $G^{(E_s)}$.
    Therefore, by \cref{lem:setwise-gen}, this $S$ is a strong generating set for the entire automorphism group $G$.
    \smallskip

    \emph{(Complexity)}
    We first bound the number of oracle calls. 
    The algorithm makes a single query to the \Apart oracle to obtain the orbit partition of the index set $[n]$.
    % By \cref{thm:APART/MAPART->PCE/LCE}, \Apart can be solved using $O(n^2)$ queries to a \PCE oracle. 
    % 
    In the main for-loop, a target $t_j \in B'$ is tested against a base coordinate $t_i \in B'$ only if both belong to the same orbit $X$. 
    For any orbit $X$, the number of distinct column classes is bounded by $|X|$. 
    Thus, the outer for-loop evaluates the test branch at most $|X|$ times for this orbit, and each inner for-loop tests at most $|X|-1$ targets. 
    This limits the number of \PCE oracle queries to $O(|X|^2)$ for each $X$.
    Because the coordinate orbits form a partition of $[n]$, the sum of $O(|X|^2)$ across all orbits $X$ is bounded above by $O(n^2)$. 
    Therefore, the algorithm makes a total of % 2\cdot O(n^2)
    $O(n^2)$ calls to the \PCE oracle.
    The \SearchPCE oracle is invoked only when the existence of an automorphism is confirmed by the \PCE oracle, and so is invoked at most $O(n^2)$ times across the entire algorithm.
    % By \cref{lem:search-to-decision-PCE}, each \SearchPCE query requires $O(n^2)$ calls to a \PCE oracle.
    % Therefore, the algorithm makes a total of $O(n^4)$ calls to the \PCE oracle.
    
    Now we bound the runtime.
    Partitioning the index set $[n]$ into redundancy classes and constructing the set $B'$ requires comparing column vectors at most $n^2$ times, so this step requires $O(n^2 \log(n))$ time.
    By \cref{prop:alg-APART->PCE}, running the \Apart oracle requires $O(n^3 \log(q))$ time.
    By \cref{lem:search-to-decision-PCE}, running the \sPCE oracle requires $O(n^3 \log(q))$ time.
    Now we evaluate the cost of constructing the test matrices, which dominates the time complexity of the last step and the entire algorithm.
    The maximum column multiplicity in $\mG$ is $m \leq n$, and in each iteration $t_i \in B'$, the number of unique column classes fixed so far is $l \leq n$. 
    The matrix $\Gpin$ of setwise stabilized columns accumulates $\sum_{a=1}^l a\cdot m$ appended columns in each iteration.
    To construct $\mG_A$ and $\mG_B$, an additional block of $(l+1)m$ columns is appended to $\Gpin$.
    Thus, the number of columns appended to $\mG$ is bounded above by $n + m l(l+1)/2 + (l+1)m = O(n^3)$. 

    Each field element in $\F_q$ requires $\log(q)$ bits to represent and $O(\log(q))$ bit operations to read, copy, or compare. 
    The initial partitioning phase and computing the size of $G^{(E_s)}$ requires $O(n^2\log(q))$ and $O(n)$ operations respectively.
    Then, for each of the $O(n^2)$ oracle queries, copying at most $O(n^3)$ column vectors for the augmented matrices requires $O(k n^3\log(q))$ operations.
    If an automorphism $\pi$ is found by the oracle, constructing the permutation $\sigma$ requires $O(n)$ local comparisons.
    Therefore, since $k \leq n$, the total runtime is bounded by $O(n^2) \cdot O(kn^3\log(q)) = O(n^6\log(q))$ time. 
\end{proof}

% ===========================================
\subsection{\MAGEN Reduces to \LCE}
\label{subsec:MAGEN->LCE}

Now we prove the monomial analog of \cref{prop:alg-AGEN->PCE}.
The reduction algorithm is nearly identical to \cref{alg:AGEN->PCE} but uses the properties of monomial equivalence and an \LCE oracle instead of \PCE. 
For completeness, we include the full algorithm and proof below.

\begin{prop}
    \label{prop:alg-MAGEN->LCE}
    \Cref{alg:MAGEN->LCE} correctly solves \MAGEN for any $n$-length $q$-ary code in $O(n^6 \log(q))$ time using one call to an \MApart oracle, $O(n^2)$ queries to a \sLCE oracle, and $O(n^2)$ additional calls to an \LCE oracle.
\end{prop}

\begin{proof}
    The proof is structurally identical to that of \cref{prop:alg-AGEN->PCE}, but using the setwise stabilizer chain framework for the monomial automorphism group.
    Let $M \geq M^{(E_1)} \geq \dots \geq M^{(E_s)}$ be the monomial setwise stabilizer chain for input generator matrix $\mG$, as in \cref{lem:mon-setwise-stab-chain-order}.
    The algorithm here finds monomial generators of the smallest subgroup $M^{(E_s)}$ by independently scaling the $c$ indecomposable summands of $\code(\mG)$ and adjusting the internal transpositions. 
    For each oracle call to \sLCE, the oracle finds a linear isometry for the augmented code that maps a column $i$ to another column $j$.
    The algorithm extracts a linear isometry for the original code $\code(\mG)$ with index set $[n]$ by remapping any out-of-bounds image coordinates as before, but now it multiplies the associated scalar by the corresponding column ratio $\lambda$. 

    The number of oracle calls is identical to that of \cref{alg:AGEN->PCE}: the \MApart oracle is called once; the number of \LCE oracle calls is bounded by the $O(n^2)$, the number of candidate test coordinate-target pairs; and the number of \sLCE oracle calls is bounded by $O(n^2)$, the number of positive responses from the \PCE oracle.
    For the time complexity, there are additional scalar operations, which take $O(\log(q))$ time in each loop.
    The column ratio $\lambda$ is also computed at most twice for each coordinate in $[n]$; since $\lambda$ can be obtained by dividing two column entries with the same index, this requires at most $O(n\cdot \log(n))$ time.
    Hence, the total runtime remains bounded by $O(n^6 \log(q))$.
\end{proof}

\begin{algorithm}[H]
    \DontPrintSemicolon
    \SetKwInOut{Input}{Input}
    \SetKwInOut{Output}{Output}
    \SetKwFunction{SearchLCE}{Search-LCE}

    \caption{Reduction from MAGEN to LCE}
    \label{alg:MAGEN->LCE}
    
    \Input{A generator matrix $\mG \in \F_q^{k \times n}$ of a linear code $\code \subseteq \F_q^n$.}
    \Output{A generating set $S \subseteq \cS_n$ for the monomial automorphism group $\MAut(\code)$.}
    \hrulefill
    
    \tcp{Step 1: Partition index set into projective redundancy classes.}
    Let $\set{\vecg_1,\ldots,\vecg_n}$ be the set of column vectors of $\mG$.\;
    $\cE \defn \emptyset$. \tcp{collection of redundancy sets}
    $\cG \defn \emptyset$. \tcp{keep track of unique projective column representatives}
    $B' \defn \emptyset$. \tcp{indices of first appearances}

    \For{$i \gets 1$ \KwTo $n$}{
        \tcp{If a column is a scalar multiple of a previously seen column, sort it into the same class.
        Otherwise, create a new class.}
        \eIf{$\exists\ \vecv \in \cG$ \emph{\textbf{such that}} $\vecg_i = \lambda \vecv$ \emph{for some} $\lambda \in \F_q^*$}{
            \For{$E\in \cE$}{
                $\vecv \gets E[1]$.\;
                \If{$\vecg_i = \lambda \vecv$ \emph{for some} $\lambda \in \F_q^*$}{
                    $E \gets E \cup \set{\vecg_i}$.\;
                }
            }
        }{
            $E \defn \set{\vecg_i}$.\;
            $\cG \gets \cG \cup \set{\vecg_i}$\;
            $\cE \gets \cE \cup \set{E}$\;
            $B' \gets B' \cup \set{i}$. \tcp{record index of first appearance}
        }
    }
    $m \gets 0$. \tcp{maximum projective column multiplicity in $\mG$}
    \For{$E \in \cE$}{
        \If{$|E| > m$}{
            $m \gets \abs{E}$\;
        }
    }
    \emph{(continued on next page)}   
\end{algorithm}
% manual page break
\begin{algorithm}[H]
    \DontPrintSemicolon
    \emph{(continued from previous page)} 
    
    % \BlankLine
    \tcp{Step 2: Build smallest subgroup $M^{(E_s)}$ from global scaling and internal swaps.}
    Run \cref{alg:decompose} on $\mG$ to obtain the decomposition $\pi(\code) = \bigoplus_{j=1}^c A_j$, for some indecomposable codes $\set{A_j}_{j=1}^c$ and permutation $\pi \in \cS_n$.
    
    $S \gets \emptyset$. \tcp{initialize generating set}
    $\alpha \defn$ a primitive element of $\F_q^*$.\;
    \For{$j \gets 1$ \KwTo $c$}{
        \tcp{Generator for independent scaling of summand $A_j$}
        $\sigma = (\vecv_\sigma, \pi_\sigma) \defn (\vecone, \id)$.\;
        \For{$u \in \supp(A_j)$}{
            $\vecv_\sigma[u] \gets \alpha$.\;
        }
        $S \gets S \cup \set{\sigma}$.\;
    }
    
    \For{$E \in \cE$}{
        $\vecg_{e_1} \gets E[1]$.\;
        \For{$\vecg_{e_i} \in E$ \emph{\textbf{where}} $\vecg_{e_i} \neq \vecg_{e_1}$}{
            \tcp{Generator for internal swap uniquely determined by scalars}
            $\sigma = (\vecv_\sigma, \pi_\sigma) \defn (\vecone, \id)$.\;
            $\pi_\sigma(e_1) \gets e_i$\;
            $\pi_\sigma(e_i) \gets e_1$\;
            Find $\lambda \in \F_q^*$ such that $\vecg_{e_i} = \lambda \vecg_{e_1}$.\; 
            $\vecv_\sigma[e_1] \gets \lambda$\;
            $\vecv_\sigma[e_i] \gets \lambda^{-1}$\;
            $S \gets S \cup \set{\sigma}$.\;
        }
    }
    \BlankLine
    
    \tcp{Step 3: Find the monomial orbit partition via \MApart oracle.}
    Run $\MAPART(\mG)$ to obtain monomial coordinate orbits $X_1,\ldots,X_\ell$.\;
    
    \BlankLine
    \tcp{Step 4: Compute all coset transversals $\sigma \in R_i$ along the stabilizer chain.}
    $\cG \gets \emptyset$. \tcp{track permanently fixed columns in subchain}
    $\Gpin \gets \mG$. \tcp{initialize global base matrix}

    \emph{(continued on next page)}   
\end{algorithm}
% manual page break
\begin{algorithm}[H]
    \DontPrintSemicolon
    \emph{(continued from previous page)}

    \For{$i \in B'$}{
        $X \defn$ unique orbit in $\set{X_1,\ldots,X_\ell}$ containing $i$.\;
        
        \tcp{Find a mapping from $E_i$ to each valid target class.}
        \For{$j \in X \cap B'$ \emph{\textbf{where}} $j > i$}{
            $l \gets \abs{\cG}$.\;
            $\mA \defn (l+1)m$ copies of $\vecg_i$.\;
            $\mB \defn (l+1)m$ copies of $\vecg_j$.\;
            $\mG_A \gets \bracks{\Gpin \mid \mA}$.\;
            $\mG_B \gets \bracks{\Gpin \mid \mB}$.\;
            
            \BlankLine
            \tcp{Check if there exists an automorphism that maps $i$ to $j$.}
            \If{$\LCE(\mG_A, \mG_B) \to $ \emph{\textsc{Yes}}}{
                \BlankLine
                \tcp{If an automorphism exists, find it via the \SearchLCE oracle.}
                Run $\SearchLCE(\mG_A, \mG_B)$ to obtain a linear isometry $\mu=(\vecv_\mu,\pi_\mu)$.\;
                
                $U \gets \set{1, \dots, n}$. \tcp{keep track of mapped indices}
                $\sigma = (\vecv_\sigma,\pi_\sigma) \defn (\vecone,\id)$. \tcp{initialize generator}

                \BlankLine
                \For{$c \gets 1$ \KwTo $n$}{
                    \tcp{Identify the original column $t$ that the oracle mapped to.}
                    \eIf{$\pi_\mu(c) \le n$}{
                        $t \defn \pi_\mu(c)$\;
                    }{
                        $t \defn$ original coordinate in $[n]$ that the padded column $\pi_\mu(c)$ is a copy of.\;
                    }
                    
                    \tcp{Find an available projectively equivalent target $t'$ in $U$.}
                    $t' \defn \min\set{ u \in U \mid \vecg_u = \lambda \vecg_t \text{ for some } \lambda \in \F_q^*}$\;
                    Find $\lambda \in \F_q^*$ such that $\vecg_{t'} = \lambda \vecg_t$.\;
                    
                    \tcp{Redirect the map and adjust the scalar.}
                    $\pi_\sigma(c) \gets t'$\;
                    $\vecv_\sigma[c] \gets \lambda \cdot \vecv_\mu[c]$\;
                    $U \gets U \setminus \set{t'}$\;
                }
                $S \gets S \cup \set{\sigma}$. \tcp{add transversal representative}
            }
        }
        
        \BlankLine
        $\mA \defn (l+1)m$ copies of $\vecg_i$.\;
        $\Gpin \gets \bracks{\Gpin \mid \mA}$. \tcp{fix base point for next iteration}
        $\cG \gets \cG \cup \set{\vecg_i}$.\;
    }
    
    \BlankLine
    \Return{$S$}.
\end{algorithm}

\section{Conclusion}%
\label{sec:conclusion}

We show that \CE is a polynomially equivalent problem to several computational problems relating to the automorphism group of a code.
These problems ask for the cardinality, an orbit partition, and a generating set of the automorphism group of a given code, and are called \ACOUNT, \APART, and \AGEN, respectively. 
The relationship between these problems provides a way to solve the permutation (\PCE) and monomial (\LCE) variants of \CE, given some information about the automorphism groups of the input codes, and vice versa.
A key component of our reductions is a method of decomposing any linear code into indecomposable codes; this framework provides a way of studying permutation and monomial equivalence, and may be of independent interest to other areas of coding theory.
For future work, it remains to see how our reductions can inform the construction and analysis of cryptosystems built on Code Equivalence.
Given the similarity between \CE and \LIP, it would also be interesting to study analogous problems for lattices.

%=====================================================================
\newpage
\bibliography{references}

\end{document}